\documentclass[12pt]{article}
\usepackage[authoryear,round]{natbib}
\usepackage{graphicx}
\usepackage{latexsym}
\usepackage[left=1.2in,top=1.2in,right=1.2in,bottom=1.2in]{geometry}\usepackage{amsmath}
\usepackage{booktabs}
\usepackage{color}
\usepackage{tikz}
\usepackage{main}
\usepackage{bm}
\usepackage{mathtools}
\usepackage{placeins}
\usepackage{amssymb}
\usepackage{amsmath}
\usepackage{amsthm} 
\usepackage{mathtools}
\usepackage{booktabs}
\usepackage{bbm}
\usepackage{stackrel}
\usepackage[ruled,vlined]{algorithm2e}

\usetikzlibrary{shapes}
\usetikzlibrary{arrows}

\definecolor{darkblue}{rgb}{0,0.4,0.9}
\definecolor{gray10}{rgb}{0.1,0.1,0.1}
\definecolor{gray20}{rgb}{0.2,0.2,0.2}
\definecolor{gray30}{rgb}{0.3,0.3,0.3}
\definecolor{gray40}{rgb}{0.4,0.4,0.4}
\definecolor{gray60}{rgb}{0.6,0.6,0.6}
\definecolor{gray80}{rgb}{0.8,0.8,0.8}
\definecolor{gray90}{rgb}{0.9,0.9,.9}
\definecolor{gray95}{rgb}{0.95,0.95,.95}
\definecolor{gray96}{rgb}{0.96,0.96,.96}
\definecolor{lgreen} {RGB}{180,210,100}
\definecolor{dblue}  {RGB}{20,66,129}
\definecolor{ddblue} {RGB}{11,36,69}
\definecolor{lred}   {RGB}{220,0,0}
\definecolor{nred}   {RGB}{224,0,0}
\definecolor{norange}{RGB}{230,120,20}
\definecolor{nyellow}{RGB}{255,221,0}
\definecolor{ngreen} {RGB}{98,158,31}
\definecolor{dgreen} {RGB}{78,138,21}
\definecolor{nblue}  {RGB}{28,130,185}
\definecolor{jblue}  {RGB}{20,50,100}
\definecolor{nnyellow}{RGB}{235,200,0}
\definecolor{purple}{RGB}{150, 0, 120}
\definecolor{sgGreen} {RGB}{20, 180, 50}
\definecolor{revised}{rgb}{0,0,0.9}
\newtheorem{definition}{Definition}

\newtheorem{theorem}{Theorem}
\newtheorem{assumption}{Assumption}

\newtheorem{lemma}{Lemma}

\newcommand{\openr}{\hbox{${\rm I\kern-.2em R}$}}
\newcommand{\openn}{\hbox{${\rm I\kern-.2em N}$}}

\newcommand{\indep}{\perp \!\!\! \perp}

\title{Adaptive Sequential Surveillance with Network and Temporal Dependence}

\author{\name Ivana Malenica* \email 
imalenica@berkeley.edu \\
       \addr Division of Biostatistics\\
       University of California, Berkeley
       \AND
       \name Jeremy R. Coyle* \email 
       jeremyrcoyle@gmail.com  \\
       \addr Preva Group\\
       WA, USA
       \AND
       \name Mark J. van der Laan \email 
       laan@berkeley.edu \\
       \addr Division of Biostatistics\\
       University of California, Berkeley
       \AND
       \name Maya L. Petersen \email 
       mayaliv@berkeley.edu \\
       \addr Division of Biostatistics\\
       University of California, Berkeley}
       
\begin{document}

\maketitle

\begin{abstract}

Strategic test allocation plays a major role in the control of both emerging and existing pandemics (e.g., COVID-19, HIV). Widespread testing supports effective epidemic control by (1) reducing transmission via identifying cases, and (2) tracking outbreak dynamics to inform targeted interventions. However, infectious disease surveillance presents unique statistical challenges. For instance, the true outcome of interest 
--- one's positive infectious status, is often a latent variable. In addition, 
presence of both network and temporal dependence reduces the data to a single observation. As testing entire populations regularly is neither efficient nor feasible, standard approaches to testing recommend simple \textit{rule-based} testing strategies (e.g., symptom based, contact tracing), without taking into account individual risk. In this work, we study an adaptive sequential design 
which allows for unspecified dependence among individuals and across time. Our causal target parameter is the mean latent outcome we would have obtained after one time-step, if, starting at time $t$ given the observed past, we had carried out a stochastic intervention that maximizes the outcome under a resource constraint. We propose an Online Super Learner for adaptive sequential surveillance that learns the optimal choice of tests strategies over time while adapting to the current state of the outbreak. Relying on a series of working models, the proposed method learn across samples, through time, or both: based on the underlying (unknown) structure in the data. We present an identification result for the latent outcome in terms of the observed data, and demonstrate the superior performance of the proposed strategy in a simulation modeling a residential university environment during the COVID-19 pandemic.

\end{abstract}

\section{Introduction}\label{sec::sec1}

Effective testing strategies play a major role in epidemic control, for both newly emerging and continuing pandemics. This is particularly true for infectious diseases such as COVID-19 and HIV, because asymptomatic transmission is common. Testing strategies have several interrelated effects on epidemic control. First, by identifying cases, testing can directly trigger interventions at the individual level to prevent onward transmission --- isolation in COVID-19 or initiation of effective antiretroviral therapy in HIV (an intervention that both benefits the individual and dramatically reduces onward transmission risk). Second, testing provides key insights on epidemic dynamics --- including information on individual risk in the context of shifting epidemic trajectory, that can be used to guide subsequent testing decisions. 

In realistic epidemic control settings, testing an entire population at regular intervals is unlikely to be most effective or efficient. Instead, the objective is to make the most effective use of a constrained testing capacity for reducing subsequent cases. While the benefit of diagnosing a single case may vary between individuals (due to heterogeneity in individual risk of onward transmission resulting from factors such as contact rates, infectiousness, and network position), numbers of new cases identified provides one common metric for guiding testing deployment. To ensure effective use of limited testing resources, testing would be prioritized to individuals with a high risk of infection. However, in practice, an individual's risk of infection is unknown and must be learned from the data at hand; this is complicated by the fact that i) the outcome of interest (infection status) is latent except for those individuals who are tested, and ii) risk profiles are likely to change over time as the epidemic evolves. 

In the face of these challenges, approaches to targeted infectious disease surveillance often rely on simple testing strategies. We differentiate between \textit{rule-based} testing, in which simple deterministic (and often static) rules are used to guide test allocation (e.g., based on symptoms, location, timing or network) and \textit{risk-based} testing, where estimated risk is learned from the data and may evolve over time (e.g., testing individuals at an estimated higher risk of infection). Approaches to COVID-19 testing employed during the first few years of the pandemic provides an example. Testing strategies for residential university campuses focused mostly on \textit{rule-based} testing: symptom tracking and contact tracing, as well as scheduled on-campus screening with varying frequency \citep{Gressman2020, Bradley2020, Martin2020, Paltiel2020, Bahl2021, Rennert2021, vanderscaaf2021, schultes2021, Chang2021, Poole2021, Lopman2021, Ghaffarzadegan2021}. In general, the predominant infectious disease testing recommendation made by the World Health Organization suggested assigning tests to individuals having (i) symptoms consistent with COVID-19, (ii) contact with confirmed or suspected COVID-19 cases and (iii) evidence of recent travel history \citep{who2020}. Alternate suggestions advocate for fast and frequent random population testing \citep{larremore2021} and scheduled screening with repeated tests \citep{walke2020}. While contact tracing via efficient tracking system can be advantageous, its implementation is costly and not comprehensive enough as the spread of infectious disease advances \citep{gardner2021}. Other simple \textit{rule-based} strategies tend to miss asymptomatic infections (e.g., symptom-based), or require significant financial burden and compliance for a large and heterogeneous population (e.g., frequent random testing). Experiences with COVID-19 showed that in many cases it is not clear how to distribute tests across different prioritization groups. Analogous \textit{rule-based} approaches are often applied in HIV testing, including  ``key population testing" among groups known to be at higher risk and index partner testing. While such approaches can be effective in finding some individuals living with HIV, they miss others who lack a clear known risk profile \citep{balzer2019}.

Other concentrated efforts consisted of finding optimal testing strategies that inform epidemic dynamics \citep{chatzimanolakis2020} and helping to reduce disease spread \citep{biswas2020, jonnerby2020,gregg2021,du2021}. In particular, \cite{jonnerby2020} focuses on optimal allocations designed as a combination of group and segmented testing; segments of the population based on occupation, age and geographical location are given testing priority. Both \cite{biswas2020} and \cite{gregg2021} advocate for contextual bandits as a possible approach to the optimal testing allocation, with \cite{biswas2020} additionally suggesting an utility-based active learning solution. On the other hand, \cite{du2021} develop a probabilistic framework accounting for resource limitations, imperfect testing and the need for prioritizing higher risk patient populations. However, all of the proposed strategies impose strong modeling assumptions --- either on the type of dependence allowed (e.g. assuming a homogeneous Markov Decision Process), modeling conditional probabilities necessary to estimate the number of positive tests, or assuming which strata of the population constitutes at-risk profile. Finally, none of the proposed strategies provide an universal approach for emerging and existing epidemics. 

In this work, we propose an adaptive sequential design for a setting with network and temporal dependence where the goal is to optimize a short term outcome. The statistical problem is handled within a fully nonparametric model, respecting the true (unknown) dependence structure. While the proposed method is very general, it is particularly suited for infectious disease surveillance and epidemic control. We consider a longitudinal structure following $n$ individuals over a trajectory until time $\tau$. At each time point $t$ for sample $i$, one observes the exposure variable $A_i(t)$ (e.g., indicator of testing), outcome $Y_i(t)$ (e.g., health status) and other time-varying covariates $L_i(t)$ (e.g., network structure, location, symptoms). For an infectious disease surveillance, a decision maker/experimenter is in charge of assigning a test $A_i(t)$ to sample $i$ at time $t$, then collecting a vector of measurements $L_i(t)$ for the same individual, including the outcome. The exposure of interest is defined as a known stochastic intervention, where each treatment denotes a specific testing design (e.g., \textit{rule-based} or \textit{risk-based} testing). We study a setting in which the same decision maker can also adapt treatment assignment over time in response to past observations. Structuring the test allocation problem as an adaptive sequential design is paramount --- one can adapt the testing strategy as the infectious disease trajectory changes and other variants become dominant. 

In a setting where the goal is surveillance and control, it is natural to define performance of a treatment rule in terms of a short-term average over samples. Our causal target parameter is defined as the point prevalence of infection we would have obtained after one time-step, if, starting at time $t$ given the observed past, we had carried out a stochastic intervention $g_{t}^*$. The main goal is to optimize the next time-point outcome  (i.e. minimize the number of individuals infected) under $g_{t}^*$, at each $t$, under a possible resource constraint. Alternatively, one can also seek to optimize the short-term outcome under stochastic intervention as an average over time, therefore targeting the entire trajectory. This history-adjusted optimal choice for a single time point intervention then defines a new adaptive design over time, which we denote the \textit{Online Super Learner (SL) for adaptive sequential surveillance}. The regret minimization objective of the proposed design ensures that we assign tests such that as many infectious individuals as possible are subsequently diagnosed. As the design is adaptive, it learns the optimal choice of test strategies over time, responding to the current state of the epidemic. 

The proposed adaptive sequential design has crucial advantages over competing methods which make it particularly effective in the infectious disease context. The key strengths of our method are that we do not have to make strong (conditional) independence assumptions, or model network and time dependence. Instead of imposing unrealistic assumptions on the statistical model, the proposed method selects among adaptive designs with a short term performance Online Super Learner \citep{benkeser2018, malenica2021personalized}. As such, it imposes an honest benchmark to choose the best performing estimate in terms of the adaptive design performance. The necessary parts of the design (e.g., conditional expectation of outcome given the past) are estimated via an Online Super Learner, which relies on working models for its dependence structure. Therefore, the proposed method decides whether to learn across samples, through time, or both, based on the underlying (unknown) structure in the data. This is in contrast to previously described adaptive sequential designs, which rely on conditional independence assumptions (across time or samples) in order to deal with unknown dependence \citep{malenica2021adaptive, bibaut2021}. Secondly, as the true infectious status is unknown, the proposed target parameter is defined in terms of a latent outcome. We show that the average of true latent infectious status at time $t$ can be identified as the average of observed outcomes. As such, the statistical target parameter is defined in terms of the observed outcome, delineated as a function of the stochastic intervention we implement.

The article structure is as follows. In Section \ref{sec::formulation} we formally define the statistical estimation problem in the most general case, consisting of specifying notation, likelihood, and the statistical model. In subsection \ref{working_models}, we describe all the relevant working models, including assumptions underlying each. We define the target parameter, causal assumptions, and provide identification results in subsection \ref{target_parameter}. In Section \ref{sec::proposed_design}, we proceed to describe the proposed adaptive design, denoted as the Online Super Learner for adaptive sequential surveillance. Section \ref{sec::proposed_design} includes various proposed selectors, aimed at learning the optimal testing strategy for the sake of adaptive design performance. Section \ref{sec::sim_results} contains simulation results based on the proposed agent-based model for moderate size residential campus. We provide details on the agent-based model of SARS-CoV-2 transmission used for simulations, as well as how each testing strategy considered can be described as a stochastic intervention, in Appendix Section \ref{sec::fullmodel}. We conclude with a short discussion in Section \ref{sec::discussion}.

\section{Formulation of the Estimation Problem}\label{sec::formulation}
\subsection{Data and the Causal Model}\label{section1}

Consider a random variable denoted as $O_i$ for $i = 1, \ldots, n$, where $O_i$ is a sample $i$ trajectory. For each individual $i$, we define the following longitudinal data structure where
\begin{equation*}
    O_i = (L_i(0), A_i(1), L_i(1), \ldots, A_i(\tau), L_i(\tau)), 
\end{equation*}
corresponding to observations from time $t = 0$ to the final time point $t= \tau$.  Within time point $t$, we arbitrarily order data points by increasing sample index $i$, such that
\begin{equation*}
    (A_1(t), \ldots, A_n(t), L_1(t), \ldots, L_n(t))
\end{equation*}
reflects the unit ordering. We further decompose the trajectory of a given sample $i$ into baseline and time-varying parts. In particular, we define $O_i(0) = L_i(0)$ as a vector of baseline covariates which, by definition, are initiated at $t=0$. For infectious disease surveillance, $L_i(0)$ includes baseline infectious status, as well as other covariates (e.g., demographic information, initial network structure). The time-varying part of sample $i$ trajectory decomposes as $O_i(t) = (A_i(t), L_i(t))$, for $t = 1, \ldots, \tau$; it includes the treatment status occurring before the response variable and time-varying covariates, all indexed by time $t$. In particular, we define the exposure variable $A_i(t)$, corresponding to an indicator of being tested in an infectious disease surveillance design at time $t$. We define $L_i(t)$ as a vector of time-varying covariates, with the first component being the response variable --- infectious status for sample $i$ at time $t$. In addition to outcome, $L_i(t)$ also tracks the risk profile of unit $i$, as well as information on other units $\{1,\ldots, n\} \setminus \{i\}$ that belong to the network of sample $i$. The network of a given individual $i$ contained in $L_i(t)$ is denoted as $F_i(t)$, which reflects all the samples connected to individual $i$ at time $t$. In particular, we allow $|F_i(t)|$ to vary in $i$, but assume that this number is bounded by some known global constant $K$ that does not depend on $n$. Finally, we emphasize that the true infectious status for each sample and at each time point is typically not observed. Hence, we define the true latent outcome, notably the infectious status for sample $i$ at time $t$, as $Y^l_i(t)$. The observed outcome for sample $i$ at time point $t$ is denoted as $Y_i(t)$, where $L_i(t) = (Y_i(t), \ldots, F_i(t), \ldots)$.

For $n$ observed trajectories, we write $O = O^{\tau, n} = \{O_i\}_{i=1}^{n}$, where $O$ is a simplified notation that does not make dependence on $\tau$ and $n$ explicit. Under this notation, data observed throughout the course of the trial is $O$, with $O(t)=\{O_i(t)\}_{i=1}^{n}$ being the collection of $n$ time $t$-specific points. Similarly, let $L(t)$ and $A(t)$ denote $n$ dimensional time-specific vectors, effectively including time $t$-specific information across all $n$ collected samples; with that, we have that $L(t) = (L_1(t), \ldots, L_n(t))$ and $A(t) = (A_1(t), \ldots, A_n(t))$. Further, we write $Pa(O(t)) = \bar{O}(t-1) = (O(0), \ldots, O(t-1))$ to represent history of all samples up to time $t$. The complete histories until node $L(t)$ and $A(t)$ are denoted as $Pa(L(t)) = (\bar{O}(t-1), A(t))$ and $Pa(A(t)) = (\bar{O}(t-1))$, which are time $t$ histories of all $n$ samples. 
We also let time and unit-specific histories $Pa({A}_i(t))$ and $Pa({L}_i(t))$ denote all observations that come before $A_i(t)$ and $L_i(t)$, according to both time and sample ordering. In particular, let $Pa({L}_i(t)) = (Pa(L(t)), L_{1}(t), \ldots, L_{i-1}(t))$ and $Pa({A}_i(t)) = (Pa(A(t)), A_{1}(t), \ldots, A_{i-1}(t))$, where $i-1$ denotes sequential samples until sample $i$. Consequently, we let $Pa({O}_i(t)) = \bar{O}_i(t) = (\bar{O}(t-1), O_{1}(t), \ldots, O_{i-1}(t))$, where $Pa({O}_i(t))$ includes all history until time $t-1$ and $t$-specific samples until individual $i$.



\subsection{Statistical Model}

Let $\mathcal{M}$ denote the \textit{statistical model} for the probability distribution of the data, which is nonparametric, beyond possible knowledge of the treatment mechanism (i.e., the testing allocation strategy used). The more we know, or are willing to assume about the experiment that produces the data, the smaller the model. Let $P_0 \in \mathcal{M}$ denote the true probability distribution of $O$, such that $O \sim P_0$, and let $P$ denote any probability distribution where $P \in \mathcal{M}$. We let $p_o$ denote the density of $P_0$ with respect to (w.r.t) a dominating measure $\mu$. The likelihood of $o$ can be factorized according to the time-ordering as follows:
\begin{align}\label{likelihood}
    p_0(o) 
    &=  \prod_{i=1}^n p_{0,l(0)}(l_i(0)) \prod_{t=1}^{\tau} p_{0,a_i(t)}(a_i(t)|  Pa({a}_i(t)))  p_{0,l_i(t)}(l_i(t)| Pa({l}_i(t))) \\
    &=  \prod_{i=1}^n p_{0,l(0)}(l_i(0)) \prod_{t=1}^{\tau} g_{0,i,t}({a}_i(t) \mid Pa({a}_i(t))) q_{0,i,t}({l}_i(t) \mid Pa({l}_i(t))) \nonumber ,
\end{align}
where we let $a_i(t) \mapsto p_{0,a_i(t)}(a_i(t) \mid Pa({a}_i(t)))$ and $l_i(t) \mapsto p_{0,l_i(t)}(l_i(t) \mid Pa({l}_i(t)))$ denote conditional densities w.r.t. the dominating measures $\mu_A$ and $\mu_L$, respectively. We use shorthand notation for conditional densities and distributions of the relevant nodes. In particular, we write $q_{0,i,t}$ as the true $(i,t)$-specific conditional density of $L_i(t)$ based on the observed past until time $t$, $Pa({L}_i(t))$. The corresponding true conditional distribution of $L_i(t)$ conditional on $Pa({L}_i(t))$ is written as $Q_{0,i,t}$. At time $t$, $g_{0,i,t}$ reflects the true probability of drawing the testing indicator $A_i(t)$ conditional on the past until time $t$. We consider testing deployment settings in which $g_{0,i,t}$ is known and in control of the experimenter for most testing allocations; practically, it denotes a particular sampling and testing design implemented for sample $i$ at time $t$. We define ${\bar{Q}}_{0,i,t}(Pa({L}(t))) = {\bar{Q}}_{0,i,t}(A(t), Pa({O}(t))) \equiv \mathbb{E}_{0}[{Y}(t) | A(t), Pa({O}(t))]$ as the true conditional expectation of ${Y}(t)$ given the observed past. In order to emphasize the dependence on the treatment mechanism $g_{i,t}$, we might also write ${\bar{Q}}^{g_{i,t}}_{0,i,t}$ as the conditional expectation of outcome given the observed past under $g_{i,t}$. 

Due to the data ordering and considered application of the proposed method, we make the following two remarks on the network dependence structure:

\vspace{3mm}
\noindent
\textbf{Remark 1} (No dependence among treatment at time $t$).
\textit{For every $t \in [\tau]$ and $i \in [n]$, $A(t) = (A_1(t), \ldots, A_n(t))$ are independent conditional on 
$Pa(A(t))$.
}

\vspace{3mm}
\noindent
\textbf{Remark 2} (No dependence among outcomes at time $t$).
\textit{For every $t \in [\tau]$ and $i \in [n]$, $Y(t) = (Y_1(t), \ldots, Y_n(t))$ are independent conditional on 
$Pa(L(t))$.
}

\vspace{3mm}
\noindent
Both Remark (1) and (2) follow from the time and sample ordering, and are simply emphasized here. Testing is allocated based on all of the observed past and not influenced by other tests at time $t$.  The observed outcome is a direct consequence of tested individuals and observed past, but not influenced by other observed infectious individuals at time $t$. Remark (1) implies there is no interference among units, as test of subject $i$ does not have an effect on the observed outcome of subject $j$. Similarly, a positive observed outcome for $i$ does not affect the outcome of sample $j$ at $t$, but can at $t+1$, for example. As such, network dependence at $t$ strives from peer effects propagated via past until time $t-1$, including network at $t-1$ and infected individuals until $t-1$. Note that by Remark (1) and (2) we can write 
${\bar{Q}}_{0,i,t}(Pa({L}(t))) = \mathbb{E}_{0}[{Y}_i(t) | A_i(t), Pa({O}(t))]$ and $g_{0,i,t}({a}_i(t) \mid Pa({A}(t)))$. Finally, recall that $Pa({O}(t)) = \bar{O}(t-1)$ denotes observed history of $O(t)$ until time $t$. With that in mind, we write $P_{O_i(t) \mid \bar{O}(t-1)}$ (shorthand, $P_{\bar{O}(t-1)}$) as the time $t$ conditional distribution of $O_i(t)$ given the observed past $\bar{O}(t-1)$, which captures both temporal and network dependence.

Note that the decomposition presented in likelihood expression \eqref{likelihood} makes no assumptions regarding the interdependence of individuals in the network or across time. Therefore, the data reduces to a dependent observation with temporal and network dependence, and we observe only a single draw from $P_0$. In order to learn relevant parts of the data generating distribution, we would have to put some restrictions on the statistical model $\mathcal{M}$. In the following, we discuss several possible working models that enable us to learn parts of the likelihood, without assuming any of them explicitly. Via the proposed working models, depending on the extent to which they hold, we can take advantage of the different types of dependence structures without explicitly assuming any. As such, one can learn through time (therefore assuming some level of conditional stationarity), learn through the number of individuals (therefore assuming independence of samples given a known network), or both. We emphasize that one of the strengths of the proposed method is that it does not impose any direct assumptions on the statistical model $\mathcal{M}$. In the following, we describe all considered working models, and motivation behind each. 

\subsubsection{Working Models}\label{working_models}

We start by restricting the complexity of dependence allowed by supposing that each $L_i(t)$ can depend on the past only through a fixed dimensional summary measure of history, instead of the entire observed history. As such, we assume that $q_{0,i,t}({L}_i(t) \mid Pa({L}_i(t)))$ depends on the past only through a fixed dimensional summary measure $C_{L_i}(t)$, where $C_{L_i}(t)$ is a function of the observed history. Therefore, for every $t \in [\tau]$ and $i \in [N]$, $L_i(t)$ is independent of its past conditional on $C_{L_i}(t)$ and $q_{i,t}({L}_i(t) \mid Pa({L}_i(t))) = q_{i,t}({L}_i(t) \mid C_{L_i}(t))$. For some applications, the summary measure might cover a limited history, and the dependent process has a finite memory allowing us to learn through time. A particular example of summary measures are fixed dimensional extractions from the complete history, such that $C_{L_i}(t) \ = \ h_{L_i}(Pa({L}_i(t))) \in \mathbb{R}^{k}$ is a $(k)$-dimensional extraction of the form $C_{L_i}(t) \ = \ \{(L_j(s),A_j(s)) \ : \ s = t, t-1, t-2, \ldots, t-k, j \in [n]\}$. For other applications, the fixed dimensional summary measure might be a function of the sample $i$'s network; as such, we might know that conditional probability of $L_i(t)$ depends only on the history of $j$ samples, where $j \in F_i(t) \cup i$. In case of both time and network dependence, $C_{L_i}(t)$ could be a function of both sample $i$'s network and previous past time-points where $C_{L_i}(t) \ = \ \{(L_j(s), A_j(s)) \ : \ s = t, t-1, t-2, \ldots, t-k, j \in F_i(t) \cup i\}$; then $C_{L_i}(t)$ is a summary measure of the history over last $k$ steps of a set $F_i(t)$ of at most $K$ friends. We note that, if $F_i(t) = \emptyset$, our formulation reduces to an i.i.d. setting across samples. In order to formally present our target parameter under a working model, we make the following assumption on the decomposition of the fixed dimensional summary measure, as stated below. 
\begin{assumption}[Decomposition of the fixed dimensional summary]\label{decompose} 
\textit{For every $t \in [\tau]$ and $i \in [n]$, the fixed dimensional summary measure $C_{L_i}(t)$ can be written as
\begin{equation*}
   C_{L_i}(t) = (A_i(t), C_{A_i}(t)),
\end{equation*}
where $C_{A_i}(t) = h_{A_i}(Pa({A}_i(t))) = h_{A_i}(Pa({O}_i(t)))$.}
\end{assumption}

\begin{assumption}[Conditional independence given a summary measure]\label{fixed_dim_summary}
\textit{For every $t \in [\tau]$ and $i \in [n]$, 
\begin{equation*}
    q_{i,t}({l}_i(t) \mid  c_{L_i}(t)) = q_{i,t}({l}_i(t) \mid Pa({l}_i(t)))
\end{equation*}
where $c_{L_i}(t)$ is the observed fixed dimensional summary of the past until time $t$.}
\end{assumption}

The following key assumption is a modeling assumption on the conditional density of $L_i(t)$ given the observed past. Consistent with Assumption 4 in \cite{bibaut2021}, we might assume that the conditional distribution of $L_i(t)$ given the observed fixed dimensional summary of the history is a constant function across samples and time. As such, there exists a common in $i$ and $t$ conditional density $q$ such that $q_{i,t} = q$. Drawing from the reinforcement learning literature, this assumption is analogous to the homogeneity assumption for the Markov Decision Process \citep{alagoz2010}. Under Assumption \ref{fixed_dim_summary} and allowing for a common in $i$ and $t$ conditional density of $L_i(t)$ given the history, we can rewrite the likelihood from equation \eqref{likelihood} as:
\begin{align}\label{likelihood_tn}
    p(o) =  \prod_{i=1}^n p_{l(0)}(l_i(0)) \prod_{t=1}^{\tau} g_{0,i,t}({a}_i(t) \mid Pa({a}_i(t))) q({l}_i(t) \mid c_{L_i}(t)).
\end{align}
Note that, since $g_{0,i,t}$ is known, we don't need to put any restrictions on the treatment mechanism given the past. We emphasize that assuming a conditional density $q$ which is common in $i$ and $t$ still allows for a rich network and time-dependent structure given $C_{L_i}(t)$. The proposed formulation allows us to learn and measure factors which result in changes over time and network, captured with varying $C_{L_i}(t)$. For example, we could have that $C_{L_i}(t) \ = \ h_{L_i}(Pa({L}_i(t))) \in \mathbb{R}^{k \times j}$ is a $(k,j)$-dimensional extraction of the form $C_{L_i}(t) \ = \ \{(L_j(s),A_j(s)) \ : \ s = t, t-1, t-2, \ldots, t-k, \ \ j \in F_i(t) \cup \ i\}$ where $F_i(t)<K$ and $K$ is not a function of $n$. As such, our working model covers finite memory time dependence and network structure where each individual has a limited number of contacts, both of which could possibly vary as the trajectory advances. Alternatively, the proposed working model could cover dependence structures described by summary measures of the time series pattern (e.g.: moving average, finite memory, features related to STL decomposition of the series, spectral entropy, Hurst coefficient) and summary measures of the current state of the network (e.g.: current state of the epidemic, percent isolated, percent wearing masks). 
Overall, in the adaptive sequential surveillance design, such modeling assumptions equate to conditional stationarity of the outcome mechanism over the entire trajectory (common in time) and for each sample (common across samples) given a fixed dimensional summary of the past. Instead of assuming a common conditional density of $L_i(t)$, we might alternatively only need to assume a common conditional expectation. Then, under the common in $(i,t)$-working model, we have that $\bar{Q}_{i,t} = \bar{Q}$ instead of the full $q = q_{i,t}$. We write assumptions on the conditional expectation of $Y(t)$ given the observed past under the common in $(i,t)$ model as Assumption \ref{common_it}. With that, Assumption \ref{decompose}, \ref{fixed_dim_summary} and $\ref{common_it}$ constitute working model $\mathcal{M}^{tn}$.

\begin{assumption}[Common in $i$ and $t$ conditional expectation of the outcome]\label{common_it}
\textit{There exists a common across samples ($i$) and time ($t$) conditional expectation of $Y_i(t)$ given the observed past, such that  $\bar{Q}_{i,t} = \bar{Q}$ for every $t \in [\tau]$ and $i \in [n]$ where}
\begin{equation*}
{\bar{Q}}_{i,t}(A_i(t), C_{A_i}(t)) = {\bar{Q}}(A_i(t), C_{A_i}(t)).
\end{equation*}
\end{assumption}

\begin{definition}[Working model $\mathcal{M}^{tn}$]\label{work_model_1}
\textit{We define a working model $\mathcal{M}^{tn}$ as the set of distributions $P$ over the domain $\mathcal{O}$ that satisfy Assumptions \ref{decompose}, \ref{fixed_dim_summary} and \ref{common_it}.}
\end{definition}

Alternatively, we may assume that the conditional expectation of $Y(t)$ given the past is a $t$-common mechanism given the history, allowing for a possibly very dense network structure which might be observed during a highly contagious epidemic. Conditional on the observed fixed dimensional summary $C_{L_i}(t)$, $q_{i,t}$ is a common in $t$ density smooth enough to be learned through time. This working model assumption is analogous to models previously described in the time-series literature, extended to multiple trajectories \citep{vanderLaan2018onlinets, malenica2021adaptive,malenica2021personalized}. We can rewrite the likelihood from equation \eqref{likelihood} under common-in-$t$ density as follows
\begin{align}\label{lh::likelihood_t}
    p(o) =  \prod_{i=1}^n p_{l(0)}(l_i(0)) \prod_{t=1}^{\tau} g_{0,i,t}({a}_i(t) \mid Pa({a}_i(t))) q_{i}({l}_i(t) \mid c_{L_i}(t)).
\end{align}
In terms of the conditional expectation, we emphasize that the functional form of $\bar{Q}_{i,t} = \bar{Q}_{i}$ is unspecified, with the only assumption being that $\bar{Q}_i$ is common in time conditional on a fixed dimensional summary. In the current application, such modeling assumptions would equate to conditional stationarity of the expected outcome over the entire trajectory (common in time), but not common across samples. We denote the working model described by the Assumptions \ref{decompose}, \ref{fixed_dim_summary} and \ref{common_t} as $\mathcal{M}^t$. 


\begin{assumption}[Common in $t$ conditional expectation of the outcome]\label{common_t}
\textit{There exists a common across time ($t$) conditional expectation of $Y_i(t)$ given the observed past, such that $\bar{Q}_{i,t} = \bar{Q}_{i}$ for every $t \in [\tau]$ and $i \in [n]$ where}
\begin{equation*}
 {\bar{Q}}_{i,t}(A_i(t), C_{A_i}(t)) = {\bar{Q}}_i(A_i(t), C_{A_i}(t)).
\end{equation*}
\end{assumption}

\begin{definition}[Working model $\mathcal{M}^{t}$]\label{work_model_2}
\textit{We define a working model $\mathcal{M}^{t}$ as the set of distributions $P$ over the domain $\mathcal{O}$ that satisfy Assumptions \ref{decompose}, \ref{fixed_dim_summary} and \ref{common_t}.}
\end{definition}

Instead of learning across time, one might instead rely on asymptotics in the number of individuals. An important ingredient of this modeling approach is to assume that any dependence of unit $i$ can be fully described by a function of the known network over time. Let $F_i(t) \leq K$ denote the network for sample $i$ at time $t$. Then, there is a common in $i$ density, $q_t$, allowing for possibly very long and elaborate time-dependence. Similarly, there is a common-in-$i$ expectation conditional on a fixed dimensional summary measure $C_{L_i}(t)$. In contrast to decomposition presented in \eqref{likelihood}, likelihood under common-in-$i$ density is written as follows
\begin{align}\label{likelihood_n}
    p(o) =  \prod_{i=1}^n p_{l(0)}(l_i(0)) \prod_{t=1}^{\tau} g_{0,i,t}({a}_i(t) \mid Pa({a}_i(t))) q_{t}({l}_i(t) \mid c_{L_i}(t)).
\end{align}
Under no conditional stationarity assumption, one could use the recent estimates of $\bar{Q}_t$ in order to optimize the next sampling mechanism w.r.t the status of the epidemic a few time points in the future. This implies that it is possible to learn the common-in-$i$ expectation $\bar{Q}_t$ from a draw $O$ as $n \rightarrow \infty$, resulting in a well-defined statistical estimation problem. For the adaptive surveillance problem, this formulation allows us to learn across samples, as dynamics of the trajectory are not stationary over time, but possibly evolving. Assumptions \ref{decompose}, \ref{fixed_dim_summary} and \ref{common_i} constitute the working model $\mathcal{M}^n$.


\begin{assumption}[Common in $i$ conditional expectation of the outcome]\label{common_i}
\textit{There exists a common across samples ($i$) conditional expectation of $Y_i(t)$ given the observed past, such that $\bar{Q}_{i,t} = \bar{Q}_{t}$ for every $t \in [\tau]$ and $i \in [n]$ where}
\begin{equation*}
 {\bar{Q}}_{i,t}(A_i(t), C_{A_i}(t)) = {\bar{Q}}_t(A_i(t), C_{A_i}(t)).
\end{equation*}
\end{assumption}

\begin{definition}[Working model $\mathcal{M}^{n}$]\label{work_model_3}
\textit{We define a working model $\mathcal{M}^{n}$ as the set of distributions $P$ over the domain $\mathcal{O}$ that satisfy Assumptions \ref{decompose}, \ref{fixed_dim_summary} and \ref{common_i}.}
\end{definition}

\subsection{Target Parameters}\label{target_parameter}

In the following, we describe a counterfactual scenario in which the initial testing mechanism is replaced by user-defined conditional distributions, and define the corresponding target parameter of interest. 
Our main aim is to describe adaptive sequential surveillance for infectious disease under unknown network and time dependence. This entails defining the time $t$-specific testing strategy which optimizes the short term outcome among a set of proposed testing schemes. The optimal testing strategy maximizes the number of detected cases, with the target parameter under a resource constraint being of particular interest in practice. Instead of focusing only on the time $t$-parameter, we also define an average over the entire trajectory as a target parameter of interest. 

\subsubsection{Structural Equation Model}

In the previous section, we discuss the distribution of the observed data. Given a dataset, we can estimate parameters of this distribution. However, without more structure, statistical parameters do not have a causal interpretation. In order to translate the scientific question of interest into a formal causal quantity, we additionally specify a structural equation model (SEM; equivalently, structural causal model (SCM)) \citep{pearl2009}.
By specifying a SEM, we assume that each component of the data structure is a function of the observed endogenous variables and an unmeasured exogenous error term \citep{pearl2009}. We encode the time-ordering of the variables using the following SEM for each $t$:
\begin{align}\label{NPSEM0}
    L_i(0) &= z_{L_{i}(0)}(U_{L_{i}}(0)) \\ 
    A_i(t) &= z_{A_{i}(t)}(\bar{O}(t-1),U_{A_{i}}(t)) \nonumber \\
    L_i(t) &= z_{L_{i}(t)}((\bar{O}(t-1),A_i(t)),U_{L_{i}}(t)), \nonumber
\end{align}
where $U := (U_A, U_L)$ with $U_A := (U_{A_{i}}(t): t \in [\tau], i \in [n])$ and $U_L := (U_{L_{i}}(t): t \in [\tau], i \in [n])$. The unmeasured exogenous variables are sampled from $P_U$, such that $U \sim P_U$. Given an input $(U,O)$, structural equations $z_{A_{i}(t)}$ and $z_{L_{i}(t)}$ for each time $t \in [\tau]$ and sample $i \in [n]$ deterministically assign a value to each of the nodes. While we have a specification of $z_{A_{i}(t)}$ in a randomized trial, the structural equations  $z_{L_{i}(t)}$ do not restrict the functional form of the causal relationships for any $t$ or $i$. The SEM defines a collection of distributions $(U,O)$ representing the full data model, here defined in terms of $U$ and observed data $O$. Let $P_{0}^F$ denote the true probability distribution of $(U,O)$; in the remainder of the article, we will use the subscript ``0'' to indicate true probability distributions or components thereof. Here we emphasize that any distribution $P^F$ on the domain of the full data fully determines a corresponding distribution $P$ on the domain of the observed data. Finally, we denote the model for $P_{0}^F$ as $\mathcal{M}^F$, known as the \textit{causal model}. Here we emphasize that all dependence across units is not due to dependence of errors, but solely due to the interdependence between units as described by the SEM which allows for individuals' treatment and outcome to depend the history of other individuals. 

We can also define the \textit{time- and history- specific causal model}. Let $\mathcal{M}_t^F(\bar{o}(t-1))$ denote the set of conditional probability distributions $P_{\bar{O}(t-1)}^F$, which condition on the observed history by time $t$, $\bar{o}(t-1)$. In particular, $\mathcal{M}_t^F(\bar{o}(t-1))$ is compatible with the structural equations model \eqref{NPSEM0} by imposing $\bar{O}(t-1) = \bar{o}(t-1)$:
\begin{align}\label{NPSEM1}
    L_i(0) &= z_{L_{i}(0)}(U_{L_{i}}(0)) \\ 
    A_i(t) &= z_{A_{i}(t)}(\bar{o}(t-1),U_{A_{i}}(t)) \nonumber \\
    L_i(t) &= z_{L_{i}(t)}((\bar{o}(t-1),A_i(t)),
    U_{L_{i}}(t)). \nonumber
\end{align}

\subsubsection{Target Parameter on the SEM and Identifiability}

The causal model allows us to define counterfactual random variables as functions of $(U,O)$ corresponding with arbitrary interventions. In particular, we can replace the observed data generating distribution for the treatment mechanism by user-specified hypothetical conditional distributions; such non-degenerate choices of intervention distributions are referred to as \textit{stochastic interventions} \citep{diaz2012}. Let $g_{i,t}^*$ denote a stochastic intervention at time $t$ identified as a conditional distribution of $A_i^*(t)$ given the observed past. We write $g_{t}^* = \{g_{1,t}^*, \ldots, g_{n,t}^*\}$ for all $n$ interventions at time $t$. With that, $O^*(t)$ is the counterfactual full data generated from the SEM described in \eqref{NPSEM1} by replacing the equation associated with the exposure node by the counterfactual intervention $g_{i,t}^*$ at time $t$,
\begin{align}\label{NPSEM2}
    L_i(0) &= z_{L_{i}(0)}(U_{L_{i}}(0)) \\ 
    A_{i,g_{i,t}^*}(t) &\sim g_{i,t}^*(\cdot | \bar{o}(t-1)) \nonumber \\
    L_{i,g_{i,t}^*}(t) &= z_{l_{i}(t)}((\bar{o}(t-1), A_{i,g_{i,t}^*}(t)),U_{L_{i}}(t)). \nonumber
\end{align}
We write $(U(t), O^*(t))$ as the full post-intervention data at time $t$, with the post-intervention distribution denoted as $P_{\bar{O}(t-1)}^{F*}$.
Consequently, the counterfactual latent outcome under $g^*_{i,t}$ is written as ${Y}_{i,g_{i,t}^*}^{l}(t)$ for the sample $i$ at time $t$. 
We define our causal parameter of interest as
\begin{equation}\label{causal_psi_t}
    {\Psi}_{t,g^{*}_{t}}^{F}(P_{\bar{O}(t-1)}^F) =\mathbb{E}_{P_{\bar{O}(t-1)}^{F*}} \left[ \frac{1}{n}\sum_{i=1}^n {Y}_{i,g_{i,t}^*}^{l}(t) \right],
\end{equation}
which is the expectation of the counterfactual random variable ${Y}_{i,g_{i,t}^*}^{l}(t)$ generated by the
modified SEM as written in equation \eqref{NPSEM2}. Our causal target parameter is the mean latent outcome we would have obtained after one time-step, if, starting at time $t$ given the observed past, we had carried out intervention $g^{*}_{t}$.


By defining the causal quantity of interest in terms of stochastic interventions on the SEM and providing a link between the causal model and the observed data, we lay the groundwork for addressing identifiability through $P_0$. In order to express ${\Psi}_{t,g^{*}_{t}}^F(P_{\bar{O}(t-1)}^F)$ as a parameter of the distribution $P_{\bar{O}(t-1)}$ of the observed data $O$, we add two key assumptions on the SEM: the sequential randomization assumption (Assumption \ref{randomization}, which automatically holds by design) and the positivity assumption (Assumption \ref{positivity}).

\begin{assumption}[Sequential Randomization]\label{randomization}
\textit{For any $t \in [\tau]$ and $i \in [n]$, 
$$A_i(t) \indep {Y}_{i,g_{i,t}^*}^{l}(t) \mid Pa(A(t)) \ \ \text{and} \ \ A_{i,g_{i,t}^*}(t) \indep {Y}_{i,g_{i,t}^*}^{l}(t) \mid Pa(A(t)).$$}
\end{assumption}

\begin{assumption}[Positivity]\label{positivity}
\textit{For any $t \in [\tau]$ and $i \in [n]$ with $P(Pa(A(t)) = Pa(a(t))) > 0$, 
$$g_o(A_i(t) \mid Pa(A(t)) = Pa(a(t))) > 0.$$}
\end{assumption}


\begin{theorem}\label{lemma:identify}
Assume assumptions \ref{randomization} and \ref{positivity} hold. Under consistency, we denote the time $t$ value under the stochastic intervention $g^*_{t}$ as
\begin{align*}
    {\Psi}_{t,g^{*}_{t}}^F(P_{\bar{O}(t-1)}^F) = {\Psi}_{t,g^{*}_{t}}(P_{\bar{O}(t-1)}) &= \int_{a} \frac{1}{n}\sum_{i=1}^n \mathbb{E}_P [{Y}_i^l(t) \mid a, \bar{o}(t-1)] g_{i,t}^*(a \mid \bar{o}(t-1)) d\mu_a(a) \\ \nonumber &= \frac{1}{n}\sum_{i=1}^n \mathbb{E}_{\bar{Q}_{i,t},g_{i,t}^*}[Y_i(t) \mid \bar{o}(t-1)] 
\end{align*}
where the observed outcome is defined as $Y_i(t) = A_i(t)Y_i^l(t)/g_{i,t}(A_i(t) \mid \bar{O}(t-1))$ and $\psi_{t,g^{*}_{t}} = {\Psi}_{t,g^{*}_{t}}(P_{\bar{O}(t-1)})$.
\end{theorem}

\begin{proof}
We allocate the derivation to the Appendix section \ref{sec:identify}.
\end{proof}

\noindent
Note that Theorem \ref{lemma:identify} identifies the causal parameter in terms of both the latent and observed outcome at each time point $t$. 
As per Theorem \ref{lemma:identify}, the statistical target parameter (estimand) is denoted as 
\begin{align}\label{eq::psit}
     \psi_{t,g^{*}_{t}} =  {\Psi}_{t,g^{*}_{t}}(P_{\bar{O}(t-1)}) = \frac{1}{n}\sum_{i=1}^n \mathbb{E}_{\bar{Q}_{i,t},g_{i,t}^*}[Y_i(t) \mid \bar{O}(t-1)]. 
\end{align}
With a slight abuse of notation, we write ${\Psi}_{t,g^{*}_{t}}(P_{\bar{O}(t-1)})$ as ${\Psi}_{t,g^{*}_{t}}(\bar{Q}_{t})$. Instead of focusing on just the time $t$-target ${\Psi}_{t,g^{*}_{t}}(P_{\bar{O}(t-1)})$, we can additionally define a time- and sample- specific estimand as 
\begin{equation}\label{eq::psiit}
    \psi_{i,t,g^*_{i,t}} = {\Psi}_{i,t,g^{*}_{i,t}}(P_{\bar{O}(t-1)}) = \mathbb{E}_{\bar{Q}_{i,t},g_{i,t}^*}[Y_i(t) \mid \bar{O}(t-1)], 
\end{equation}
where $\psi_{t,g^*_{t}} = 1/n \sum_{i=1}^n \psi_{i,t,g^*_{i,t}}$. We emphasize here that $\psi_{t,g^{*}_{t}}$ is $(t,g^{*}_{t})$- specific and $\psi_{i,t,g^{*}_{i,t}}$ is $(i,t,g^{*}_{i,t})$- specific. We are also interested in a target parameter that is defined over the entire trajectory. In terms of an infectious disease outbreak, the target parameter defined in Equation \eqref{eq::psi} would equate to continuous surveillance and epidemic control over the entire observational period $\tau$:
\begin{equation}\label{eq::psi}
    \psi^{\tau} = {\Psi}{\tau}(P) = \frac{1}{\tau} \sum_{t=1}^{\tau} \Psi_{t,g^{*}_{t}}(P_{\bar{O}(t-1)}).
\end{equation}
We refer to all three in the following sections, with a particular focus on parameters in Equation \eqref{eq::psit} and \eqref{eq::psi}. 

Finally, as testing resources are often constrained in the context of infectious disease testing (either due to shortages of tests themselves, as occurred early in the COVID-19 pandemic, or due to resources available to support testing), we assume a fixed testing capacity at each time-point until the end of the epidemic. As such, it is necessary to provide an optimal allocation of the available testing resources, analogous to the resource-constrained optimal individualized treatment literature \citep{luedtke2016resource}. Suppose that the number of available tests are limited at each time point $t$, so that at most $k \in (0,1)$ proportion of the population can get tested. Our ultimate interest might be in optimizing Equation \eqref{eq::psi} under a resource constraint, meaning that we want to optimize the true number of infected individuals by the end of the trajectory. The more positive cases we can detect at each $t$ under the $k$ testing constraint, which would then be isolated and treated, the fewer incidence of downstream transmission can occur, resulting in a greater infection control. 

\section{Online Super Learner for Adaptive Surveillance}\label{sec::proposed_design}


Let $\{g^*_{t,1}, \ldots, g^*_{t,S}\} \in \mathcal{G}$ denote a collection of $S$ user-specified stochastic interventions for all samples $i \in [n]$ at $t$.  Note that all considered testing schemes are elements of the space $\mathcal{G}$, which consists of a finite number of testing strategies considered at each time point. Therefore, $g^*_{t,s}$ is a $s$-specific conditional distribution of $A(t)$ given the observed past $\bar{O}(t-1)$ at time $t$. For the $s$-specific stochastic intervention, it then follows that
\begin{equation*}
    \psi_{t,g^{*}_{t,s}} = {\Psi}_{t,g^{*}_{t,s}}(\bar{Q}_{t}) = \frac{1}{n}\sum_{i=1}^n \mathbb{E}_{\bar{Q}_{i,t},g_{i,t,s}^*}[Y_i(t) \mid \bar{O}(t-1)],
\end{equation*}
and we have a separate $\psi_{t,g^{*}_{t,s}} = {\Psi}_{t,g^{*}_{t,s}}(\bar{Q}_{t}) $ for each $s \in \{1, \ldots, S\}$ at $t$. As infectious disease progression evolves over time, we want the proposed adaptive sequential design to be able to respond to the current state of the epidemic. At the beginning of the disease trajectory, catching the few infected individuals and testing their proximal network might be enough to control the spread. However, as the contagion reaches the state of an epidemic, identifying individuals at high risk (but possibly asymptomatic) might be crucial. While one of the $s$-specific stochastic interventions might be optimal at the beginning of the trajectory, another one might be optimal at later points. The enforced adaptive sequential surveillance should therefor evolve and adapt over time in response to the current state of the infectious disease progression. The problem then becomes a matter of selecting amongst stochastic interventions in $\mathcal{G}$, over the entire trajectory, while not imposing assumptions on the statistical model $\mathcal{M}$. In the following, we describe an \textit{Online Super Learner for adaptive sequential surveillance} which uses different selectors to pick the optimal stochastic intervention $s$ at time $t$, over the entire trajectory. 

\subsection{Loss-based selector}\label{sec::loss-selector}

We can define an adaptive sequential design for surveillance as an online algorithm that at each time point $t$ fits the conditional distribution of treatment given the past observations. As such, it's an online mapping of past data into $g_{t}^*(A(t) \mid \bar{O}(t-1))$, while learning over time how to adapt $g_{t}^*$ in order to optimize a short term outcome. With that it mind, we can formulate the problem at hand within the loss-based estimation paradigm \citep{dudoit2003a,dudoit2003b,vaart2006,laan2006oracle}. In the following, we proceed to define key concepts necessary for establishing an Online Super Learner --- including a valid loss, risk, cross-validation scheme, and the discrete Super Learner \citep{sl2007,benkeser2018,malenica2021personalized}.

To start, let $P_{n,t}$ denote the empirical distribution of $n$ time-series collected until time $t$. We define the estimator mapping, $\hat{\Psi}_{t,g^{*}_{t}}$, as a function from the empirical distribution to the parameter space. In particular, let $P_{n,t} \mapsto \hat{\Psi}_{t,g^{*}_{t}}$ represent a mapping from $P_{n,t}$ into a function $\hat{\Psi}_{t,g^{*}_{t}}(P_{n,t})$. Then, $\hat{\Psi}_{t,g^{*}_{t}}(P_{n,t})(\bar{O}(t-1))$ denotes the target function evaluated at the observed past. Similarly, the estimator mapping $\hat{\bar{Q}}_{i,t}$ is defined as a function of the empirical distribution. We can write $\hat{\bar{Q}}_{i,t}(P_{n,t})(A_i(t),\bar{O}(t-1))$ as the predicted outcome for unit $i$ of the estimator $\hat{\bar{Q}}_{i,t}(P_{n,t})$ at time $t$, based on $(A_i(t),\bar{O}(t-1))$. In Section \ref{sec::osl}, we elaborate on the loss-based parameter definition and estimation given the past under working models described in Section \ref{working_models}. 

Let $C(i,m)$ denote the unit $i$-
and time $m$-specific collection where $C(i,m) = (Y_i(m), A_i(m), \bar{O}(m-1))$; similarly, we write $C(m) = (Y(m), A(m), \bar{O}(m-1))$ as the time $m$- specific record. Let $L(\hat{\Psi}_{t,g_{t}}(P_{n,t}))(C(m))$ define a loss function for the time- specific target, such that $L(\hat{\Psi}_{t,g_{t}}(P_{n,t})) : \mathcal{C} \rightarrow \mathbb{R}$. By construction, a valid loss for a given parameter of interest is defined as a function whose true conditional mean is minimized by the true value of the target. For instance, for the time $t$-specific target we then have that
\begin{equation*}
    P_{0,\bar{O}(t-1)} L(\hat{\Psi}_{t,g_{0,t}}(\bar{Q}_{0,t}))(C(t)) = \min_{\hat{\Psi}_{t,g_{t}}(\bar{Q}_{t})} P_{0,\bar{O}(t-1)} L(\hat{\Psi}_{t,g_{t}}(\bar{Q}_{t}))(C(t)).
\end{equation*}
For a binary outcome, we can further define $L(\hat{\Psi}_{t,g_{t}}(P_{n,t}))$ as the inverse weighted mean squared error function (MSE), which is the loss we are trying to minimize
\begin{equation}\label{eqn::weighted_loss}
    \frac{1}{n} \sum_{i=1}^n \frac{1}{g_{i,t}(A_i(t) \mid \bar{O}(t-1))} \left( Y_{i,g_{i,t}}(t) - \hat{\Psi}_{i,t,g_{i,t}}(P_{n,t})(\bar{O}(t-1)) \right)^2.
\end{equation}


The true risk is defined as the expected value of the loss evaluated w.r.t the true distribution. As such, it establishes the true measure of performance of the target parameter with respect to the specified loss --- however, it is an unattainable quantity, as the truth is unknown. In order to obtain an unbiased estimate of the true risk, we use cross-validation. Let $P_{n,t}^0$ denote the empirical distribution of the training sample until time $t$, with $P_{n,t}^1$ the corresponding empirical distribution of the validation set. In general, we use different cross-validation schemes to evaluate how well an estimator trained on specific samples' past is able to predict an outcome for samples in the future, which is reflected in different empirical distributions $P_{n,t}^0$ and $P_{n,t}^1$. 
For an infectious disease, we might expect its natural trajectory to vary over time, but have a similar profile across proximal time points. Therefore, we let $P_{n,t}^0$ be the empirical distribution of all the data until time $t$, with $P_{n,t}^1$ consisting of samples at the next time step $t+1$.  The cross-validated risk over all times then corresponds to
\begin{equation*}
    R_{CV,m} =   
    \frac{1}{\tau - m} \sum_{m=t+1}^{\tau}
    L(\hat{\Psi}_{m-1,g_{m-1}}(P_{n,m-1}^0))(C(m)).
\end{equation*}

Let $\hat{\Psi}_{t,g^*_{t,s}}(P_{n,t})$ denote the estimator of the target parameter under design $g^*_{t,s}$, where we have a separate $\hat{\Psi}_{t,g^*_{t,s}}(P_{n,t})$ for each $s \in S$. We can evaluate the performance of each stochastic intervention $g^*_{t,s}$ using the loss-based framework. The proposed evaluation therefore proceeds as follows: with each new $C(t+1)$, we evaluate the loss $L(\hat{\Psi}_{t,g^*_{t,s}}(P_{n,t}^0))(C(t+1))$ for each $s \in S$; add this loss to the current estimate of the online CV risk; update each online estimator $\hat{\Psi}_{t,g^*_{t,s}}$ into $\hat{\Psi}_{t+1,g^*_{t+1,s}}$ using $C(t+1)$. Upon observing the next batch of data, $C(t+2)$, the process is repeated. The Online CV risk gives us an estimated performance of the adaptive design over time. We can use the full online CV risk if we are interested in optimizing over the entire trajectory, or an average over a more recent window for time $t$-specific target parameter. We define the discrete SL design $s_t$ as the design which minimizes the online CV risk at time $t$:
\begin{align*}
    s_t &= \min_s \frac{1}{\tau_w - m} \sum_{m=t+1}^{\tau_w} L(\hat{\Psi}_{m,g^*_{m,s}}(P_{n,m}))(C(m)), 
\end{align*}
where $\tau_w$ is a future time point based on the window size $w$, and the loss function is defined as in Equation \eqref{eqn::weighted_loss}.

\subsection{TMLE- and TMLE-CI-based selector}\label{sec::tmle-selector}

Continuing with the definition of an adaptive sequential design being an online algorithm, we use past data in order to fit relevant parts of the likelihood of $O$. At each $t$, we run a simulation under a different design $g_{t,s}^*$ (and current estimate of the relevant parts of the likelihood), and select the design which optimizes a short term mean outcome. 
The loss-based selector in the previous section optimizes over a window of recent losses (e.g., inverse weighted MSE over a window of time points). We could instead optimize for the MSE such that we would also have inference for the target parameter --- allowing us to pick a design by taking into account uncertainty in the point estimate as well. This motivates a new selector, based on the Targeted Minimum Loss Estimation (TMLE) \citep{Rubin_2006, book2011,book2018}. The estimated mean outcome under each of the $s$ designs is a TMLE based on a working model, optimized for MSE with inference. Here we emphasize that reliance on working models is a necessary step in order to obtain a ranking of designs based on the TMLE, but our proposed method does not rely on assumptions imposed by the working models. 

For the TMLE-based selector, we want to obtain a TMLE of each $s$-specific target parameter. The standard TMLE, as originally defined by \cite{Rubin_2006}, first computes an initial estimator of $\bar{Q}_{0,i,t}$. In general, the functional form of $\bar{Q}_{0,i,t}$ is unknown, with arbitrary dependence structure. In order to avoid unnecessary assumptions, we resort to data-adaptive predictive methods such as the Online Super Learner under various working models; as such, we allow for flexibility in the specification of the functional form and dependence structure. 
Consistent estimation of $\bar{Q}_{0,i,t}$ is key for achieving asymptotic efficiency of the target parameter \citep{book2011,book2018}. We denote the initial estimator of the conditional mean outcome given the past as $\hat{\bar{Q}}_{i,t}(P_{n,t}^0)$, trained on the training data available until time $t$, $P_{n,t}^0$.

The initial estimator of the conditional mean outcome given the past is then updated in such a way that the efficient influence function (EIF) estimating equation is zero when computed at the updated estimate. The TML estimators defined in this way generally require optimizing a loss function iteratively for the likelihood of the observed data. Achieving a solution to the EIF estimating equation guarantees, under regularity assumptions, that the estimator enjoys optimality properties such as double robustness and local efficiency \citep{Rubin_2006,book2011,book2018}. We solve the estimating equation by fitting the following logistic model 
\begin{equation*}
    \text{logit} \ \hat{\bar{Q}}_{t,\epsilon}(A(t),C_{A}(t))  = 
    \text{logit} \ \hat{\bar{Q}}_{t}(A(t),C_{A}(t)) + \epsilon,
\end{equation*}
with weights defined as $w_t = g_{t}^*(A(t) \mid C_{A}(t))/g_{t}(A(t) \mid C_{A}(t))$. We emphasize that $g_{t}(A(t) \mid C_{A}(t))$ denotes the treatment mechanism generating the data so far for all $n$ samples at time $t$. The estimate of $\epsilon$ is written as $\epsilon_t$, with the updated initial estimator of $\bar{Q}_{0,i,t}$ evaluated at $(A_i(t),C_{A_i}(t))$ denoted as $\hat{\bar{Q}}^*_{i,t}(A_i(t),C_{A_i}(t)) = \hat{\bar{Q}}_{i,t,\epsilon_t}(A_i(t),C_{A_i}(t))$. The targeted estimate $\hat{\bar{Q}}^*_{i,t}(A_i(t),C_{A_i}(t))$ then solves the following EIF estimating equation,
\begin{equation*}
    \frac{1}{n}\sum_{i=1}^n \frac{g_{t,i}^*(A_i(t) \mid C_{A_i}(t))}{g_{t,i}(A_i(t) \mid C_{A_i}(t))}(Y_i(t) - \hat{\bar{Q}}^*_{i,t}(A_i(t),C_{A_i}(t))) = 0.
\end{equation*}
The TMLE of the $s$-specific stochastic intervention is defined as the plug-in estimator under the targeted estimate $\hat{\bar{Q}}^*_{i,t}$ and $g_{i,t,s}^*$,
\begin{equation*}
    {\Psi}_{t,g^{*}_{t,s}}(\hat{\bar{Q}}^*_{t}) = \frac{1}{n}\sum_{i=1}^n \mathbb{E}_{\hat{\bar{Q}}^*_{i,t},g_{i,t,s}^*}[Y_i(t) \mid C_{A_i}(t)].
\end{equation*}
\noindent
In the following, we refer to Equation \eqref{eq::psiit} denoting the time- and sample- specific target parameter in order to more easily define the canonical gradient and the first order expansion; our target parameter is defined in Equation \eqref{eq::psit}. The desired target parameter and the subsequent analysis is then defined as an average over samples at time $t$, under the working model $\mathcal{M}^{n}(\bar{O}(t-1))$. We present the canonical gradient, first order expansion and asymptotic normality of the TMLE results in Lemma \ref{lemma::gradient} and Theorem \ref{theorem::normality_fin}. Here, network structure is known at time $t$ (alas, some latent structure might be allowed as in statistical models described in \cite{laan2012networks} and \cite{ogburn2022}). The asymptotic results are in the number of samples and specific to $t$. The corresponding TMLE-CI algorithm under working model $\mathcal{M}^{n}$ is presented in Appendix Algorithm \ref{alg::V2}. In order to address the full trajectory-based parameter, we also present analysis under working model $\mathcal{M}^{tn}(\bar{O}(t-1))$ in Appendix Section \ref{sec:tmle_proofs}. The target parameter defined in Equation \eqref{eq::psi} is appropriate for smoother transitions in testing designs across time and shorter trajectories. 

\begin{lemma}[Canonical gradient and first order expansion]\label{lemma::gradient}
Let $\bar{Q}_{i,t}=\bar{Q}_t$ denote the common across $i$ conditional expectation under working model $\mathcal{M}^{n}$. The canonical gradient of $\Psi_{t,g_t}$ w.r.t. $\mathcal{M}^{n}$ at $P_{\bar{O}(t-1)}$ is
\begin{equation*}
    D_{\bar{O}(t-1)}^*(\bar{Q}_t)(o(t)) =  \frac{1}{n} \sum_{i=1}^n \frac{g_{i,t}^*(A_i(t)| C_{A_i}(t))}{g_{i,t}(A_i(t)| C_{A_i}(t))}(Y_i(t) - {\bar{Q}_t}(A_i(t),C_{A_i}(t))).
\end{equation*}
The time-specific target parameter admits the following first order expansion:
\begin{align*}
    {\Psi}_{t,g_t}(\bar{Q}_t) - {\Psi}_{t,g_t}(\bar{Q}_{0,t}) &= -P_{0,\bar{O}(t-1)} D_{\bar{O}(t-1)}^*(\bar{Q}_t) + R(\bar{Q}_t, \bar{Q}_{0,t}, g_t, g_{0,t}),
\end{align*}
where $R$ is a second order remainder that is doubly-robust, with
$R(\bar{Q}_t, \bar{Q}_{0,t}, g_t, g_{0,t})=0$ if either $\bar{Q}_t = \bar{Q}_{0,t}$ or $g_t = g_{0,t}$.
\end{lemma}

\noindent
Since we are in a randomized trial and the treatment mechanism is known, the second order remainder in Lemma \ref{lemma::gradient} is zero. All further theoretical analysis relies on the fact that the difference between the TML estimator and the target can be decomposed as the average of a centered sequence, as shown in Theorem \ref{theorem::normality_fin} and \ref{theorem::normality_fin2}. 

\begin{theorem}[Asymptotic Normality of the time $t$ TMLE]\label{theorem::normality_fin} Let $\bar{Q}_{0,i,t}=\bar{Q}_{0,t}$ denote the truth and $\hat{\bar{Q}}_{i,t}^*=\hat{\bar{Q}}_{t}^*$ the TMLE which is common across $i$ under working model $\mathcal{M}^{n}$. We denote  
$\psi_{t,g_t} = \frac{1}{n} \sum_{i=1}^{n} {\Psi}_{i,t,g_{i,t}}(\hat{\bar{Q}}_t^*)$ and equivalently $\psi_{0,t,g_t} = \frac{1}{n} \sum_{i=1}^{n} {\Psi}_{i,t,g_{i,t}}(\hat{\bar{Q}}_{0,t})$ . Under weak conditions we have that 
\begin{equation*}
    \psi_{t,g_t} - \psi_{0,t,g_t} \xrightarrow{d} \mathcal{N}(0, \sigma_{t,\infty}^2),
\end{equation*}
where $\sigma_{t,\infty}^2$ is the asymptotic variance of a limit distribution $\lim_{n \rightarrow \infty} \frac{1}{n} \sum_{i=1}^{n} D_{\bar{O}(t-1)}^*(\bar{Q}_{0,t})$.
\end{theorem}

\begin{theorem}[Asymptotic Normality of the average over time TMLE]\label{theorem::normality_fin2} Let $\bar{Q}_{0,i,t}=\bar{Q}_{0}$ denote the truth and $\hat{\bar{Q}}_{i,t}^*=\hat{\bar{Q}}^*$ the TMLE which is common across $t$ and $i$ under working model $\mathcal{M}^{tn}$. We denote  
$\psi^{\tau} = \frac{1}{\tau} \sum_{t=1}^{\tau} {\Psi}_{t,g_{t}}(\hat{\bar{Q}}^*)$ and equivalently $\psi_0^{\tau} =  \frac{1}{\tau} \sum_{t=1}^{\tau} {\Psi}_{t,g_{0,t}}(\bar{Q}_{0})$. Under weak conditions we have that 
\begin{equation*}
    \psi^{\tau} - \psi_0^{\tau} \xrightarrow{d} \mathcal{N}(0, \sigma_{\infty}^2),
\end{equation*}
where $\sigma_{\infty}^2$ is the asymptotic variance of a limit distribution $\lim_{\tau \rightarrow \infty} \frac{1}{\tau} \sum_{t=1}^{\tau} D_{\bar{O}(t-1)}^*(\bar{Q}_{0,t})$.
\end{theorem}

We allocate proof and discussion corresponding to Lemma \ref{lemma::gradient}, Theorem \ref{theorem::normality_fin} and Theorem \ref{theorem::normality_fin2} to the Appendix Section \ref{sec:tmle_proofs}. Let $\sigma_{t,s}^2$ denote the efficient influence based variance for adaptive design $g_{t,s}^*$ at time $t$. We can estimate $\sigma_{t,s}^2$ using the empirical variance estimator as follows
\begin{equation*}
    \hat{\sigma}_{t,s}^2 = \frac{1}{n} \sum_{i=1}^n 
    \left(\frac{g_{t,i}^*(A_i(t)|C_{A_i}(t))}{g_{t,i}(A_i(t)| C_{A_i}(t))}(Y_i(t) - \hat{\bar{Q}}^*_{t}(A_i(t),C_{A_i}(t))) \right)^2,
\end{equation*}
Therefore, each adaptive design $g_{t,s}^*$ has a confidence interval for its overall mean outcome given by 
\begin{equation*}
    \frac{1}{n} \sum_{i=1}^n \int_a \hat{\bar{Q}}^*_{t}(a, C_{A_i}(t)) g_{i,t,s}^*(a \mid C_{A_i}(t)) \pm 1.96 \frac{\hat{\sigma}_{t,s}}{\sqrt{n}}.
\end{equation*}

We can define two different selectors based on the TML estimator. First, denoted as the TMLE-based selector, chooses the design $s_t$ among $s \in S$ that maximizes the point TMLE estimate of the target parameter, such that \begin{equation*}
    s_t = \max_{s} {\Psi}_{t,g^{*}_{t,s}}(\hat{\bar{Q}}_t^*).
\end{equation*}
Alternatively, we take advantage of the asymptotic normality of the TMLE. The second selector, denoted TMLE-CI-based selector, maximizes the lower bound of the confidence interval (CI). In particular, it picks the design $s_t$ with the highest minimum value of the confidence interval
\begin{equation*}
   s_t = \max_{s} [{\Psi}_{t,g^{*}_{t,s}}(\hat{\bar{Q}}_t^*) - 1.96 \sigma_{t,s}/\sqrt{n}]. 
\end{equation*}
Using the TMLE-CI-based selector allows us to incorporate uncertainty in the testing design selection, instead of relying solely on a point estimate. Either way, $s_t$ corresponds to a discrete Online Super Learner selector (analogous to the previous subsection on the loss-based selector). At the next time point, we use the design picked at the previous time point in order to assign tests. Once we observe outcomes for samples tested, we obtain a new TML estimate, and generate a new discrete SL which maximizes either the TMLE point estimate or the lower bound of its CI. We repeat the process until no more data is available. 

\section{Simulations}\label{sec::sim_results}
 
Most higher education institutions faced a difficult decision during the COVID-19 pandemic: reopen and conduct in-person instruction, or face financial challenges and negative social and psychological impacts associated with continued closure. The spread of SARS-CoV-2 in a residential college is particularly hazardous for the broader community due to a large percent of younger, potentially asymptomatic individuals, higher likelihood of shared accommodation, and abundant social contacts \citep{matheson2021}. In the absence of pertinent prior experience, most institutions turned to simulation models and sequential testing in order to track, and contain, the spread of COVID-19. A rich literature on different modeling techniques emerged as a consequence --- resulting in variations of compartmental models, contact networks and agent- or individual-based models \citep{Gressman2020, Rennert2021, Paltiel2020, Martin2020, Lopman2021, Muller2021, Ghaffarzadegan2021, Chang2021}. The interest in effective and safe reopening strategy for an university campus extended across continents \citep{Tuells2021, Hill2021}, campus size \citep{Weeden2020,Bahl2021} and urban settings \citep{Hamer2021}. Other groups resorted to empirical proximity networks of college students in order to simulate and study the spread of the virus \citep{hambridge2021}. 
 
In the following, we illustrate the utility of the proposed Online Super Learner for adaptive sequential surveillance by simulating an environment analogous to the University of California, Berkeley in the Fall of 2020. A detailed description of the agent-based model used to depict transmission of SARS-CoV-2 can be found in Appendix subsection \ref{sec::agent-based-model}. In Appendix subsection \ref{tests}, we elaborate on how all testing strategies (\textit{risk-based} and \textit{rule-based}) can be seen as stochastic interventions with different sampling strategies. We compare different proposed selectors (TMLE-, TMLE-CI- and loss-based) while assessing the state of the infection spread during the length of an academic semester. In terms of testing strategies, we investigate \textit{rule-based} (\textit{symptomatic}, \textit{contact tracing} and \textit{random} testing) and \textit{risk-based} testing with risk estimated via a generalized linear model (denoted as \textit{risk-based} with glm) and an Online Super Learner. All designs are further compared to benchmarks, including \textit{no testing} and when true infectious status is known (``perfect"), corresponding respectively to the lower and upper bounds of performance for any intervention. 

All simulations results represent averages over $250$ Monte Carlo draws and trajectory of $t=120$ time points. While the size of the population is set to $n=20,000$, we investigate performance of the proposed methodology under various resource constraints, testing $1\% - 4\%$ of the total population at each time point (corresponding to $k=\{200,400,600,800\}$). In the Appendix Section \ref{sec::add_sims} we also consider different levels of outside transmission, reflected by the risk scale parameter. Higher risk scale scores correspond to a higher role of the latent parts of the network and individual risk on transmission dynamics. All simulations are initiated with $8$ exposed, $2$ temporarily infectious and $2$ symptomatic cases of COVID-19. We intentionally focus on the scenario where simple rule-based strategies might do well (knowing the network of the few infected individuals), and it is difficult to learn one's risk due to a limited number of infections. We also want to mimic a new start of a semester in an environment with a stable number of daily infections, as otherwise the in-person instruction might be omitted. While we only present results with the $(E=8,It=2,Is=2,Ia=0)$ configuration, other random seeds result in a similar design performance and ranking. 

\subsection{Testing Performance}

We evaluate testing performance by the cumulative incidence curve at each time point and the cumulative percent of infected individuals at $t=120$ (final cumulative incidence). The best performing testing strategy keeps the infection rate low over time, and achieves the lowest cumulative incidence at the end of observation. Testing performance is a function of testing strategy, number of available tests, and the number of currently infected individuals. As such, we evaluate multiple different testing designs (from simple \textit{rule-based} and \textit{risk-based}, to Online SL for adaptive sequential surveillance with the TMLE-, TMLE-CI- and loss-based selector) under different resource constraints ($k=\{200,400,600,800\}$, where $k$ is the number of available tests) across the entire trajectory of the infectious disease progression. 

The average cumulative incidence curves at each time point for all considered designs and available resources are shown in Figure \ref{fig::trajectory}. For instance, the \textit{random} strategy performs as well as no testing when $k=200$, but improves as more tests become available; however, it is always outperformed by competing testing strategies in our simulations, no matter the number of available tests or starting conditions. The \textit{symptomatic + contact} and the \textit{risk-based} strategy with risk estimated with a generalized linear model perform better than \textit{random} and no testing strategies. The Online SL for adaptive surveillance achieves the lowest cumulative incidence of infection (i.e., the best epidemic control) compared to competitors across all times and all selectors. Differences between individual selectors occur at the beginning and end of the trajectory. As can be seen in Figure \ref{fig::trajectory}, the loss-based approach achieves a lower cumulative incidence early in the epidemic trajectory. However, as the epidemic evolves, the TMLE-based selectors achieve the optimal epidemic control (``oracle"). While testing only $1\% - 4\%$ of the total population at each time point, no design achieves performance of the oracle that knows the true infectious status. As more testing resources are allocated, epidemic control improves across all strategies. However, the gains are particularly pronounced for when testing allocation is based on the Online SL for adaptive surveillance. 

The average cumulative incidence by time point $t=120$ is shown in Figure \ref{fig::final}; we refer to it as the average final cumulative incidence. The Online SL for adaptive surveillance with TMLE-CI selector outperforms all competing designs across all simulation setups --- with the most stark difference at $k=800$; confidence interval (CI) for the TMLE-CI selector was $(1.8\%,2.1\%)$ vs. $(2.4\%,2.9\%)$ for the second best design, TMLE-based selector. As expected, the performance of the TMLE-CI selector improves as a function of more available tests (CI with $k=200$: $(14.9\%,16.2\%)$; CI with $k=400$: $(7.1\%,8.0\%)$; CI with $k=600$: $(3.6\%,4.2\%)$; CI with $k=800$: $(1.8\%,2.1\%)$). As such, even with testing only $4\%$ of a large campus population, we can achieve good control of the infectious disease spread in our simulations. Compared to simple \textit{rule-based} and \textit{risk-based} competitors, Online SL for adaptive surveillance achieves much lower cumulative incidence by $t=120$, across all proposed selectors. 
When compared to the \textit{symptomatic + contact} scheme, which might be considered standard best practice, mean final cumulative incidence was reduced from $23.8\%$ to $15.5\%$ at $k=200$, and from $3.6\%$ to $1.9\%$ at $k=800$ with the TMLE-CI selector (on average over 250 trajectories). Figure \ref{fig::final} also shows that TMLE-based and loss-based selector have similar performance under higher $k$ values. 
One possible explanation could be that the advantage of smooth transitions across designs achieved by weights in the TMLE-based selector is offset by averaging loss over a recent window (size 5 in simulations) in the loss-based selector. Ultimately, both TMLE- and loss-based strategies perform worse than the TMLE-CI selector in terms of final cumulative incidence. Finally, as also observed in Figure \ref{fig::trajectory}, all designs perform better than no testing (except for \textit{random} at $k=200$) and worse than the oracle which knows the true infectious status at all values of $k$. 

\subsection{Selected Designs}

Designs used as candidates in the Online SL for adaptive surveillance include a combination of \textit{rule-based} and \textit{risk-based} strategies. Figure \ref{fig::gstars} demonstrates the selected designs (also refereed to as ``discrete Super Learner" or ``cross-validation selector" following notation in Section \ref{sec::proposed_design}) over time and across $500$ simulations using the TMLE-CI-based selector. 

As shown in Figure \ref{fig::gstars}, the \textit{symptomatic + contact} design is often selected at the beginning of the trajectory, while limited data are available to adaptively learn how to test, but as more data become available, it is selected less. All the risk-based strategies include Online Super Learners of $\bar{Q}_{0,i,t}$ with candidate algorithms which reflect working models described in Section \ref{working_models}. In particular, some of the candidates train on the full training history collected, windows of past points (window sizes used: $(7,10,14)$), and exponential weights of the form $(1-\text{rate})^\text{lag}$ where rate is $(0.01,0.05,0.1)$. As shown in Figure \ref{fig::gstars}, designs trained on full history are often picked at the beginning of the trajectory. As time progresses, candidates trained on particular extractions of the past are selected under most of the resource constraints studied, converging to a stable allocation of selected designs. This could be explained by the fact that, as the trajectory progresses, the distant past becomes irrelevant in the current state of epidemic due to the rapidly evolving nature of an infectious disease. Once the selector learns that designs trained on the more recent past are effective, it starts selecting them more consistently. Other working models not shown in Figure \ref{fig::gstars} include various network components, corresponding to different working models for the network dependence. 

The Online SL for adaptive surveillance can also pick among different sampling schemes. As described in Section \ref{tests}, one can allocate tests based on a rank or sample/draw based on the estimated risk. As shown in Figure \ref{fig::gstars}, each risk-based strategy was a candidate design based on both sampling and ranking strategy; hence, when the full training data was used to estimate $\bar{Q}_{0,i,t}$, tests were allocated both based on the top ranked samples and sample proportional to the estimated $\bar{Q}_{0,i,t}$. In terms of sampling schemes, the rank (top) based strategy seems to outperform sampling across all simulation scenarios considered. This could be explained by the fact that the sampling strategy introduces more exploration than necessary for this problem, especially as we achieve a good estimate of $\bar{Q}_{0,i,t}$ with more data over time. While a deterministic sampling strategy with a good estimate of $\bar{Q}_{0,i,t}$ is preferable in our simulations, it is important to keep the sampling option as a candidate design for situations where more exploration is necessary.

\FloatBarrier
\vspace{5mm}
\begin{figure*}[!htbp]
\centering \makeatletter\includegraphics[width=1\linewidth]{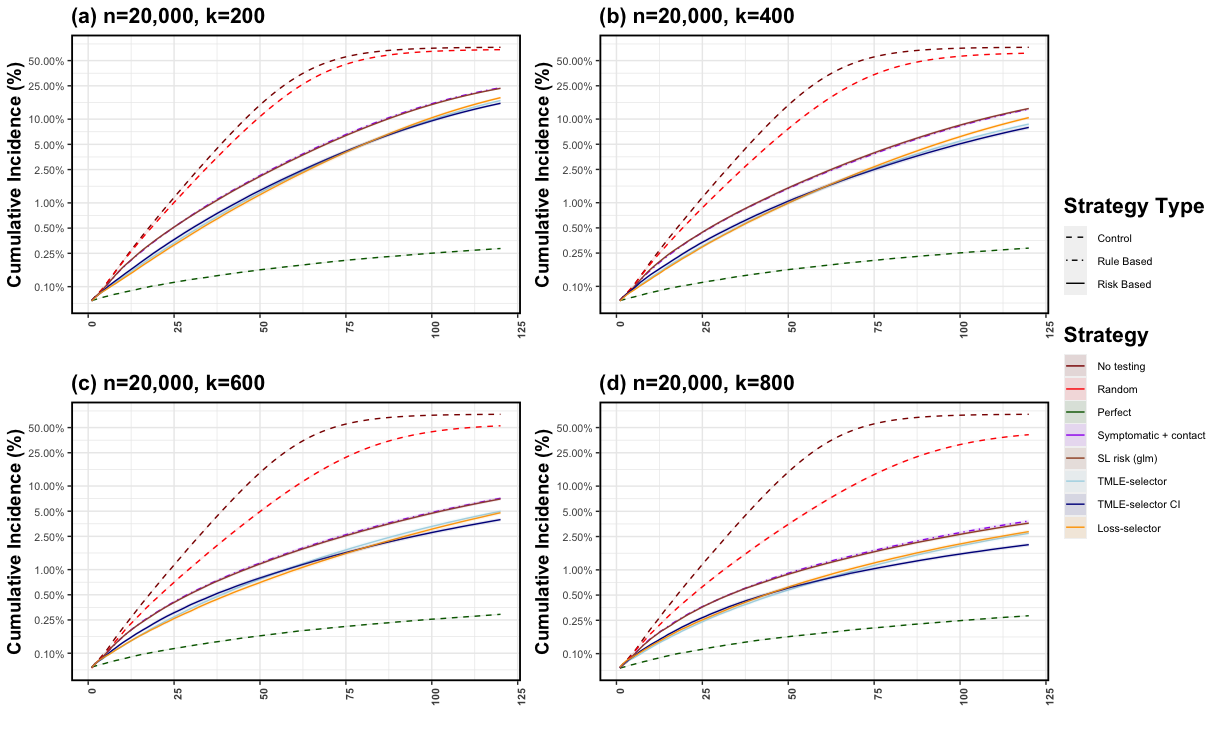}
\makeatother 
\caption{{Average cumulative incidence at each time point over $250$ simulations with $n=20,000$ sample size and testing capacity $k=\{200,400,600,800\}$ using TMLE-based, TMLE-CI-based and loss-based selectors of the testing strategy. We compare different proposed selectors to \textit{symptomatic + contact}, \textit{random} and \textit{glm risk-based} testing, with \textit{perfect} as the upper and \textit{no testing} as the lower bound on performance.}}
\label{fig::trajectory}
\end{figure*}

\FloatBarrier
\vspace{5mm}
\begin{figure*}[!htbp]
\centering \makeatletter\includegraphics[width=1\linewidth]{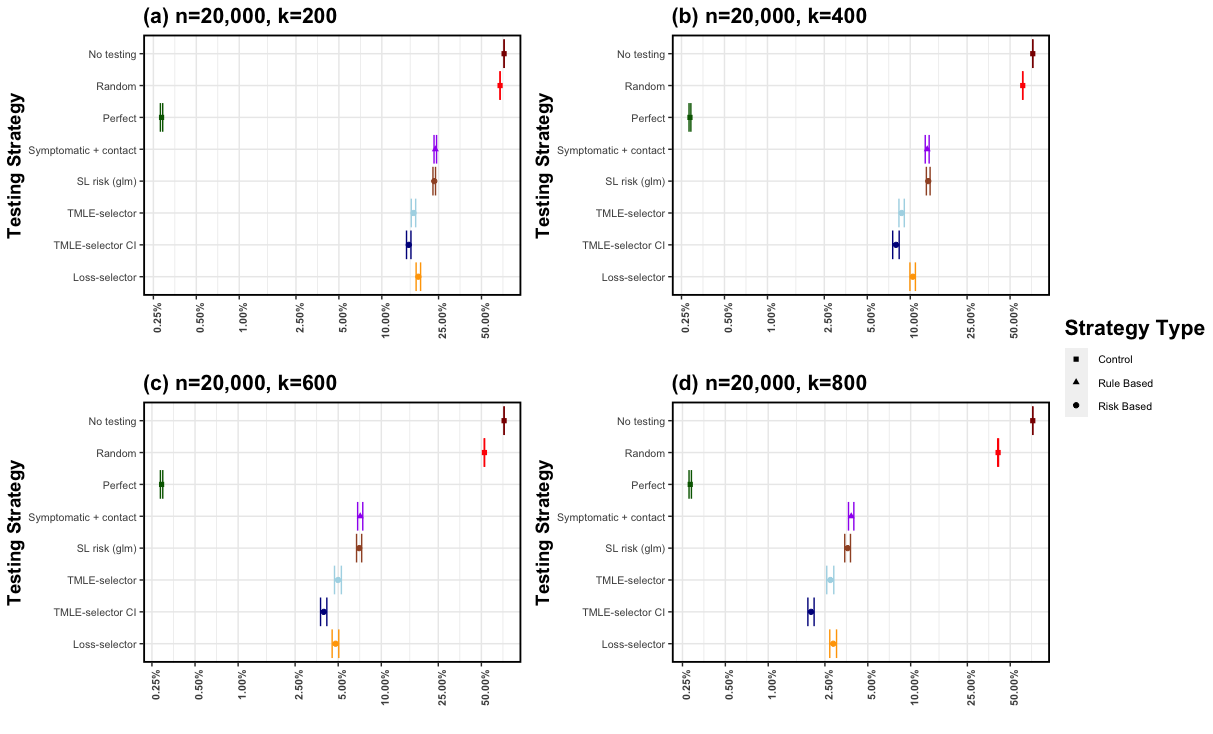}
\makeatother 
\caption{{Average final cumulative incidence at $t=120$ over $250$ simulations, with $n=20,000$ sample size and testing capacity $k=\{200,400,600,800\}$ using TMLE-based, TMLE-CI-based and loss-based selectors. We compare different proposed selectors to \textit{symptomatic + contact}, \textit{random} and \textit{glm risk-based} testing, with \textit{perfect} as the upper and \textit{no testing} as the lower bound on performance.}}
\label{fig::final}
\end{figure*}
\vspace{50mm}

\FloatBarrier
\vspace{5mm}
\begin{figure*}[!htbp]
\centering \makeatletter\includegraphics[width=1\linewidth]{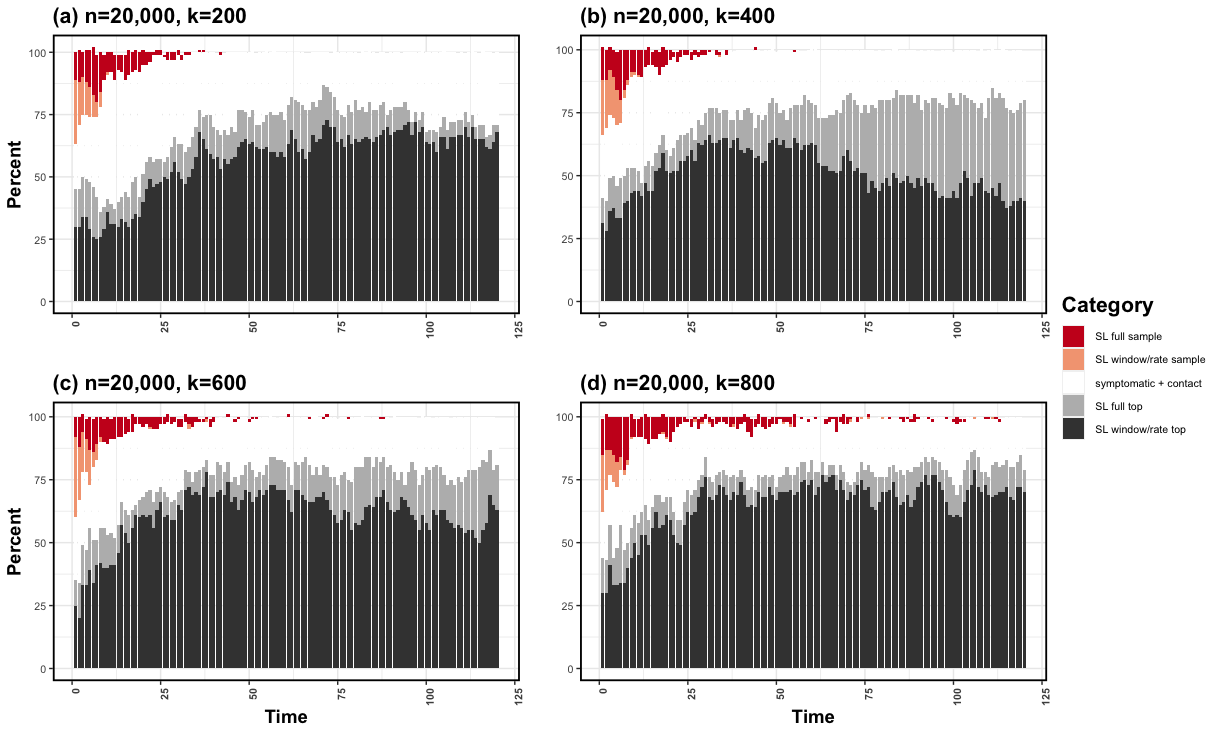}
\makeatother 
\caption{{Percent of times a given candidate testing allocation design is selected over the full trajectory and $250$ simulations with $n=20,000$ sample size and testing capacity $k=\{200,400,600,800\}$ using the TMLE-CI-based selector. Candidate designs include \textit{symptomatic + contact} and various \textit{risk-based} designs where the Super Learner is either trained on the full past, exponentially weighted past, or a window of $t=14$ days. All \textit{risk-based} testing strategies also consider different sampling schemes, including picking the top samples for testing or randomly sampling based on the estimated risk.}}
\label{fig::gstars}
\end{figure*}

\section{Discussion}\label{sec::discussion}

In this work, we develop an Online Super Learner for the adaptive sequential design under an unknown dependence structure. Our proposed method is particularly suited for infectious disease surveillance and control,
and generalizes immediately to any adaptive sequential problem with unknown dependence (across individuals and time) within a fully nonparameteric model. The data structure constitutes a typical longitudinal structure of $n$ individuals over a period of $\tau$ time points. Within each $t$-specific time block, one observes the exposure variable (e.g., indicator of testing), time-varying covariates (e.g., network structure, health status) and outcome (e.g., infectious status) for all $n$ individuals. 
Our causal target parameter is the mean outcome we would have obtained after one time-step, if, starting at time $t$ given the observed past, we had carried out a stochastic intervention $g_{t}^*$. The main goal is to optimize the next time-point outcome under $g_{t}^*$ at each $t$, or as an average of short term outcomes over time, under a possible resource constraint. As such, the history-adjusted optimal choice for a single time point intervention defines the adaptive design over time. In the setting of an infectious disease outbreak, we define exposures of interest as user-defined stochastic interventions, where each $g_{t}^*$ denotes a specific testing design (\textit{symptomatic}, \textit{random}, \textit{contact tracing}, \textit{risk-based} testing, etc). The proposed Online Super Learner for adaptive sequential surveillance then learns the optimal choice of test strategies over time, adapting to the current state of the epidemic. In the application considered here, the optimal testing allocation aims to maximize the number of infectious individuals identified, allowing for prompt isolation and prevention of further spread. 

The infectious disease context, however, presents unique technical challenges. For instance, our target parameter is defined in terms of a latent outcome, as the true infectious status is unknown (unless a person is tested). We present an identification result for the causal target parameter in terms of the observed outcome, defined as a function of the stochastic intervention. In addition, unlike the usual i.i.d. settings, infectious disease propagation induces both network and temporal dependence. The adaptive sequential designs described in the literature typically exploit asymptotics in the number of subjects enrolled in the trial \citep{chambaz2017}, or in the number of time points \citep{malenica2021adaptive}. Unlike these settings, the presence of both network and temporal dependence reduces data to a single observation. Instead of imposing unrealistic assumptions on the statistical model $\mathcal{M}$, we rely on working models in order to estimate the conditional mean function $\bar{Q}_{0,i,t}$ and use an honest benchmark to choose the best performing estimate for the sake of the adaptive design performance. Therefore, the proposed method decides whether to learn across samples, through time, or both, based on the underlying (unknown) structure in the data at each time point of the disease trajectory. 

As part of the Online SL for adaptive sequential surveillance, we propose a method for selecting among different adaptive designs. Namely at each time $t$, we evaluate the performance of a choice $g^*_t$ by proportion of infected individuals detected. We might use as criterion for an adaptive design its average loss over a recent time window (loss-based selector), the TMLE estimate under $g^*_t$ and working model (TMLE-based), or the maximum lower confidence interval under the stochastic intervention (TMLE-CI-based). All of the evaluation strategies aim to provide smooth transitions from discrete SL at time $t$ to the one at $t+1$, by either averaging over a window of recent performances or weighting by the ratio of current and proposed probability of treatment given the past. Therefore, the Online SL for adaptive sequential surveillance is an adaptive design itself, and the discrete Super Learner evolves and changes as a function of the underlying disease dynamics (which is unspecified and often unknown).
The key strength of the proposed method is that it does not depend on a strong statistical model, or imposes unrealistic assumptions. Instead, it relies on different working models to estimate $\bar{Q}_{0,i,t}$, and selects among adaptive designs with a short term performance Online Super Learner. As such, the proposed adaptive design avoids stationarity and independence assumptions, which are unrealistic in an infectious disease setup.

In addition to proposing a new adaptive sequential design suitable for studying infectious disease, we have also developed an agent-based model for a moderate size campus during an epidemic. Here, the model was parameterized to resemble the University of California (Berkeley). However, all the settings and simulations can be easily modified to reflect any residential campus and infectious disease. Within the simulation framework defined by the agent-based model, the Online SL for adaptive sequential surveillance outperforms all considered gold standard testing schemes (\textit{random}, \textit{contact tracing + symptomatic}) including the \textit{risk-based} testing alone. The advantages of the proposed adaptive design are evident over a variety of scenarios, including varying resource constraints and level of problem difficulty (determined by the percent latent component of the network and individual risk). 
As response to COVID-19 pandemic evolved, most universities started to require mandatory vaccination, mask-wearing, social distancing, enhanced cleaning protocols, increased availability of sanitizing products and no large social gatherings. In future work, we plan to investigate performance of the proposed Online SL for adaptive surveillance in addition to a wide array of other transmission-limiting measures.

Our proposed method can be extended in various ways. Instead of using a single selector as done in our simulations, we could instead have an Online SL for adaptive surveillance where each candidate is one of the described methods. As such, it would be possible to pick among TMLE-, TMLE-CI- and loss-based selectors at any time point $t$. This strategy would be particularly advantageous at the beginning of the trajectory, when loss-based methods seems to perform best, but TMLE-CI-based selector minimizes the final cumulative incidence. In addition, we could extend the proposed methodology to consider a convex combination of different designs, instead of focusing on the discrete Online SL. While this extension could provide better performance (in terms of optimizing our target parameter), it might be more difficult to interpret and implement in practice. In addition, while the Online SL for adaptive surveillance outperforms all considered testing schemes, it does not provides inference for our main target parameter. If we were willing to make additional assumptions, or known more about the data generating process  --- for example, assume a known network structure over time or conduct detailed surveillance as done in some countries --- we could analyze the TMLE of our target parameter under one of the working models discussed in Section \ref{working_models}. Alternatively, one could data-adaptively learn the underlying true model by giving up certain statistical properties of the estimator, such as regularity. We intend to explore all of these interesting avenues and extensions in future work.



\section*{Acknowledgments}\label{sec::acknowledgments}

This work was supported by the National Institute of Allergy and Infectious Diseases (NIH R01 AI074345).

\newpage
\section{Appendix}\label{Appendix}
\subsection{Identifiability Results}\label{sec:identify}

\textbf{Theorem 1} \textit{Assume assumptions \ref{randomization} and \ref{positivity} hold. Under consistency, we denote the time $t$ value under the stochastic intervention $g^*_{t}$ as
\begin{align*}
    {\Psi}_{t,g^{*}_{t}}^F(P_{\bar{O}(t-1)}^F) = {\Psi}_{t,g^{*}_{t}}(P_{\bar{O}(t-1)}) &= \int_{a} \frac{1}{n}\sum_{i=1}^n \mathbb{E}_P [{Y}_i^l(t) \mid a, \bar{o}(t-1)] g_{i,t}^*(a \mid \bar{o}(t-1)) d\mu_a(a) \\ \nonumber &= \frac{1}{n}\sum_{i=1}^n \mathbb{E}_{\bar{Q}_{i,t},g_{i,t}^*}[Y_i(t) \mid \bar{o}(t-1)] 
\end{align*}
where the observed outcome is defined as $Y_i(t) = A_i(t)Y_i^l(t)/g_{i,t}(A_i(t) \mid \bar{O}(t-1))$ and $\psi_{t,g^{*}_{t}} = {\Psi}_{t,g^{*}_{t}}(P_{\bar{O}(t-1)})$.}

\vspace{3mm}
\begin{proof}
First, we identify the causal parameter in terms of the conditional distribution of the observed data $P_{\bar{O}(t-1)}$ and latent outcome $Y^l(t)$. 
Under Assumptions \ref{randomization} (A6) and \ref{positivity} (A7), jointly with consistency (A8), we can denote value at $t$ under the stochastic intervention $g^*_{t}$ as
\begin{align}
    &\mathbb{E}\left[ \frac{1}{n} \sum_{i=1}^n Y_{i,g_{i,t}^*}^l(t)\right] = \mathbb{E}[ \bar{Y}_{g_{t}^*}^l(t)] \\ \nonumber
    &\stackrel{def}= \int_{a} \frac{1}{n}\sum_{i=1}^n \mathbb{E} [{Y}_{i,g_{i,t}^*}^l(t) = y^l \mid A_{i,g_{i,t}^*}(t) = a, \bar{o}(t-1)] g_{i,t}^*(a \mid \bar{o}(t-1)) d\mu(a) \\ \nonumber
    &\stackrel{def}= \int_{a} \frac{1}{n}\sum_{i=1}^n \mathbb{E} [{Y}_{i,a}^l(t) = y^l \mid A_{i,g_{i,t}^*}(t) = a, \bar{o}(t-1)] g_{i,t}^*(a \mid \bar{o}(t-1)) d\mu(a) \\ \nonumber
    &\stackrel{A6}= \int_{a} \frac{1}{n}\sum_{i=1}^n \mathbb{E} [{Y}_{i,a}^l(t) = y^l \mid \bar{o}(t-1)] g_{i,t}^*(a \mid \bar{o}(t-1)) d\mu(a) \\ \nonumber
    &\stackrel{A8}= \int_{a} \frac{1}{n}\sum_{i=1}^n \mathbb{E} [{Y}_i^l(t) = y^l \mid A_i(t) = a, \bar{o}(t-1)] g_{i,t}^*(a \mid \bar{o}(t-1)) d\mu(a) \\ \nonumber
    &\stackrel{def}= \int_{a} \frac{1}{n}\sum_{i=1}^n \mathbb{E} [{Y}_i^l(t) = y^l \mid A_i(t) = a, \bar{o}(t-1)] g_{i,t}^*(a \mid \bar{o}(t-1)) d\mu(a). \nonumber
\end{align}

\noindent
Note that, for conditional expectations to be well defined, Assumption \ref{positivity} must hold. The last equality in above expression gives us the identification result in terms of the conditional expectation of the latent outcome. We proceed to define the observed outcome as
\begin{equation}\label{eq::observedY}
    Y_i(t) = A_i(t)Y_i^l(t)/g_{i,t}(A_i(t) \mid \bar{O}(t-1)), 
\end{equation}
with the conditional expectation of the observed outcome as follows
\begin{align*}
    &\mathbb{E}[\bar{Y}(t) \mid \bar{o}(t-1)] =
    \frac{1}{n}\sum_{i=1}^n \mathbb{E}\left[ \frac{A_i(t)}{g_{i,t}(A_i(t) \mid \bar{o}(t-1))} Y_i^l(t) \mid \bar{o}(t-1) \right] \\
    &= \frac{1}{n}\sum_{i=1}^n \int_y \frac{y}{g_{i,t}(A_i(t) \mid \bar{o}(t-1))} P(Y_i^l(t) = y \mid 1, \bar{o}(t-1)) g_{i,t}(A_i(t) \mid \bar{o}(t-1)) d\mu(y) \\
    &=  \frac{1}{n}\sum_{i=1}^n \mathbb{E}_{g_{i,t}}[Y_i^l(t) \mid 1, \bar{O}(t-1)=\bar{o}(t-1)].
\end{align*}
Therefore, the conditional expectation of the observed outcome defined as in equation \eqref{eq::observedY} is equal to the conditional expectation of the latent outcome $Y_i^l(t)$. We can write the final identification results as
\begin{align*}
    &\mathbb{E}\left[ \frac{1}{n}\sum_{i=1}^n Y_{i,g_{i,t}^*}^l(t)\right] = \int_{a} \frac{1}{n}\sum_{i=1}^n \mathbb{E} [{Y}_i^l(t) = y^l \mid A_i(t) = a, \bar{o}(t-1)] g_t^*(a \mid \bar{o}(t-1)) d\mu(a) \\
    &= \int_{a} \frac{1}{n}\sum_{i=1}^n \mathbb{E} [\frac{A_i(t)}{g_{i,t}(A_i(t) \mid \bar{o}(t-1))} y^l \mid A_i(t) = a, \bar{o}(t-1)] g_{i,t}^*(a \mid \bar{o}(t-1)) d\mu(a) \\
    &= \frac{1}{n}\sum_{i=1}^n \mathbb{E}_{\bar{Q}_{i,t}, g_{i,t}^*}[Y_i(t) \mid  \bar{o}(t-1)].
\end{align*}

\end{proof}

\subsection{TMLE Results}\label{sec:tmle_proofs}

In the following, at times, it proves useful to use notation from empirical process theory. Specifically, we define $Pf$ to be the empirical average of the function $f$ w.r.t.~the distribution $P$, that is, $Pf = \int f(o)dP(o)$.
In order to alleviate notation, we define a following centered process $(M_{t}(f) : f)$ for a function $f$ in a class of multivariate real valued functions as follows:
\begin{align*}
    M_{t}(f) &= D_{\bar{O}(t-1)}^*(f)(O(t)) - \mathbb{E}_{0}{D_{\bar{O}(t-1)}^*(f)}.
\end{align*}

\subsubsection{Working model $\mathcal{M}^{n}$}

\textbf{Lemma 1} \textit{Let $\bar{Q}_{i,t}=\bar{Q}_t$ denote the common across $i$ conditional expectation under working model $\mathcal{M}^{n}$. The canonical gradient of $\Psi_{t,g_t}$ w.r.t. $\mathcal{M}^{n}$ at $P_{\bar{O}(t-1)}$ is
\begin{equation*}
    D_{\bar{O}(t-1)}^*(\bar{Q}_t)(o(t)) =  \frac{1}{n} \sum_{i=1}^n \frac{g_{i,t}^*(A_i(t)| C_{A_i}(t))}{g_{i,t}(A_i(t)| C_{A_i}(t))}(Y_i(t) - {\bar{Q}_t}(A_i(t),C_{A_i}(t))).
\end{equation*}
The time-specific target parameter admits the following first order expansion:
\begin{align*}
    {\Psi}_{t,g_t}(\bar{Q}_t) - {\Psi}_{t,g_t}(\bar{Q}_{0,t}) &= -P_{0,\bar{O}(t-1)} D_{\bar{O}(t-1)}^*(\bar{Q}_t) + R(\bar{Q}_t, \bar{Q}_{0,t}, g_t, g_{0,t}),
\end{align*}
where $R$ is a second order remainder that is doubly-robust, with
$R(\bar{Q}_t, \bar{Q}_{0,t}, g_t, g_{0,t})=0$ if either $\bar{Q}_t = \bar{Q}_{0,t}$ or $g_t = g_{0,t}$.}

\begin{proof}
\vspace{2mm}
\noindent
Notes on the expansion: Recall that von Mises expansion results in the following approximation of the estimation error at time $t$:
\begin{align*}
\frac{{\Psi}_{t,g_{0,t}}(\bar{Q}_{0,t})-{\Psi}_{t,g_{t}}(\bar{Q}_t)}{\delta \epsilon} &\approx \mathbb{E}_{P_{\bar{O}(t-1)}}[D_{\bar{O}(t-1)}^*(\bar{Q}_t) s] \\ 
    {\Psi}_{t,g_{t}}(\bar{Q}_t) - {\Psi}_{t,g_{0,t}}(\bar{Q}_{0,t}) &\approx -\mathbb{E}_{P_{\bar{O}(t-1)}}[D_{\bar{O}(t-1)}^*(\bar{Q}_t) s] (1-0) \\
    &= -P_{0,\bar{O}(t-1)} D_{\bar{O}(t-1)}^*(\bar{Q}_t) + R
\end{align*}
where we define $R$ as the second-order remainder $R = ({\Psi}_{t,g_{t}}(\bar{Q}_t) - {\Psi}_{t,g_{0,t}}(\bar{Q}_{0,t})) + P_{0,\bar{O}(t-1)} D_{\bar{O}(t-1)}^*(\bar{Q}_t)$, and the second to last equality is due to the change-of-variables formula and the definition of the score $s$ along a path
\begin{align*}
    \int D_{\bar{O}(t-1)}^*(\bar{Q}_t) \ s \ dP_{\bar{O}(t-1)} = \int D_{\bar{O}(t-1)}^*(\bar{Q}_t) dP_{0,\bar{O}(t-1)} = P_{0,\bar{O}(t-1)} D_{\bar{O}(t-1)}^*(\bar{Q}_t).
\end{align*}

Let $\{P_{\Bar{O}(t-1),\epsilon}^h : \epsilon \in (-\delta, \delta)  \subset \mathcal{M}^{n} \}$ be a class of one-dimensional parametric models indexed by a direction $h$, such that $P_{\Bar{O}(t-1),\epsilon=0}^h = P_{0,\Bar{O}(t-1)}$. Let $s^h = \left.\frac{\delta}{\delta \epsilon} \log P_{\Bar{O}(t-1),\epsilon}^h(o)\right\vert_{\epsilon =0}$ be the corresponding score function of the $h$-specific path. Then the pathwise derivative of parameter $\Psi_{t,g_t}$ at $P_{\Bar{O}(t-1)}$ along the path defined by $P_{\Bar{O}(t-1),\epsilon}^h$ is given by $\frac{\delta}{\delta \epsilon} \Psi(P_{\Bar{O}(t-1),\epsilon}^h)$ at $\epsilon=0$. By the Riesz representation theorem we have that $\left.\frac{\delta}{\delta \epsilon} \Psi(P_{\Bar{O}(t-1),\epsilon}^h)\right\vert_{\epsilon =0} = \mathbb{E}[\phi(o) s^h(o)]$ for some zero mean function $\phi(o)$. In the following, the direction index is implied in the notation and we write $P_{\Bar{O}(t-1),\epsilon}^h = P_{\Bar{O}(t-1),\epsilon}$. A common choice of a submodel is, for some nonzero function $s : \mathcal{O} \rightarrow \mathbb{R}$, $P_{\Bar{O}(t-1),\epsilon}(o) = (1 + \epsilon s)P_{\Bar{O}(t-1)}(o)$. Note that for this submodel the score function is $\left.\frac{\delta}{\delta \epsilon} \log P_{\Bar{O}(t-1),\epsilon}(o)\right\vert_{\epsilon =0} =  s(o)$. It then follows that, 
\begin{align*}
    \left.\frac{\delta \Psi_{t,g_{t}^*}(P_{\Bar{O}(t-1),\epsilon})}{\delta \epsilon}\right\vert_{\epsilon =0} &= \left.\int y \frac{\delta}{\delta \epsilon} \prod_{i} q_{t,\epsilon} \right\vert_{\epsilon =0} \prod_{i} g_{i,t}^* \\
    &= \left.\int y \frac{\delta}{\delta \epsilon}
    \prod_{i}(1 + \epsilon s)q_t \right\vert_{\epsilon =0} \prod_{i} g_{i,t}^* \\
    &= \left.\int y \prod_{i} (q_t g_{i,t}^*) \frac{\delta}{\delta \epsilon} \log \prod_{i} q_{t,\epsilon} \right\vert_{\epsilon =0}\\
    &= \left.\int \prod_{i} (q_t g_{i,t}) \left[ \prod_{i} \frac{g_{i,t}^*}{g_{i,t}} y \right] \frac{\delta}{\delta \epsilon} \log \prod_{i} q_{t,\epsilon} \right\vert_{\epsilon =0} \\
    &= \left.\int \prod_{i} (q_t g_{i,t}) \left[\prod_{i} \frac{g_{i,t}^*}{g_{i,t}} y - \Psi_{i,t,g_{i,t}}(q_t) \right] \frac{\delta}{\delta \epsilon} \log \prod_{i} q_{t,\epsilon} \right\vert_{\epsilon =0} \\
    &= \mathbb{E}_{q_t,g}( \left[\prod_{i} \frac{g_{i,t}^*}{g_{i,t}} y - \Psi_{i,t,g_{i,t}}(q_t) \right] \sum_i s).
\end{align*}
\end{proof}

It's interesting to point out that the asymptotic behavior of the estimator is based on a single draw as $n \rightarrow \infty$. The asymptotic variance is still characterized by the EIF, as is the case for $n$ i.i.d. observations and a regular estimator. In the following, we provide a sketch proof for asymptotic normality of the proposed TMLE estimator.

\vspace{2mm}
\noindent
\textbf{Theorem 2} \textit{Let $\bar{Q}_{0,i,t}=\bar{Q}_{0,t}$ denote the truth and $\hat{\bar{Q}}_{i,t}^*=\hat{\bar{Q}}_{t}^*$ the TMLE which is common across $i$ under working model $\mathcal{M}^{n}$. We denote  
$\psi_{t,g_t} = \frac{1}{n} \sum_{i=1}^{n} {\Psi}_{i,t,g_{i,t}}(\hat{\bar{Q}}_t^*)$ and equivalently $\psi_{0,t,g_t} = \frac{1}{n} \sum_{i=1}^{n} {\Psi}_{i,t,g_{i,t}}(\hat{\bar{Q}}_{0,t})$ . Under weak conditions we have that 
\begin{equation*}
    \psi_{t,g_t} - \psi_{0,t,g_t} \xrightarrow{d} \mathcal{N}(0, \sigma_{t,\infty}^2),
\end{equation*}
where $\sigma_{t,\infty}^2$ is the asymptotic variance of a limit distribution $\lim_{n \rightarrow \infty} \frac{1}{n} \sum_{i=1}^{n} D_{\bar{O}(t-1)}^*(\bar{Q}_{0,t})$.}

\begin{proof}
By Lemma 1 we have the following first order expansion for the time-specific parameter which takes into account all samples at time $t$:
\begin{align}\label{eqn::expansion}
    {\Psi}_{t,g_{t}}(\hat{\bar{Q}}^*_{t}) - {\Psi}_{t,g_{t}}(\bar{Q}_{0,t}) = -P_{0,\bar{O}(t-1)} D_{\bar{O}(t-1)}^*(\hat{\bar{Q}}^*_{t}) + R(\hat{\bar{Q}}^*_{t}, \bar{Q}_{0,t}, g_{t}, g_{0,t}), 
\end{align}
where $\hat{\bar{Q}}^*_{t}$ is the TMLE at time $t$. Due to the randomized setting, the second order remainder is $o_P(1/\sqrt{n})$, and thus negligible. Note that $D_{\bar{O}(t-1)}^*(\hat{\bar{Q}}^*_{t})(O(t)) = \frac{1}{n}\sum_{i}^n D_{\bar{O}(t-1)}^*(\hat{\bar{Q}}^*_{t})(O_i(t)) = 0$, so we can write the Equation \eqref{eqn::expansion} as
\begin{align*}
    {\Psi}_{t,g_{t}}(\hat{\bar{Q}}^*_{t}) - {\Psi}_{t,g_{t}}(\bar{Q}_{0,t}) &= \frac{1}{n}\sum_{i}^n D_{\bar{O}(t-1)}^*(\hat{\bar{Q}}^*_{t})(O_i(t)) -P_{0,\bar{O}(t-1)} D_{\bar{O}(t-1)}^*(\hat{\bar{Q}}^*_{t}) + o_P(1/\sqrt{n}) \\
    &= \frac{1}{n}\sum_{i}^n D_{\bar{O}(t-1)}^*(\bar{Q}_{0,t})(O_i(t)) -P_{0,\bar{O}(t-1)} D_{\bar{O}(t-1)}^*(\bar{Q}_{0,t}) \\ 
    &+ \frac{1}{n}\sum_{i}^n D_{\bar{O}(t-1)}^*(\hat{\bar{Q}}^*_{t})(O_i(t)) -P_{0,\bar{O}(t-1)} D_{\bar{O}(t-1)}^*(\hat{\bar{Q}}^*_{t}) \\
    &- \frac{1}{n}\sum_{i}^n D_{\bar{O}(t-1)}^*(\bar{Q}_{0,t})(O_i(t)) -P_{0,\bar{O}(t-1)} D_{\bar{O}(t-1)}^*(\bar{Q}_{0,t}) + o_P(1/\sqrt{n}) \\
    &= M_f(\bar{Q}_{0,t}) + M_f(\hat{\bar{Q}}^*_{t}) - M_f(\bar{Q}_{0,t}) + o_P(1/\sqrt{n}).
\end{align*} 

Under conditions outlined in Theorem 5 of \cite{laan2012networks}, we can establish asymptotic normality of the above decomposition. As opposed to Theorem 5 of \cite{laan2012networks}, we don't need an orthogonal decomposition of the centered difference, as our target parameter is conditional on the fixed dimensional summary of the past (not just on the known network). Appendix Section B2-4 in \cite{laan2012networks} establish asymptotic equicontinuity results for a process akin to $M_f(\hat{\bar{Q}}^*_{t},\bar{Q}_{0,t})$. Therefore, $M_{t}(\hat{\bar{Q}}^*_{t}) - M_{t}(\bar{Q}_{0,t}) = o_P(1)$, so that $M_{t}(\hat{\bar{Q}}^*_{t})$ behaves as $M_{t}(\bar{Q}_{0,t})$. It remains to investigate the weak convergence of $M_{t}(\bar{Q}_{0,t})$ as $n \rightarrow \infty$, which follows from CLT in our setting.
\end{proof}

\subsubsection{Working model $\mathcal{M}^{tn}$}

\textbf{Lemma 2} \textit{Let $\bar{Q}_{i,t}=\bar{Q}$ denote the common across $i$ and $t$ conditional expectation under working model $\mathcal{M}^{tn}$. The canonical gradient of $\Psi$ w.r.t. $\mathcal{M}^{tn}$ at $P$ is given by 
\begin{equation*}
   D^*(\bar{Q})(o) = \frac{1}{n \tau} \sum_{i=1}^n \sum_{t=1}^{\tau} D_{\bar{O}_i(t-1)}^*(\bar{Q})(o_{i}(t)),
\end{equation*}
where
\begin{equation*}
    D_{\bar{O}(t-1)}^*(\bar{Q})(o(t)) =  \frac{1}{n} \sum_{i=1}^n \frac{g_{i,t}^*(A_i(t)| C_{A_i}(t))}{g_{i,t}(A_i(t)| C_{A_i}(t))}(Y_i(t) - {\bar{Q}}(A_i(t),C_{A_i}(t))).
\end{equation*}
The average over time target parameter admits the following first order expansion:
\begin{align*}
    {\Psi}_{\tau}(\bar{Q}) - {\Psi}_{\tau}(\bar{Q}_{0}) &= -P_{0} D^*(\bar{Q}) + R(\bar{Q}, \bar{Q}_{0}, g, g_{0}) \\
    &= \frac{1}{\tau} \sum_{t=1}^{\tau} (- P_{0, \bar{O}(t-1)} D_{\bar{O}(t-1)}^*(\bar{Q}) + R(\bar{Q}, \bar{Q}_{0}, g_{t}, g_{0,t})),
\end{align*}
where $R$ is a second order remainder that is doubly-robust, with
$R(\bar{Q}, \bar{Q}_{0}, g, g_{0})=0$ if either $\bar{Q} = \bar{Q}_{0}$ or $g = g_{0}$.}

\begin{proof}
\vspace{2mm}
\noindent
The argument follows the same derivation as for Lemma 1. 
\end{proof}



\noindent
\textbf{Theorem 3} \textit{Let $\bar{Q}_{0,i,t}=\bar{Q}_{0}$ denote the truth and $\hat{\bar{Q}}_{i,t}^*=\hat{\bar{Q}}^*$ the TMLE which is common across $t$ and $i$ under working model $\mathcal{M}^{tn}$. We denote  
$\psi^{\tau} = \frac{1}{\tau} \sum_{t=1}^{\tau} {\Psi}_{t,g_{t}}(\hat{\bar{Q}}^*)$ and equivalently $\psi_0^{\tau} =  \frac{1}{\tau} \sum_{t=1}^{\tau} {\Psi}_{t,g_{0,t}}(\bar{Q}_{0})$. Under weak conditions we have that 
\begin{equation*}
    \psi^{\tau} - \psi_0^{\tau} \xrightarrow{d} \mathcal{N}(0, \sigma_{\infty}^2),
\end{equation*}
where $\sigma_{\infty}^2$ is the asymptotic variance of a limit distribution $\lim_{\tau \rightarrow \infty} \frac{1}{\tau} \sum_{t=1}^{\tau} D_{\bar{O}(t-1)}^*(\bar{Q}_{0,t})$.}

\begin{proof}
Note the following TMLE expansion
\begin{align*}
    \psi^{\tau} - \psi_0^{\tau} = T_{1,\tau,n}(\bar{Q}_{0}) + T_{2,\tau,n}(\hat{\bar{Q}}^*,\bar{Q}_{0}) + R(\hat{\bar{Q}}^*, \bar{Q}_{0}), 
\end{align*}
where 
\begin{align*}
    T_{1,\tau,n}(\bar{Q}_{0}) &= \frac{1}{\tau} \sum_{t=1}^{\tau} D_{\bar{O}(t-1)}^*(\bar{Q}_{0,t})(O(t)) - \mathbb{E}_{0} D_{\bar{O}(t-1)}^*(\bar{Q}_{0,t}) \\
    T_{2,\tau,n}(\hat{\bar{Q}}^*,\bar{Q}_{0}) &= \frac{1}{\tau} \sum_{t=1}^{\tau} D_{\bar{O}(t-1)}^*(\hat{\bar{Q}}^*_{t})(O(t)) - \mathbb{E}_{0} D_{\bar{O}(t-1)}^*(\hat{\bar{Q}}^*_{t}) \\
    &- \frac{1}{\tau} \sum_{t=1}^{\tau} D_{\bar{O}(t-1)}^*(\bar{Q}_{0,t})(O(t)) + \mathbb{E}_{0} D_{\bar{O}(t-1)}^*(\bar{Q}_{0,t}),
\end{align*}
and $R(\hat{\bar{Q}}^*, \bar{Q}_{0})$ is the second order remainder that's robust in the usual sense and negligible due to the treatment process being known. The first term $T_{1,\tau,n}(\bar{Q}_{0})$ is a sum of a martingale difference sequence, and the process $\sqrt{n\tau} T_{1,\tau,n}(\bar{Q}_{0})$ converges weakly to a Wiener process by the CLT for martingale triangular arrays. A set of sufficient conditions for the asymptotic normality of the term $T_{1,\tau,n}(\bar{Q}_{0})$ is that (a) the terms $D_{\bar{O}(t-1)}^*(\bar{Q}_{0,t})(O(t))$ remain bounded, and (b) the average of the conditional variances of $D_{\bar{O}(t-1)}^*(\bar{Q}_{0,t})(O(t))$ stabilize. In particular, we denote $\sigma_{\infty}^2$ as the variance under a limit distribution $\lim_{\tau \rightarrow \infty} \frac{1}{\tau} \sum_{t=1}^{\tau} D_{\bar{O}(t-1)}^*(\bar{Q}_{0,t})$. The second term can be bounded by the supremum of the process $\{T_{2,\tau,n}(\hat{\bar{Q}}^*,\bar{Q}_{0}) : \hat{\bar{Q}}^* \in \mathcal{M_{\bar{Q}}} \}$, which is an empirical process generated by the sequence of contexts and past data. To analyze $T_{2,\tau,n}(\hat{\bar{Q}}^*,\bar{Q}_{0})$, we rely on a maximal inequality in \cite{handel2009} under mixing and entropy conditions. For the exact equicontinuity result, we refer the interested reader to Theorem 4 in \cite{malenica2021adaptive} 
\end{proof}

\subsection{Comment on the Oracle Adaptive Sequential Design}

The oracle for the proposed adaptive design aims to maximize the overall number of detected cases at each time step under a resource constraint, such that
\begin{align}\label{eqn::objective}
    \arg \max_{g^{*}_{t} \in \mathcal{G}} ({\Psi}_{t,g^{*}_{t}}(P_{\bar{O}(t-1)})) \ \ \text{subject to} \ E[g^*_{t}(1 \mid Pa(A(t)))] \leq k \ \text{for all } t.
\end{align} 
The proposed Online Super Learner for adaptive sequential surveillance aims to approximate and learn the oracle design in Equation \eqref{eqn::objective}.
Since we are optimizing a short term outcome, this equates to optimizing the mean over time parameter under a resource constraint at each $t$. The optimal strategy $g^{opt}_{t}$ then corresponds to the intervention $g^*_{t}$ that maximizes Equation \eqref{eqn::objective} at $t$ under the $k$ percent constraint. Under the $t$-optimal testing allocation, contagious individuals can be quickly isolated from the general population with the ultimate goal of minimizing transmission at future time-points by prompt detection of active, circulating infections. Consequently, maximizing the number of caught infected individuals at each time step results in minimizing the total number of circulating infectious by time $\tau$. As such, the short term outcome serves as a \textit{surrogate outcome} for the ultimate goal of minimizing the number of active infections at the end of the observed trajectory.

\subsection{Online Super Learner under Working Models}\label{sec::osl}

In the following, we outline the Online Super Learner algorithm under flexible working models described in Subsection \ref{working_models}.
The proposed Online Super Learner is used to estimate $\bar{Q}_{0,i,t}$, the true conditional expectation of ${Y}_i(t)$ given the observed past. 

\subsubsection{Loss-based Parameter Definition and Estimation}\label{loss}

Let $\mathcal{M}_Y$ denote the working model for $\bar{Q}_{0,i,t}$, consisting of all functions $\bar{Q}_{i,t}$ that map $\mathcal{A} \times \mathcal{C_A}$ to $[0,1]$. The working model $\mathcal{M}_{Y}$ reflects a large class of candidate working models described in subsection \ref{working_models}. With that, we define $\mathcal{M}_Y$ as a collection of all conditional expectations with some  possible dependence structure across time and/or network that could have given rise to the observed data.
More specifically, for all $\bar{Q}_{i,t} \in \mathcal{M}_Y$, the conditional expectation of the outcome depends on the past only through a fixed dimensional summary measure $C_{L_{i}}(t)$, with $\mathcal{M}_Y$ satisfying Assumption \ref{fixed_dim_summary}. Under the decomposition of the fixed dimensional summary outlined in Assumption \ref{decompose}, $C_{L_i}(t) = (A_i(t), C_{A_i}(t))$ where $C_{A_i}(t) = h_{A_i}(Pa({A}(t))) = h_{A_i}(Pa({O}(t)))$. In addition, any $\bar{Q}_{i,t} \in \mathcal{M}_Y$ could be common in both samples and times $(i,t)$, in time $(t)$, or across samples $(i)$ --- thereby satisfying some combination of Assumption \ref{common_it}, \ref{common_t} or \ref{common_i}, respectively. By specifying $\bar{Q}_{i,t} \in \mathcal{M}_Y$, we are implying there is some structure to the dependent process, as described by one of the working models in Section \ref{working_models}; we are, however, not specifying the exact structure. 

Let $\bar{Q}_{i,t} \in \mathcal{M}_Y$. Under the working model $\mathcal{M}_Y$, we write the risk-based target parameter $\psi^{\text{risk}}_{0,t}$ corresponding to the \textit{risk-based} testing strategy as
\begin{align*}
    \psi^{\text{risk}}_{0,t}(A_i(t), C_{A_i}(t)) = \Psi^{\text{risk}}_{t}(P_0)(A_i(t), C_{A_i}(t)) \equiv {\bar{Q}}_{0,i,t}(A_i(t), C_{A_i}(t)),
\end{align*}
denoting the true conditional expectation of $Y_i(t)$ given the fixed dimensional summary of the observed past. Note that, in addition to defining the \textit{risk-based} testing strategy, estimating ${\bar{Q}}_{0,i,t}$ is an integral part of the Online SL for adaptive sequential surveillance (both for the loss- and TMLE- based selectors). In the following, we will refer mostly to the risk-based target parameter, but the formulation follows for all applications of ${\bar{Q}}_{0,i,t}$. Let $L$ denote a valid loss function for $\Psi^{\text{risk}}_{t}(P_0)$, and $C(i,m)$ the time $m$- and unit $i$-specific record $C(i,m) = (Y_i(m), A_i(m), \bar{O}(m-1))$. A valid loss $L$ is defined as a function whose true conditional expectation is minimized by the true value of the target parameter; here, the minimizer is therefore ${\bar{Q}}_{0,i,t}$. Further, let $L$ be a function that maps every $\Psi^{\text{risk}}_{t}(P)$ to $L(\Psi^{\text{risk}}_{t}(P)) : C(i,t) \mapsto L(\Psi^{\text{risk}}_{t}(P))(C(i,t))$. As our parameter of interest is a conditional mean, we could use the square error to define the loss, resulting in 
\begin{equation*}
    L(\Psi^{\text{risk}}_{t}(P))(C(i,t)) = w(i,t)\left[Y_i(t) - \Psi^{\text{risk}}_{t}(P)(A_i(t), C_{A_i}(t))\right]^2,
\end{equation*}
for sample $i$ and time $t$, where $w(i,t)$ is the subject and time specific weight. Our accent on appropriate loss functions strives from their multiple use within our framework --- as a theoretical criterion for comparing the estimator and the truth, and as a way to compare multiple estimators of the target parameter \citep{dudoit2003b,dudoit2003a, laan2006oracle, vaart2006}. The loss function $L$ and the working model $\mathcal{M}_Y$ will be used to estimate ${\bar{Q}}_{0,i,t}$.

We define the true risk, $R(P_0,\Psi^{\text{risk}}_{t}(P)) = R(P_0,\psi^{\text{risk}}_{t})$, as the expected value of $L(\Psi^{\text{risk}}_{t}(P))(C(i,t))$ w.r.t the true probability distribution over all samples and times. As $\psi^{\text{risk}}_{0,t} = \Psi^{\text{risk}}_{t}(P_0)$, we note that $\psi^{\text{risk}}_{0,t}$ is the minimizer of the true risk over all evaluated $\psi^{\text{risk}}_{t}$ in the parameter space, such that $\psi^{\text{risk}}_{0,t} = \text{arg}\min_{\psi^{\text{risk}}_{t}} R(P_0,\psi^{\text{risk}}_{t})$. Therefore, the true risk establishes the true measure of performance for $\Psi^{\text{risk}}_{t}(P)$, with $\Psi^{\text{risk}}_{t}(P_0)$ denoting the minimum. Further, we define the estimator mapping, $\hat{\Psi}^{\text{risk}}_{t}$, as a function from the empirical distribution to the parameter space $\bm{\Psi}^{\text{risk}}$. As previously defined, let $P_{n,t}$ be the empirical distribution of $n$ time series collected until time $t$. In particular, $P_{n,t} \mapsto \hat{\Psi}^{\text{risk}}_{t}(P_{n,t})$ represents a mapping from $P_{n,t}$, with $n$ time-series collected until time $t$, into a predictive function $\hat{\Psi}^{\text{risk}}_{t}(P_{n,t})$. Further, the predictive function $\hat{\Psi}^{\text{risk}}_{t}(P_{n,t})$ maps $(A_i(t), C_{A_i}(t))$ into a time- and subject-specific outcome, $\hat{\Psi}^{\text{risk}}_{t}(P_{n,t})(A_i(t), C_{A_i}(t))$. We can write $\psi^{\text{risk}}_{n,t}(A_i(t), C_{A_i}(t)) \coloneqq \hat{\Psi}^{\text{risk}}_{t}(P_{n,t})(A_i(t), C_{A_i}(t))$ as the predicted outcome for unit $i$ of the estimator $\hat{\Psi}^{\text{risk}}_{t}(P_{n,t})$ at time $t$, based on $(A_i(t), C_{A_i}(t))$. 

\subsubsection{Online Cross-validation Selector}\label{cv}

We resort to appropriate CV for dependent data in order to obtain an unbiased estimate of the true risk. To derive a general representation for cross-validation, we define a time $t$ specific split vector $B_t$, where $t$ indicates the final time-point of the currently available data where for all $1 \leq i \leq n$, $B_t(i, \cdot) \in \{-1,0,1\}^t$. Let $v$ be a particular $v$-fold, where $v$ range from $1$ to $V$. A realization of $B_t$ defines a particular split of the learning set into corresponding three disjoint subsets,
\[
  B_t^v(i,m,\cdot)=\begin{cases}
               -1, \ \  C(i,m) \  \text{not used}\\
                \ \ \  0, \ \ C(i,m) \  \text{in the training set}\\
                 \ \ \ 1, \ \ C(i,m) \  \text{in the validation set,}
            \end{cases}
\]
where $B_t^v(i,m)$ reflects the $v$-fold assignment of, at minimum, unit $i$ at time point $m$ for split $B_t^v$ trained on data until time $t$. Then, for each $t$, we define $P_{n,t}^0$ as the empirical distribution of the training sample until time $t$. Similarly, we let $P_{n,t}^1$ denote the empirical distribution of the validation set. Sets $\mathcal{B}_{t,v}^0$ and $\mathcal{B}_{t,v}^1$ contain all $(i,m)$ indexes in the training and validation sets for fold $v$, respectively. 

Suppose we have $K$ candidate estimators $\hat{\Psi}^{\text{risk}}_{k,t}$, $k=1, \ldots, K$, that can be applied to $(A_i(t), C_{A_i}(t))$ for $i \in [n]$ and $t \in [\tau]$. For a given problem, a library of prediction algorithms can be proposed. In particular, the candidate estimators for the outcome regression should include different learners corresponding to the underlying working model(s) in $\mathcal{M}_Y$. The algorithms in the candidate library may range from single time-series learners to networks of time-series, as well as learners that pool data across time, network, or all the data available up to time $t$. The Online Super Learner library can also include algorithms that put decaying weight of different rates on components of the past, or consider learners indexed by subsets of a network. We utilize working models in the estimation procedure without explicit reliance on any of the described working models in particular; we let the cross-validation procedure determine the underlying structure of the process at each time step. In order to evaluate performance of each $\hat{\Psi}^{\text{risk}}_{k,t}$, we use cross-validation for dependent data to estimate the average loss for each candidate. In particular, each $\hat{\Psi}^{\text{risk}}_{k,t}$ is trained on the training set $P_{n,t}^0$ and results in a predictive function $\hat{\Psi}^{\text{risk}}_{k,t}(P_{n,t}^0)$ for $k= 1, \ldots, K$. We define the online cross-validated risk for each candidate estimator as:
\begin{align*}
R_{t,CV}(P_{n,t}^1,\hat{\Psi}^{\text{risk}}_{k,t}(\cdot)) &= \sum_{j=1}^{t} \sum_{v=1}^V \sum_{(i,s) \in \mathcal{B}_{j,v}^1} L(\hat{\Psi}^{\text{risk}}_{k,t}(P_{n,j}^0))(C(i,s)), 
\end{align*}
where $R_{t,CV}(P_{n,t}^1,\hat{\Psi}^{\text{risk}}_{k,t}(\cdot))$ is the cumulative performance of $\hat{\Psi}^{\text{risk}}_{k,t}$ trained on the training sets and evaluated on the corresponding validation samples until time~$t$. For instance, while $\hat{\Psi}^{\text{risk}}_{k,t}(P_{n,t}^0)$ is trained on the training set $P_{n,t}^0$, its performance will be evaluated over the validation set $P_{n,t}^1$. The online cross-validated risk estimates the following true online cross-validated risk, denoted as $R_{t,CV}(P_{0}, \hat{\Psi}^{\text{risk}}_{k,t}(\cdot))$ and expressed as
\begin{align*} 
R_{t,CV}(P_{0},\hat{\Psi}^{\text{risk}}_{k,t}(\cdot)) &= \sum_{j=1}^{t} \sum_{v=1}^V \sum_{(i,s) \in \mathcal{B}_{j,v}^1} E_{0}[L(\hat{\Psi}^{\text{risk}}_{k,t}(P_{n,j}^0))(C(i,s)) | C_{A_i}(s)].
\end{align*}
Note that $R_{t,CV}(P_{0}, \hat{\Psi}^{\text{risk}}_{k,t}(\cdot))$ reflects the true average loss for the candidate estimator with respect to the true conditional distribution. As opposed to the true online cross-validated risk, $R_{t,CV}(P_{n,t}^1, \hat{\Psi}^{\text{risk}}_{k,t}(\cdot))$ gives an empirical measure of performance for each candidate estimator $k$ trained on training data until time $t$. In light of that, we define the discrete online cross-validation selector as:
\begin{align}
k_{n,t} = \text{arg}\min_{k=1,\ldots, K} R_{t,CV}(P_{n,t}^1, \hat{\Psi}^{\text{risk}}_{k,t}(\cdot)),
\end{align}
which is the estimator that minimizes the online cross-validated risk. The discrete (online) Super Learner is the estimator that at each time point uses the estimates from the discrete online cross-validation selector. Since each of the $k$ learners can reflect different working models, the discrete (online) SL picks one of the candidate dependence structures for time point $t$. We emphasize that the discrete SL can switch from one learner to another as $t$ progresses, in response to accumulating more data and detecting changes in the network and trajectory. Note that, if all the candidate estimators are online estimators, the discrete (online) SL is itself an online estimator. 

In order to study performance of an estimator, we construct loss-based dissimilarity measures. In particular, loss-based dissimilarity compares the performance of a particular estimator to the true parameter, defined as
\begin{equation*}
d_{0,t}(\hat{\Psi}^{\text{risk}}_{k,t}, \psi^{\text{risk}}_{0,t})
= \sum_{j=1}^{t} \sum_{v=1}^V \sum_{(s) \in \mathcal{B}_{j,v}^1} E_{0}\bigg[ \Big( L(\hat{\Psi}^{\text{risk}}_{k,t}(P_{n,t}^0)) - L({\psi}^{\text{risk}}_{0,t})\Big) (C(i,s)) \bigg| C_{A_i}(s)\bigg].
\end{equation*}
We define the time $t$ \textit{oracle} selector as the unknown estimator that uses the candidate closest to the truth in terms of the defined dissimilarity measure:
\begin{align}
\overline{k}_{n,t} = \text{arg}\min_{k=1,\ldots,K} d_{0,t}(\hat{\Psi}^{\text{risk}}_{k,t}(P_{n,t}^0), \psi^{\text{risk}}_{0,t}).
\end{align}
Due to it being a function of the true conditional mean, the oracle selector cannot be computed in practice. However, we can utilize it as benchmark in order to describe performance of the online cross-validation based estimator. Following the argument given by \cite{benkeser2018}, it follows that the performance of the discrete Online Super Learner is asymptotically equivalent to that of the oracle selector. The result relies on the martingale finite-sample inequality by \cite{handel2009} to show that, as $t \rightarrow \infty$, 
\begin{align}
\frac{d_{0,t}(\hat{\Psi}^{\text{risk}}_{k_{n,t},t}(P_{n,t}^0),\psi^{\text{risk}}_{0,t})}
{d_{0,t}(\hat{\Psi}^{\text{risk}}_{\overline{k}_{n,t},t}(P_{n,t}^0),\psi^{\text{risk}}_{0,t})} \rightarrow_p 1.
\end{align}

\subsubsection{Ensemble of Candidate Estimators}\label{ensemble}

In this section, we consider a more flexible online learner by considering an ensemble of a given set of estimators. As individual learners reflect different candidate working models for time and network dependence, a weighted combination of candidate estimators reflects a mixture of working models. We define $\hat{\Psi}^{\text{risk}}_{\beta,t}$ as an ensemble of $K$ estimators indexed by a finite-dimensional vector of coefficients $\beta$, where $\beta = (\beta_1, \ldots, \beta_K)$. For example, $\hat{\Psi}^{\text{risk}}_{\beta,t}$ could represent a convex linear combination:
\begin{equation*}
    \hat{\Psi}^{\text{risk}}_{\beta,t} = \sum_{k=1}^K \beta_k \hat{\Psi}^{\text{risk}}_{k,t},
\end{equation*}
such that $\sum_{k=1}^K \beta_k = 1$ with $\beta_k \geq 0$ 
for all $\beta_k$. Let $R_{t,CV}(P_{n,t}^1,\hat{\Psi}^{\text{risk}}_{\beta,t}(\cdot))$ be the online cross-validated risk for $\hat{\Psi}^{\text{risk}}_{\beta,t}$ given by
\begin{align*}
R_{t,CV}(P_{n,t}^1,\hat{\Psi}^{\text{risk}}_{\beta,t}(\cdot)) &= \sum_{j=1}^{t} \sum_{v=1}^V \sum_{(i,s) \in \mathcal{B}_{j,v}^1} L(\hat{\Psi}^{\text{risk}}_{\beta,t}(P_{n,j}^0))(C(i,s)).  
\end{align*}
We denote $\beta_{n,t}$ as the choice of $\beta$ that minimizes the online cross-validated risk,
\begin{equation}
    \beta_{n,t} = \text{arg}\min_{\beta} R_{t,CV}(P_{n,t}^1,\hat{\Psi}^{\text{risk}}_{\beta,t}(\cdot)).
\end{equation}
Note that $\beta_{n,t}$ itself is not an online estimator, since it involves recomputing the minimum for each $t$. 
We can define an oracle selector for this class of estimators as the choice of weights that minimizes the true average of the loss-based dissimilarity:
\begin{equation}
\overline{\beta}_{n,t} = \text{arg}\min_{\beta} d_{0,t}(\hat{\Psi}^{\text{risk}}_{\beta,t}(P_{n,t}^0), \psi^{\text{risk}}_{0,t}).     
\end{equation}
The oracle results extend to an ensemble of candidate estimators, and we can show that
\begin{align}
\frac{d_{0,t}(\hat{\Psi}^{\text{risk}}_{\beta_{n,t},t}(P_{n,t}^0),\psi^{\text{risk}}_{0,t})}
{d_{0,t}(\hat{\Psi}^{\text{risk}}_{\overline{\beta}_{n,t},t}(P_{n,t}^0),\psi^{\text{risk}}_{0,t})} \rightarrow_p 1
\end{align}
as $t$ goes to infinity. As such, the performance of the Online Super Learner is asymptotically equivalent with the optimal ensemble of candidate estimators, which reflect different working models.

\subsection{Algorithms}\label{sec::algs}

\begin{algorithm}[H]\caption{Loss-based Selector}\label{alg::V1}
\SetAlgoLined

\textbf{Input:} $k$, $n$, $\{g^*_{t,1}, \ldots, g^*_{t,S}\} \in \mathcal{G}$ \\

\vspace{2mm}
\For{$t=1$ {\bfseries to} $t=\tau$}{
    \vspace{2mm}
    \eIf{$t < 2$}{
        \vspace{2mm}
        Let each $g^*_{t,s}$ correspond to random testing with success probability $1/n$. \\
        
        \vspace{2mm}
        Pick the $s^{th}$ stochastic intervention randomly, and let $s_t = s$. \\
        
        \vspace{2mm}
        Distribute $k$ tests according to $g^*_{t,s_t}$, and observe $Y(t)$.
        }{
        \vspace{2mm}
        Generate an estimate $\hat{\bar{Q}}_{i,t}$ of $\bar{Q}_{0,i,t}$ based on the training data $P_{n,t}^0$. \\
        (based on various working models for the dependence structure) \\
        
        \vspace{2mm}
        Define interventions $g^*_{t,s}$ for the current time point $t$. \\
        
        \vspace{2mm}
        Assign tests according to $g_{t-1,s_{t-1}}$, and observe $Y(t)$.
        
        \vspace{2mm}
        \For{$s=1$ {\bfseries to} $s=S$}{
        
        \vspace{2mm}
        Estimate the $s$-specific $\psi_{t,g^{*}_{t,s}}$ as \\
        \begin{equation*}
         {\Psi}_{t,g^{*}_{t,s}}(\hat{\bar{Q}}_{t}) = \frac{1}{n}\sum_{i=1}^n \mathbb{E}_{\hat{\bar{Q}}_{i,t},g_{i,t,s}^*}[Y_i(t) \mid \bar{O}(t-1)].
        \end{equation*}
        
        \vspace{2mm}
        Evaluate the online CV risk for the $s^{th}$ design over a window: \\
        \begin{align*}
         \sum_{m=t+1}^{\tau_w} 
          \sum_{i=1}^n \frac{1}{g^*_{i,m,s}(1 \mid \bar{O}_i(t-1))} \left( Y_{i,g^*_{i,m,s}}(m)
          - \hat{\Psi}_{i,m,g^*_{i,m,s}}(\hat{\bar{Q}}_{m})(\bar{O}(m-1)) \right)^2.
         \end{align*}
        }
        
        \vspace{2mm}
        Pick the discrete SL design $s_t$:
   \begin{align*}
     s_t &= \min_s 
     \sum_{m=t+1}^{\tau_w} L(\hat{\Psi}_{m,g^*_{m,s}}(P_{n,m}))(C(m)). 
\end{align*}
        }
    } 
\end{algorithm}

\newpage
\begin{algorithm}[H]\caption{TMLE-Selector}\label{alg::V2}
\SetAlgoLined

\textbf{Input:} $k$, $n$, $\{g^*_{t,1}, \ldots, g^*_{t,S}\} \in \mathcal{G}$ \\

\vspace{2mm}
\For{$t=1$ {\bfseries to} $t=\tau$}{
    \vspace{2mm}
    \eIf{$t < 2$}{
        \vspace{2mm}
        Let each $g^*_{t,s}$ correspond to random testing with success probability $1/n$. \\
        
        \vspace{2mm}
        Pick the $s^{th}$ stochastic intervention randomly, and let $s_t = s$. \\
        
        \vspace{2mm}
        Distribute $k$ tests according to $g^*_{t,s_t}$, and observe $Y(t)$.
        }{
        \vspace{2mm}
        Generate an estimate $\hat{\bar{Q}}_{t}$ of $\bar{Q}_{0,i,t}$ based on the training data $P_{n,t}^0$. \\
        (based on various working models for the dependence structure) \\
        
        \vspace{2mm}
        Define interventions $g^*_{t,s}$ for the current time point $t$. \\
        
        \vspace{2mm}
        Assign tests according to $g_{t-1,s_{t-1}}$, and observe $Y(t)$.
        
        \vspace{2mm}
        \For{$s=1$ {\bfseries to} $s=S$}{
        
        \vspace{2mm}
        Update $\hat{\bar{Q}}_{t}$ by fitting a logistic model 
        \begin{equation*}
        \text{logit} \ \hat{\bar{Q}}^*_{t}(A(t),\bar{O}(t-1))  = \text{logit} \ \hat{\bar{Q}}_{t}(A(t),\bar{O}(t-1)) + \epsilon,
        \end{equation*}
        
        \noindent
        with weights defined as $w_t = g_{t}^*(A(t) \mid \bar{O}(t-1))/g_{t}(A(t) \mid \bar{O}(t-1))$. \\
        
        \vspace{2mm}
        The updated estimate solves the EIF estimating equation
        \begin{equation*}
        \frac{1}{n}\sum_{i=1}^n \frac{g_{t,i}^*(A_i(t) \mid \bar{O}(t-1))}{g_{t,i}(A_i(t) \mid \bar{O}(t-1))}(Y_i(t) - \hat{\bar{Q}}^*_{t}(A_i(t),\bar{O}(t-1))) = 0.
        \end{equation*} 
        
        \vspace{2mm}
        The TMLE of the $s$-specific stochastic intervention is
        \begin{equation*}
        {\Psi}_{t,g^{*}_{t,s}}(\hat{\bar{Q}}^*_{t}) = \frac{1}{n}\sum_{i=1}^n \mathbb{E}_{\hat{\bar{Q}}^*_{t},g_{i,t,s}^*}[Y_i(t) \mid \bar{O}(t-1)].
        \end{equation*}
        
        \noindent
        with estimated variance $\hat{\sigma}_{t,s}^2$.
        }
        
        \vspace{2mm}
        Pick the discrete SL design $s_t$ which maximizes: \\
        \vspace{2mm}
        \noindent
        1) Point estimate: $\max_{s} {\Psi}_{t,g^{*}_{t,s}}(\hat{\bar{Q}}^*_{t})$ \\
        
        \noindent
        2) Lower bound of the CI: $\max_{s} [{\Psi}_{t,g^{*}_{t,s}}(\hat{\bar{Q}}^*_{t}) - 1.96 \hat{\sigma}_{t,s}/\sqrt{n}].$ \\
        }
    } 
\end{algorithm}

\newpage

\subsection{Agent-based model of the university campus}\label{sec::fullmodel}

\subsubsection{Model Description}\label{sec::agent-based-model}
We have developed a dynamic, agent-based model of transmission of SARS-CoV-2 on a residential university campus. The model parameters reflect transmission dynamics among students and faculty on a medium-size public university with on- and off-campus living arrangements, class number and size compatible with the University of California, Berkeley. In particular, we present an agent-based model for a setting with 20,000 university-affiliated individuals with varying age, social network, class size, risk-level and living accommodations in a U.S. college town. While we parameterized the model according to U.C. Berkeley, our simulations can be easily modified to reflect any residential higher-education institution. In the following subsections, we describe the core features, dynamics and assumptions underlying the agent-based model. The main features of the infectious disease dynamics are illustrated in Figure 4,
with the key parameter values provided in Table \ref{table:population} and \ref{table:parameters}.

\vspace{3mm}
\noindent
\textbf{Population and Network Structure} \\
We focus on a residential campus community as the target population, broadly reflecting students and faculty/staff at a medium size university. During the pandemic, most universities implemented a wide array of measures aimed at limiting transmission on campus --- including mandatory vaccination, mask-wearing, social distancing, enhanced cleaning protocols, increased availability of sanitizing products and cancelling large social gatherings. While each of these interventions has benefits, we focus solely on the optimal testing strategy. By not including other interventions which can alleviate the spread of the infectious disease, we aim to investigate performance of the proposed method under the worst-case scenario. We don't model any breaks or high-risk events during the semester, such as holiday travel. Most importantly, we assume the modeled population is constant, where individuals can only be removed due to isolation. While we model the campus population as a closed community, infections induced by interactions with individuals outside the campus population are allowed - thus providing a steady stream of new cases.


We assume the campus population is divided into distinct groups with different collective behavior, dictated by their covariates. For example, we model three sub-populations --- students living on-campus, students living off-campus, and faculty/staff, as shown in Table \ref{table:population}. The distinct groups are differentiated by baseline covariates, underlying risk of infection, and different degrees of interaction within one's network. The age distribution is modeled in order to reflect a predominately younger campus population, with 6 age categories ($<18$, $18-28$, $29-38$, $39-48$, $49-68$, $>68$) sampled with probabilities $(0,1, 0.5, 0.2, 0.07, 0.07, 0.06)$, respectively. The baseline risk of infection is sampled from a $\textit{beta(1, n/36000)}$ distribution, reflecting dependence on the size of the observed population. The modeled network structure consists of several components, including off-campus living for both students and faculty, on-campus housing, in person classes, and random exposure. Each network type has an unique probability of transmission within a graph, with highest being for students in a communal housing. Below, we describe each component of the network structure in more details. 

\begin{enumerate}
    \item \textbf{Off-campus housing for students:} We model off-campus housing separately for students and faculty. Students of similar age category are grouped in off-campus housing units, with household sample size drawn from $(1-8)$ range from a \textit{negative binomial} distribution with $\mu = 2$. While we concentrate on the university campus population, this is not exclusively represented in the observed housing structure population, with students being vastly dispersed in the region. With that in mind, we assign housemates randomly to a total of $n*0.01$ off-campus houses of varying household size. The probability of transmission among the same household is $0.03$.
     
    \item \textbf{Off-campus housing for faculty and staff:} We model off-campus housing for faculty/staff similar to the student off-campus accommodations. Individuals older that $28$ years old are candidates for faculty/staff housing. We sample the size of a faculty/staff household from $(0-2)$ range from a \textit{negative binomial} distribution with $\mu = 0.5$. We assign housemates randomly with no age preference (accounting for possible family ties). We distribute sampled household individuals to a total of $n*0.005$ faculty/staff houses. The probability of transmission among the same household is $0.03$, as in off-campus living arrangements for students.
    
    \item \textbf{On-campus housing:} The probability of a student being part of on-campus housing available at U.C. Berkeley is derived based on the $25$ communal living buildings with $8908$ students at maximum occupancy. During the regular school year, $14\%$ of the student population is living in a university provided dorms and Greek housing. In the Summer of 2020, approximately $5\%$ were projected to return to communal housing in the Fall of 2020, with an average occupancy of 2 students per room. We sample the size of a household from a \textit{negative binomial} distribution with $\mu = 2$, and allocate people to the $25$ available communal buildings with equal probability. We assume that the risk of transmission from the community to off-campus students is lower than the risk of transmission from the community to on-campus students, based on the evidence from a campus outbreak of H1N1 in 2009 \citep{guh2011}. We further assume that on-campus students living in communal settings are considered a higher risk for transmission. If a student is part of communal housing (e.g. dorms or Greek housing), the risk of infection increases by $0.01$ in addition to the regular housing risk. 
    
    \item \textbf{In-person classes:} We modeled in-person, online and hybrid available classes in Fall of 2020 based on the suggested schedule announced by the U.C. Berkeley administration in the Summer of 2020. This consisted of $111$ hybrid and $98$ flexible classes, and up to $314$ in-person lectures with a maximum of $25$ students per class. This is in accordance with the wide-spread university policy to strive for a majority of online classes, with few small-in-size in-person lectures and staggered class times in order to decrease student contact. We sample the in-person class size from a \textit{negative binomial} distribution with a minimum of $15$ and maximum of $25$ students per class. We further assume $18\%$ of total student and faculty/staff population would return to U.C. Berkeley for in-person classes and university-based responsibilities. The probability of transmission for an in-person attendance among class members and people that frequent the same classroom/building is $0.01$.
    
    \item \textbf{Random:} Finally, we account for random exposure from people being in close contact during their regular day-to-day activities (e.g., taking public transportation). We model random exposure as the number of people outside one's network that come in close contact with the individual in question. We sample the number of nodes in a random graph from the \textit{negative binomial} distribution with a minimum of $3$ and maximum of $25$ encounters per day. This allocation is dependent on individual risk, and corresponds to the latent part of the network structure. The probability of transmission in a random graph is $0.005$. In Appendix Section \ref{sec::add_sims}, we investigate how the proposed method responds to increased individual risk.

\end{enumerate}

\vspace{3mm}
\noindent 
\textbf{Individual Disease Progression} \\
We assume that the disease evolution always progresses through set stages in each infected individual, as exemplified in Figure 4
and Table \ref{table:parameters}. We separate the overall student and faculty/staff populations into the following compartments at each time step: susceptible ($S$), exposed ($E$), detectable infectious ($It$), symptomatic infectious ($Is$), asymptomatic infectious ($Ia$), recovered ($R$) and isolated ($I$).
The probability of a new infection, conditional on being susceptible, depends on the hazard at the current time-step $t$. In particular, the hazard function takes into account the individual's current stage, time spent in the state, one's full network, individual risk and covariates, and the current state of the epidemic; the more infectious people and the more advanced the epidemic, the greater the risk of a new infection. If an individual is never infected, they remain susceptible until the end of the observation period. Infected individuals advance to the next compartment, or remain in their current one, stochastically at each time step. The transition probabilities and average length of stay for each compartment are modeled based on the literature available in Summer of 2020, and described in more details in the following. We initiate the infectious disease trajectory with $8$ exposed, $2$ temporarily infectious and $2$ symptomatic infectious cases of SARS-CoV-2 infection at the start of data collection.

\begin{enumerate}
    \item \textbf{Susceptible (S):} Except for the seeded samples, all individuals start as susceptible. The $S$ stage has no time to next state, as all units remain susceptible until infection or end of the semester. 
    
    \item \textbf{Exposed (E):} If infected, a susceptible individual transitions to an exposed status. This stage is not yet infectious due to a low viral load, and it is not detectable via testing. Exposed units spend 4.5 days on average as exposed, with the number of days in the $E$ compartment sampled from
    $\textit{gamma}(9,2)$ distribution. 
    
    \item \textbf{Detectable Infectious (It):} From the $E$ compartment, exposed individual transitions to a temporary state $It$. Each individual spends 1 day on average as $It$, with the number of days in the detectable infectious compartment sampled from the $\textit{gamma}(1,1)$ distribution. During this transition period, individuals are not symptomatic, but are infectious. The compartment $It$ is also the first stage of the disease trajectory at which the infection can be detected via a test. We model the ability to infect others as reduced while $It$, but increasing with each time step, with a peak at the transition to the next compartment. Since a significant fraction of COVID-19 patients are asymptomatic, especially within the younger population, we divided the next infectious compartment into symptomatic and asymptomatic with distinct transition probabilities and duration of the infectious state. Determination whether one is asymptomatic or symptomatic is a Bernoulli with success probability depending on the age of the exposed individual. For instance, all samples within the age gap $<18$ to $28$ were symptomatic with success probability 0.4, $29-48$ with probability 0.6, and the oldest members of the campus probability were symptomatic with probability 0.8. 
    
    \item \textbf{Symptomatic (Is):} Infected individual can  transition to a symptomatic compartment following the temporarily infectious state. One spends 13 days on average as a symptomatic patient, with the number of days in the $Is$ stage sampled from the $\textit{gamma}(13,1)$ distribution. The symptomatic compartment encompasses any COVID-19 symptom, or combination of, that warrants a test; we assume a patient that is symptomatic exhibits symptoms the entire time while $Is$. Consequently, a test can be requested and administered at any time point during the symptomatic stage, not just at the onset of symptoms. As a $Is$, a person is infectious for the duration of the state; however, their infectiousness decreases the longer they are in the current stage.
    
    \item \textbf{Asymptomatic (Ia):} Infected individual can also transition to an asymptomatic infectious state following the $It$ compartment. If asymptomatic, one spends 7.5 days on average as $Ia$, with the number of days in the current stage sampled from the $\textit{gamma}(7.5,1)$ distribution. While there is still emerging research in this area, we assume relative infectiousness for the asymptomatic class to be less than for the symptomatic class; the reduction in transmission probability is assumed at $39\%$ \citep{Perkins2020}. We let asymptomatic individuals remain asymptomatic for the duration of the infection, assuming no transition from $Ia$ to $Is$ compartment. As the campus population is predominately young, the percentage of asymptomatic individual is higher than in the general population. This makes detecting active infections a more difficult task for the campus community.
    
    \item \textbf{Recovered (R):} An individual recovers by going through the full disease cycle, $(S,E,It,Is/Ia,R)$. As the campus population is predominately young and healthy, we don't model death as an outcome. Instead, all individuals eventually reach a terminal state $R$. We assume recovered individuals obtain at least a temporary immunity that lasts the length of the semester. Therefore, recovered individuals do not become susceptible again. Studies are still ongoing regarding the duration, if any, of temporary immunity \citep{wu2021}. 
    
    \item \textbf{Isolated (I):} Finally, an infected member of the campus population can be identified via testing at any point of the disease trajectory, except for the $E$ and $R$ state. Positive test results at time $t$ lead to possible contact tracing at $t+1$ for the complete known network. Individuals administered a test and diagnosed positive for COVID-19 are isolated. Quarantine conditions are modeled on a continuum via the isolation factor, in an attempt to mimic realistic conditions. In particular, if the isolation factor is $0$, quarantine is modeled as isolation with complete reduction in one's contact rate for the duration of the infection. For higher isolation factor, detected infectious individual is still able to infect others at a reduced rate. If caught, the infected individual goes though all the usual disease stages with a varying isolation level, until it reaches a terminal $R$ state. As mentioned previously, isolated, and eventually recovered patients, do not become susceptible again.
\end{enumerate}



\begin{table}[htbp]\caption{Simulated university population during COVID-19 pandemic \label{table:population}}
        \centering
        \small
        \setlength\tabcolsep{4pt}
        \begin{tabular}{|c|c|c|c|}
            \toprule
               \textbf{Parameter} &  \textbf{Population}       &    \textbf{On-campus housing}       &  \textbf{In-person class }\\
            \midrule
                Value (total or $\%$) & 20,000 & $5\%$ & $18\%$ \\
            \bottomrule
        \end{tabular}
\end{table}

\begin{table}[htbp]\caption{Simulation model parameters for the COVID-19 pandemic
\label{table:parameters}}
        \centering
        \small
        \setlength\tabcolsep{4pt}
        \begin{tabular}{|c|c|c|c|c|}
            \toprule
               \textbf{Parameter} & \textbf{Symbol} &    \textbf{Average (days)} & \textbf{Range (Q1,Q3)} & \textbf{Distribution} \\
            \midrule Latent Period & \textbf{E} & 4.5 & (3.4, 5.4) & \textit{gamma}(9,2) \\
            Detectable and Infectious & \textbf{It} & 1 & (0.7, 1.4) & \textit{gamma}(1,1) \\
            Asymptomatic & \textbf{Ia} & 7.5 & (5.5, 9.1) & \textit{gamma}(7.5,1) \\ 
            Symptomatic & \textbf{Is} & 13 & (10.4, 15.2) & \textit{gamma}(13,1) \\
            \bottomrule
        \end{tabular}
\end{table}

\subsubsection{Testing Strategies as Stochastic Interventions}\label{tests}

We can describe a wide range of testing schemes by a stochastic intervention $g^*_{t}$. The familiar case of static interventions --- defined by setting $A_i(t)$ to a value $a$ in its support $\mathcal{A}$ --- can be recovered by choosing degenerate candidate distributions which place all mass on just a single value \citep{stock1989nonparametric,diaz2012}. In particular, we can conceptualize a testing strategy as an intervention that assigns all individuals testing allocations $\{A_1(t), \ldots, A_n(t)\} \in \{0,1\}^n$ at time $t$ in a two-step procedure. Operationally, we delineate between the probability of receiving a test based on the available history, and how the said probabilities are used in order to output a final testing decision. 
For example, $g^*_{t}$ could assign tests based on the 
ranking of the probability of receiving a test given the past, or based on the probability itself. Alternatively, $g^*_{t}$ could be a deterministic intervention conditional on the observed past (e.g., \textit{rule-based}, static interventions). Let $f$ denote a function that takes the probability of being assigned a test given the past, and outputs either a stochastic or deterministic rule as to how such probabilities are to be used to assign a test. We formally define the stochastic intervention $g^*_{t}$ as a function that maps every $f$ to $g^*_{t}(f) : (Pa({A}(t))) \mapsto g^*_{t}(f)(1 \mid Pa({A}(t))) = g^*_{t}(f)(1 \mid \bar{O}(t-1))$. If $f$ is an identity function, and each sample is to receive a test just based on its conditional probability, we simply write $g^*_{t}(f)(1 \mid Pa({A}(t))) = g^*_{t}(1 \mid Pa({A}(t)))$.


In the following, we define several testing strategies used to pick $k$ percent of the population to be tested --- from a new \textit{risk-based} testing scheme to commonly allocated strategies, such as \textit{symptomatic} testing and \textit{contact tracing}. In connection to our general stochastic intervention framework, we emphasize that the testing allocation could be a static rule depending on just the current knowledge of the network, recovering the \textit{contact tracing} testing scheme. On the other hand, a static rule depending on the reported symptoms results in \textit{symptomatic} testing. Instead of relying on simple rules, $g^*_{t}$ could depend on the current estimate of ${\bar{Q}}_{0,i,t}$, in which case we can incorporate the current risk of being infected as part of the test allocation strategy. Further, we can then sample based on the current estimate of ${\bar{Q}}_{0,i,t}$, or pick the top $k$ percent ranked samples. We compare all proposed testing strategies to benchmarks where we either know the true unknown status of each individual or the true ${\bar{Q}}_{0,i,t}$. 
We denote as ``realistic'' all testing strategies that can be implemented in the general population during an epidemic. 

\vspace{3mm}
\noindent
\textbf{Testing allocation functions} \\
Let $f$ denote any function that takes as input the conditional probability of being assigned a test given the observed past, and outputs either a stochastic or deterministic rule as to how such probabilities are to be used to assign a test. The function $f$, together with the collected past, defines a stochastic intervention $g_{t}^*$ at time $t$. We study two such functions: $f=f_S$ and $f=f_R$, defining sample and rank functions, respectively. We define $f_S$ as an identity function such that
\begin{equation}\label{rule_sample}
    g^*_{t}(f_S)(1 \mid \bar{O}(t-1)) \equiv g^*_{t}(1 \mid \bar{O}(t-1)) = p^*_{a_i(t)}(A_i(t) \mid \bar{O}(t-1)).
\end{equation}
With that, $g^*_{t}(f_S)(1 \mid \bar{O}(t-1))$ is a stochastic intervention which assigns tests according to the probability of being tested given the past. On the other hand, $f_R$ ranks the current estimate of the conditional probability of being tested and allocates $k$ percent of tests to the top ranked samples. In particular, let $S_P$ denote the survival function of $g^*_{t}(1 \mid \bar{O}(t-1))$ such that $c \mapsto P(g^*_{t}(1 \mid \bar{O}(t-1)) > c)$. Then, we can define $c^* \equiv \inf \{c : S_P(c) \leq k\}$ as the cutoff at which at most $k$ percent of individuals get tested based on $S_P$. We define the rank-based stochastic intervention as 
\begin{equation}\label{rule_rank}
    g^*_{t}(f_R)(1 \mid \bar{O}(t-1)) \equiv \mathbb{I}(g^*_{t}(1 \mid \bar{O}(t-1)) > c^*),
\end{equation}
which allocates tests to the top $k$ percent of individuals with the highest ranked probability of testing given the observed history. In the following subsections, we describe and compare a range of testing allocation schemes based on $g^*_{t}(f)(1 \mid \bar{O}(t-1))$ for each $i \in [n]$ and $t \in [\tau]$, with $f=f_S$ and $f=f_R$.

\vspace{3mm}
\noindent
\textbf{Realistic Testing Strategies}\label{realistic} \\
The ``realistic'' testing strategies include test allocations often described in the literature, including \textit{symptomatic} and \textit{contact tracing}. In addition, it includes the new \textit{risk-based} strategy based on the current estimate of the conditional expectation of ${Y}_i(t)$ given the observed past. In the following, we assume $L(t)$ includes covariates that describe the current presence of symptoms associated with the infectious disease in question ($L^{\text{symp}}(t)$) and the time $t$ network of each individual ($F(t)$). 

\begin{enumerate}
    \item \textbf{Symptomatic Testing} \\
    One of the most commonly described test allocation strategy in the literature is \textit{symptomatic} testing. Briefly, \textit{symptomatic} testing entails giving a test to $k$ percent of individuals with reported symptoms, where $L_i^{\text{symp}}(t)=1$ describes a symptomatic patient at time $t$. We can denote such testing strategy as the following deterministic intervention $g^*_{t}$ at $t$:
\[
   g^*_{t}(f)(1 \mid \bar{O}(t-1))=\begin{cases}
               1, \ \  L_i^{\text{symp}}(t)=1 \\
               0, \ \ L_i^{\text{symp}}(t)=0,
            \end{cases}
\]
    where $E[g^*_{t}(f)(1 \mid \bar{O}(t-1))] \leq k$.    

\item \textbf{Contact Tracing} \\
Another commonly described and implemented testing strategy is based on \textit{contact tracing}. Here, each individual with a current positive test has their entire known network tested for the same cause as well, while respecting the testing resource constraint. The sample $i$'s network consists of family, friends, colleges and all other individuals who came in close contact with an infected individual within a specified time frame. Typically the more comprehensive one's network, the more effective contact tracing is as a testing strategy. We write \textit{contact tracing} as the following deterministic intervention at time $t$:
\[
   g^*_{t}(f)(1 \mid \bar{O}(t-1))=\begin{cases}
               1, \ \  F_i(t)=1 \\
               0, \ \  F_i(t)=0.
            \end{cases}
\]

\item \textbf{Random} \\
    \textit{Random} testing corresponds to assigning tests with equal probability to $k$ proportion of available individuals. Here, no information on the samples, or the current state of the epidemic trajectory, is used to assign tests. The stochastic intervention $g^*_{t}$ corresponds to assigning uniform testing weights
    \begin{equation*}
        g^*_{t}(f)(1 \mid \bar{O}(t-1)) = 1/n,
    \end{equation*}
    at each time $t \in [\tau]$ and for all samples $i \in [n]$.

\vspace{1em}
\item \textbf{Risk-Based Testing} \\
Instead of relying solely on the current symptoms, known network of each patient or the combination of both --- we can instead incorporate an estimate of the current risk of being infected into the testing scheme. In particular, we use the current estimate of ${\bar{Q}}_{0,i,t}$ fit on the training set and available covariates in order to assign tests at the next time step. 
We can define the \textit{risk-based} testing scheme as the following intervention at time $t$:
\begin{equation*}
    g^*_{t}(f)(1 \mid \bar{O}(t-1)) \equiv f(\hat{\bar{Q}}_{i,t}(Pa(L(t)))
\end{equation*}
which assigns tests based on the current estimate of one's risk of being infected, given their observed past. 
\end{enumerate}

\vspace{3mm}
\noindent
\textbf{Benchmark Testing Strategies}\label{benchmark} \\
In the following, we define several benchmark testing strategies used to evaluate performance of above described ``realistic'' testing schemes. As previously noted, most benchmark allocations are impossible to implement in practice due to ethical concerns, dependence on unknown statistical quantities, or latent variables.

\begin{enumerate}
    \item \textbf{No Testing} \\
     The first benchmark strategy includes the natural progression of the epidemic when no individuals are isolated, and no tests are allocated. As such, we define $g^*_{t}$ at each time $t \in [\tau]$ and for all samples $i \in [n]$ as $g^*_{t}(f)(1 \mid \bar{O}(t-1)) = 0$.
     
    \item \textbf{True Risk} \\
    The true risk benchmark entails using the true risk function in order to assign tests at each time point $t \in [\tau]$, instead of the current estimate of ${\bar{Q}}_{0,i,t}$. The true risk scheme is based on unknown components of the process, and thus is not possible to implement in practice (unless we \textit{a priori} know ${\bar{Q}}_{0,i,t}$). Nevertheless, it is a useful benchmark for \textit{risk-based} testing strategies, setting an upper limit on their performance over time. We can denote the true risk testing benchmark as the following intervention at time $t$:
    \begin{equation*}
        g^*_{t}(f)(1 \mid \bar{O}(t-1)) \equiv f({\bar{Q}}_{0,i,t}(Pa(L(t)))
        \end{equation*}
    which assigns tests based on the true risk of being infected at time $t$, given the observed past. 
    
     \item \textbf{Perfect (True Status)} \\
    The true status benchmark entails testing $k$ proportion of individuals with current active infection. As $Y^l(t)$ is a latent variable, it is not observed. While impossible to implement in practice, the true status benchmark serves as an upper bound on the effectiveness of testing for infection control. We define it 
    as the following deterministic intervention $g^*_{t}$ at time $t$:
    \[
        g^*_{t}(f)(1 \mid \bar{O}(t-1))=\begin{cases}
               1, \ \  Y^l_i(t)=1 \\
               0, \ \  Y^l_i(t)=0.
            \end{cases}
    \]
\end{enumerate}


\tikzstyle{decision} = [diamond, draw=red, fill=white!20, 
    text width=4em, text badly centered, node distance=3.5cm, inner sep=0pt, font=\bfseries]
\tikzstyle{block} = [rectangle, draw, fill=white!20, 
    text width=8em, text centered, rounded corners, minimum height=6em, font=\bfseries]
\tikzstyle{line} = [draw, -latex']
\tikzstyle{cloud} = [draw, ellipse,fill=white!20, node distance=4cm, minimum height=6em]

\begin{center} 
\begin{figure}[!htbp]\label{figure::transition}
\begin{tikzpicture}[node distance = 3cm, auto]
    \node [block, draw=red] (S) {Susceptible (S)};
    \node [block, draw=red, below of=S] (E) {Exposed (E)};
    \node [block,draw=red, below of=E] (D) {Detectable Infection (It)};
    \node [decision, below of=D] (Symp) {Onset (Is/Ia)};
    \node [cloud, below of=Symp] (Is) {Symptomatic};
    \node [cloud, left of=Symp] (Ia) {Asymptomatic};
    \node [decision, right of=Symp] (F) {Final Stage (R/I)};
    \node [cloud, right of=F] (R) {Recovered};
    \node [cloud, below of=F] (Iso) {Isolated};
    
    \path [line] (S) -- (E);
    \path [line] (E) -- node{4.5 days}(D);
    \path [line] (D) -- node{1 day}(Symp);
    \path [line,dashed] (Symp) -- (Is);
    \path [line,dashed] (Symp) -- (Ia);
    \path [line] (Symp) -- (F);
    \path [line,dashed] (F) -- (R);
    \path [line,dashed] (F) -- (Iso);
    
\end{tikzpicture}

\caption{{Disease progression schematics of SARS-CoV-2 in an infected individual in the agent-based model used for all simulations. Average time to transition from Exposed status (E) to Detectable Infection (It), and from It to either Asymptotic (Ia) or Symptomatic (Is) status is depicted in-between relevant states.}}
\end{figure}
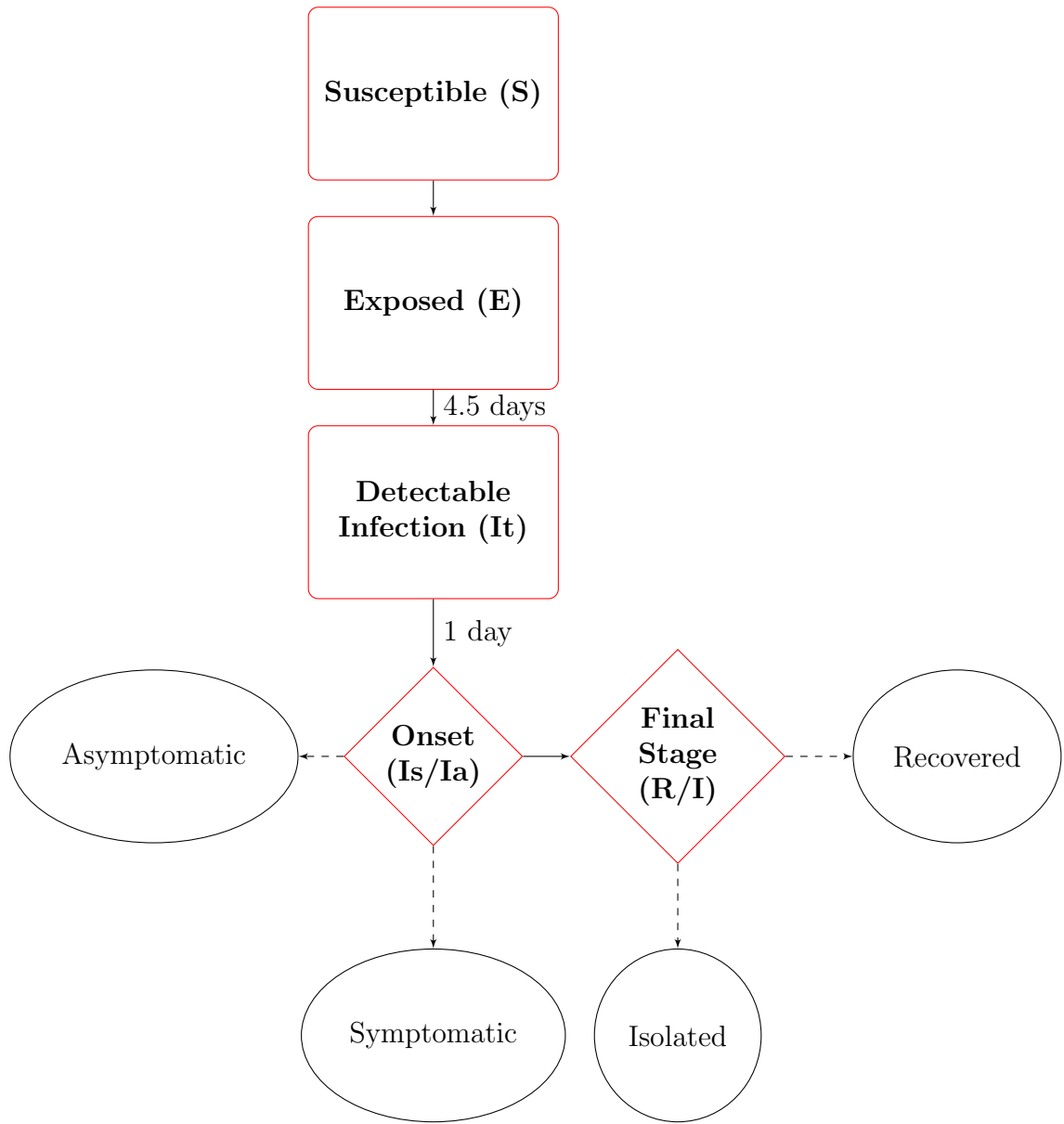  
\end{center}
\newpage

\subsection{Additional Simulations}\label{sec::add_sims}

We investigate design performance under different levels of outside transmission governed by the risk scale parameter. In our agent-based model, higher values of the risk parameter correspond to higher weighted latent parts of the network and individual risk on the transmission dynamics. Intuitively, as the risk parameter increases, it becomes more difficult to identify infectious individuals. In the following, we show the cumulative incidence curve at each time point (Figure \ref{fig::trajectory_rs}) and the final cumulative incidence by $t=120$ (Figure \ref{fig::final_rs}). Similarly to previously reported results for the risk scale parameter (value 0.5) in Section \ref{sec::sim_results}, the Online SL for adaptive surveillance with TMLE-CI-based selector outperforms all other designs over a grid of risk parameters (values considered: 0.4, 0.6, 0.7). In fact, as the problem becomes harder, TMLE-CI-based selector seems to gain even more advantage over the second best design (CI for risk parameter 0.4: (2.1\%,2.5\%) vs. (2.8\%,2.9\%); CI for risk parameter 0.5: (3.6\%,4.2\%) vs. (4.3\%,5.0\%); CI for risk parameter 0.6: (5.6\%,6.5\%) vs. (8.4\%,9.9\%); CI for risk parameter 0.7: (10.5\%,12.1\%) vs. (16.2\%,18.1\%)). In comparison to the standard testing practice, \textit{symptomatic + contact}, Online SL for adaptive surveillance with TMLE-CI-based selector vastly outperforms across all risk scale parameters (CI for risk parameter 0.4: (2.1\%,2.5\%) vs. (3.5\%,4.1\%); CI for risk parameter 0.5: (3.6\%,4.2\%) vs. (6.8\%,7.7\%); CI for risk parameter 0.6: (5.6\%,6.5\%) vs. (12.1\%,13.8\%); CI for risk parameter 0.7: (10.5\%,12.1\%) vs. (21.9\%,23.4\%)). Comparison with the \textit{risk-based} testing scheme remains parallel to \textit{symptomatic + contact}, with high advantages over \textit{random} testing as well, as shown in Figure \ref{fig::final_rs}. Finally, as expected, higher risk scale parameter values correspond to a higher average final cumulative incidence under any design. Even for the TMLE-CI-based selector, as the risk scale jumps to 0.7, average final cumulative incidence at $k=600$ is $11.3\%$ vs. $3.9\%$ at a value of 0.5. Nevertheless, no matter the level of difficulty of the problem, the ranking of the designs in terms of testing performance remains the same. 
\FloatBarrier

\vspace{5mm}
\begin{figure*}[!htbp]
\centering \makeatletter\includegraphics[width=1\linewidth]{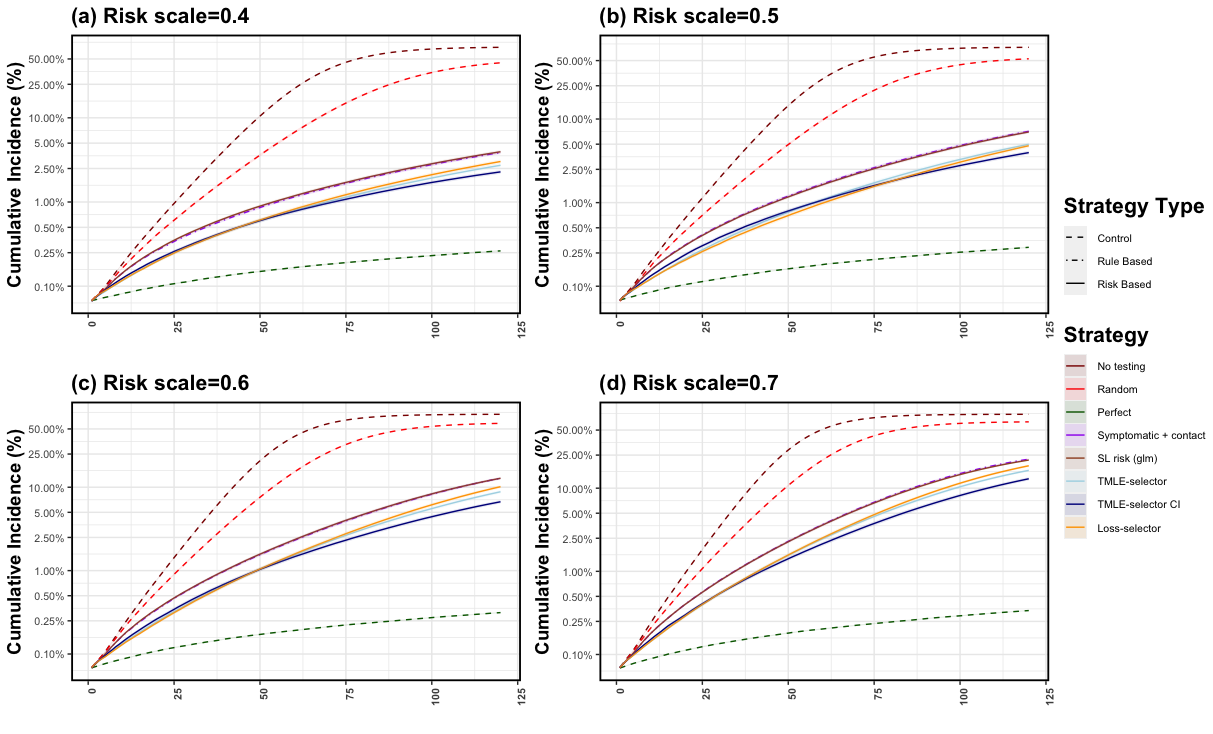}
\makeatother 
\caption{{Average cumulative incidence at each time point over $250$ simulations with $n=20,000$ sample size and testing capacity $k=600$ using TMLE-based, TMLE-CI-based and loss-based selectors of the testing strategy. Performance is evaluated at different risk scale values $\{0.4,0.5,0.6,0.7\}$, where higher risk scale score corresponds to higher individual risk. We compare different proposed selectors to \textit{symptomatic + contact}, \textit{random} and \textit{glm risk-based} testing, with \textit{perfect} as the upper and \textit{no testing} as the lower bound on performance.}}
\label{fig::trajectory_rs}
\end{figure*}

\FloatBarrier
\vspace{5mm}
\begin{figure*}[!htbp]
\centering \makeatletter\includegraphics[width=1\linewidth]{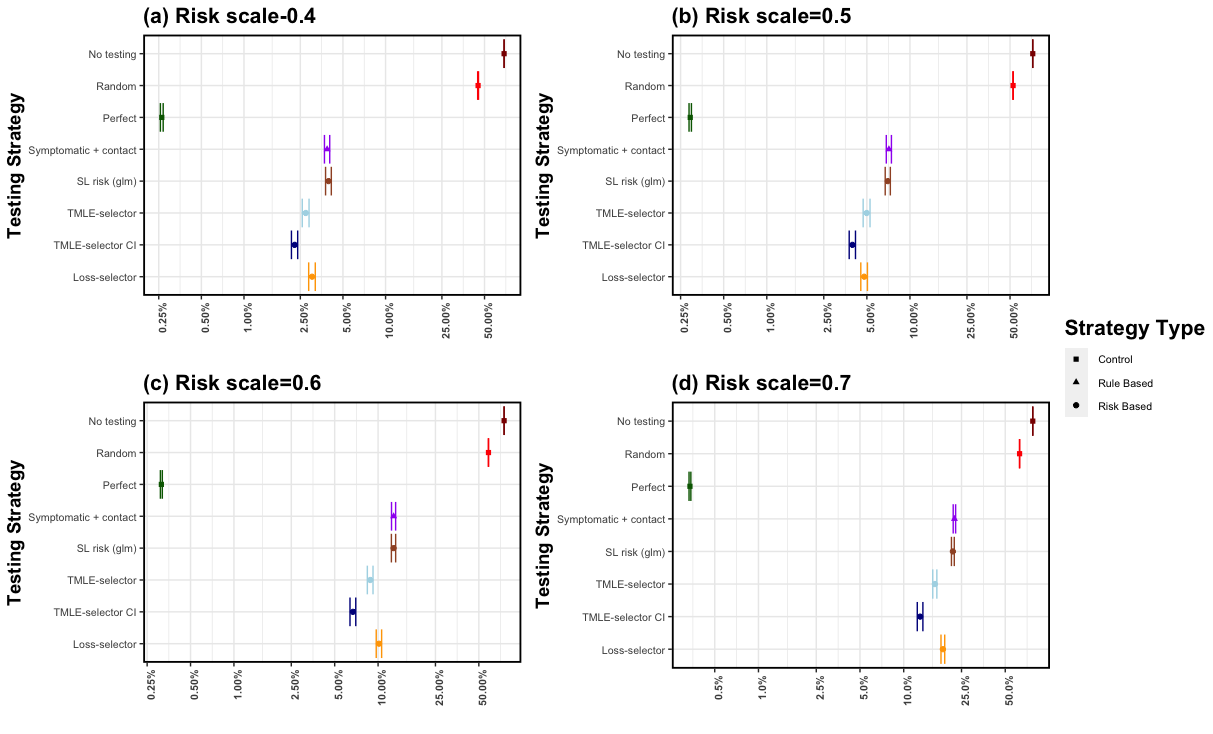}
\makeatother 
\caption{{Average final cumulative incidence at $t=120$ over $250$ simulations, with $n=20,000$ sample size and testing capacity $k=600$ using TMLE-based, TMLE-CI-based and loss-based selectors. Performance is evaluated at different risk scale values $\{0.4,0.5,0.6,0.7\}$, where higher risk scale score corresponds to higher individual risk. We compare different proposed selectors to \textit{symptomatic + contact}, \textit{random} and \textit{glm risk-based} testing, with \textit{perfect} as the upper and \textit{no testing} as the lower bound on performance.}}
\label{fig::final_rs}
\end{figure*}

\FloatBarrier
\vspace{5mm}
\begin{figure*}[!htbp]
\centering \makeatletter\includegraphics[width=1\linewidth]{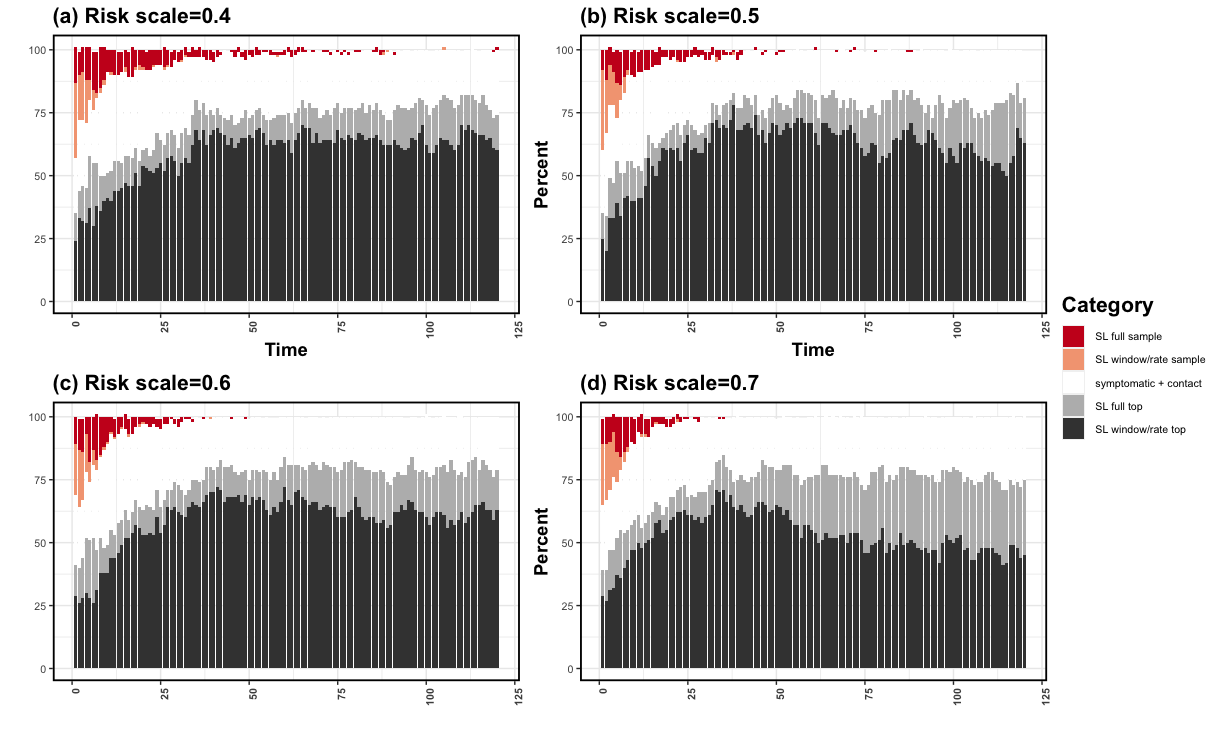}
\makeatother 
\caption{{Percent of times a given candidate testing allocation design is selected over the full trajectory and $250$ simulations with $n=20,000$ sample size and testing capacity $k=600$ using the TMLE-CI-based selector. Performance is evaluated at different risk scale values $\{0.4,0.5,0.6,0.7\}$, where higher risk scale score corresponds to higher individual risk. Candidate designs include \textit{symptomatic + contact} and various \textit{risk-based} designs where the Super Learner is either trained on the full past, exponentially weighted past, or a window of $t=14$ days. All \textit{risk-based} testing strategies also consider different sampling schemes, including picking the top samples for testing or randomly sampling based on the estimated risk.}}
\label{fig::gstars_rs}
\end{figure*}

\newpage
\bibliography{online_sl_adapt_testing}%

\begin{thebibliography}{54}
\providecommand{\natexlab}[1]{#1}
\providecommand{\url}[1]{\texttt{#1}}
\expandafter\ifx\csname urlstyle\endcsname\relax
  \providecommand{\doi}[1]{doi: #1}\else
  \providecommand{\doi}{doi: \begingroup \urlstyle{rm}\Url}\fi

\bibitem[Alagoz et~al.(2010)Alagoz, Hsu, Schaefer, and Roberts]{alagoz2010}
O.~Alagoz, H.~Hsu, A.J. Schaefer, and M.S. Roberts.
\newblock {{M}arkov decision processes: a tool for sequential decision making
  under uncertainty}.
\newblock \emph{Med Decis Making}, 30\penalty0 (4):\penalty0 474--483, 2010.

\bibitem[Bahl et~al.(2021)Bahl, Eikmeier, Fraser, Junge, Keesing, Nakahata, and
  Reeves]{Bahl2021}
R.~Bahl, N.~Eikmeier, A.~Fraser, M.~Junge, F.~Keesing, K.~Nakahata, and
  L.~Reeves.
\newblock {{M}odeling {C}{O}{V}{I}{D}-19 spread in small colleges}.
\newblock \emph{PLoS One}, 16\penalty0 (8):\penalty0 e0255654, 2021.

\bibitem[Balzer et~al.(2019)Balzer, Havlir, Kamya, Chamie, Charlebois, Clark,
  Koss, Kwarisiima, Ayieko, Sang, Kabami, Atukunda, Jain, Camlin, Cohen,
  Bukusi, van~der Laan, and Petersen]{balzer2019}
L.B Balzer, D.V Havlir, M.R Kamya, G~Chamie, E.D Charlebois, T.D Clark, C.A
  Koss, D~Kwarisiima, J~Ayieko, N~Sang, J~Kabami, M~Atukunda, V~Jain, C.S
  Camlin, C.R Cohen, E.A Bukusi, M~van~der Laan, and M.L Petersen.
\newblock {Machine Learning to Identify Persons at High-Risk of Human
  Immunodeficiency Virus Acquisition in Rural Kenya and Uganda}.
\newblock \emph{Clinical Infectious Diseases}, 71\penalty0 (9):\penalty0
  2326--2333, 11 2019.
\newblock ISSN 1058-4838.
\newblock \doi{10.1093/cid/ciz1096}.
\newblock URL \url{https://doi.org/10.1093/cid/ciz1096}.

\bibitem[Benkeser et~al.(2018)Benkeser, Ju, Lendle, and {van der
  Laan}]{benkeser2018}
D.~Benkeser, C.~Ju, S.~Lendle, and M.J. {van der Laan}.
\newblock Online cross-validation-based ensemble learning.
\newblock \emph{Statistics in Medicine}, 37\penalty0 (2):\penalty0 249--260,
  2018.
\newblock ISSN 1097-0258.
\newblock \doi{10.1002/sim.7320}.

\bibitem[Bibaut et~al.(2021)Bibaut, Petersen, Vlassis, Dimakopoulou, and {van
  der Laan}]{bibaut2021}
A.~Bibaut, M.~Petersen, N.~Vlassis, M.~Dimakopoulou, and M.J. {van der Laan}.
\newblock Sequential causal inference in a single world of connected units,
  2021.

\bibitem[Biswas et~al.(2020)Biswas, Bannur, Jain, and Merugu]{biswas2020}
A.~Biswas, S.~Bannur, P.~Jain, and S.~Merugu.
\newblock Covid-19: Strategies for allocation of test kits, 2020.

\bibitem[Bradley et~al.(2020)Bradley, An, and Fox]{Bradley2020}
E.H. Bradley, M.W. An, and E.~Fox.
\newblock {{R}eopening {C}olleges {D}uring the {C}oronavirus {D}isease 2019
  ({C}{O}{V}{I}{D}-19) {P}andemic-{O}ne {S}ize {D}oes {N}ot {F}it {A}ll}.
\newblock \emph{JAMA Netw Open}, 3\penalty0 (7):\penalty0 e2017838, 07 2020.

\bibitem[Chambaz et~al.(2017)Chambaz, Zheng, and {van der Laan}]{chambaz2017}
A.~Chambaz, W.~Zheng, and M.J. {van der Laan}.
\newblock {Targeted sequential design for targeted learning inference of the
  optimal treatment rule and its mean reward}.
\newblock \emph{Ann Stat}, 45\penalty0 (6):\penalty0 2537--2564, 2017.

\bibitem[Chang et~al.(2021)Chang, Crawford, and Kaplan]{Chang2021}
J.T. Chang, F.W. Crawford, and E.H. Kaplan.
\newblock {{R}epeat {S}{A}{R}{S}-{C}o{V}-2 testing models for residential
  college populations}.
\newblock \emph{Health Care Manag Sci}, 24\penalty0 (2):\penalty0 305--318, Jun
  2021.

\bibitem[Chatzimanolakis et~al.(2020)Chatzimanolakis, Weber, Arampatzis,
  Wälchli, Kičić, Karnakov, Papadimitriou, and
  Koumoutsakos]{chatzimanolakis2020}
M.~Chatzimanolakis, P.~Weber, G.~Arampatzis, D.~Wälchli, I.~Kičić,
  P.~Karnakov, C.~Papadimitriou, and P.~Koumoutsakos.
\newblock {{O}ptimal allocation of limited test resources for the
  quantification of {C}{O}{V}{I}{D}-19 infections}.
\newblock \emph{Swiss Med Wkly}, 150:\penalty0 w20445, 12 2020.

\bibitem[Diaz and {van der Laan}(2012)]{diaz2012}
I.~Diaz and M.J. {van der Laan}.
\newblock {{P}opulation intervention causal effects based on stochastic
  interventions}.
\newblock \emph{Biometrics}, 68\penalty0 (2):\penalty0 541--549, Jun 2012.

\bibitem[Du et~al.(2021)Du, Beesley, Lee, Zhou, Dempsey, and Mukherjee]{du2021}
J.~Du, L.J. Beesley, S.~Lee, X.~Zhou, W.~Dempsey, and B.~Mukherjee.
\newblock Optimal diagnostic test allocation strategy during the covid-19
  pandemic and beyond.
\newblock \emph{Statistics in Medicine}, n/a\penalty0 (n/a), 2021.
\newblock \doi{https://doi.org/10.1002/sim.9238}.
\newblock URL \url{https://onlinelibrary.wiley.com/doi/abs/10.1002/sim.9238}.

\bibitem[Dudoit and {van der Laan}(2005)]{dudoit2003a}
S.~Dudoit and M.J. {van der Laan}.
\newblock Asymptotics of cross-validated risk estimation in estimator selection
  and performance assessment.
\newblock \emph{Statistical Methodology}, 2\penalty0 (2):\penalty0 131 -- 154,
  2005.
\newblock ISSN 1572-3127.
\newblock \doi{https://doi.org/10.1016/j.stamet.2005.02.003}.
\newblock URL
  \url{http://www.sciencedirect.com/science/article/pii/S1572312705000158}.

\bibitem[Gardner and Kilpatrick(2021)]{gardner2021}
B.J. Gardner and A.M. Kilpatrick.
\newblock {{C}ontact tracing efficiency, transmission heterogeneity, and
  accelerating {C}{O}{V}{I}{D}-19 epidemics}.
\newblock \emph{PLoS Comput Biol}, 17\penalty0 (6):\penalty0 e1009122, 06 2021.

\bibitem[Ghaffarzadegan(2021)]{Ghaffarzadegan2021}
N.~Ghaffarzadegan.
\newblock {{S}imulation-based what-if analysis for controlling the spread of
  {C}ovid-19 in universities}.
\newblock \emph{PLoS One}, 16\penalty0 (2):\penalty0 e0246323, 2021.

\bibitem[Gonsalves et~al.(2021)Gonsalves, Copple, Paltiel, Fenichel, Bayham,
  Abraham, Kline, Malloy, Rayo, Zhang, Faulkner, Morey, Wu, Thornhill, Iloglu,
  and Warren]{gregg2021}
G.S. Gonsalves, J.T. Copple, A.D. Paltiel, E.P. Fenichel, J.~Bayham,
  M.~Abraham, D.~Kline, S.~Malloy, M.F. Rayo, N.~Zhang, D.~Faulkner, D.A.
  Morey, F.~Wu, T.~Thornhill, S.~Iloglu, and J.L. Warren.
\newblock Maximizing the efficiency of active case finding for sars-cov-2 using
  bandit algorithms.
\newblock \emph{Medical Decision Making}, 41\penalty0 (8):\penalty0 970--977,
  2021.
\newblock \doi{10.1177/0272989X211021603}.
\newblock URL \url{https://doi.org/10.1177/0272989X211021603}.
\newblock PMID: 34120510.

\bibitem[Gressman and Peck(2020)]{Gressman2020}
P.T. Gressman and J.R. Peck.
\newblock {{S}imulating {C}{O}{V}{I}{D}-19 in a university environment}.
\newblock \emph{Math Biosci}, 328:\penalty0 108436, 10 2020.

\bibitem[Guh et~al.(2011)Guh, Reed, Gould, Kutty, Iuliano, Mitchell, Dee,
  Desai, Siebold, Silverman, Massoudi, Lynch, Sotir, Armstrong, and
  Swerdlow]{guh2011}
A.~Guh, C.~Reed, L.~H. Gould, P.~Kutty, D.~Iuliano, T.~Mitchell, D.~Dee,
  M.~Desai, J.~Siebold, P.~Silverman, M.~Massoudi, M.~Lynch, M.~Sotir,
  G.~Armstrong, and D.~Swerdlow.
\newblock {{T}ransmission of 2009 pandemic influenza {A} ({H}1{N}1) at a
  {P}ublic {U}niversity--{D}elaware, {A}pril-{M}ay 2009}.
\newblock \emph{Clin Infect Dis}, 52 Suppl 1:\penalty0 S131--137, Jan 2011.

\bibitem[Hambridge et~al.(2021)Hambridge, Kahn, and Onnela]{hambridge2021}
H.L. Hambridge, R.~Kahn, and J.P. Onnela.
\newblock Examining sars-cov-2 interventions in residental colleges using an
  emprical network.
\newblock \emph{medRxiv}, Apr 2021.

\bibitem[Hamer et~al.(2021)Hamer, White, Jenkins, Gill, Landsberg, Klapperich,
  Bulekova, Platt, Decarie, Gilmore, Pilkington, MacDowell, Faria, Densmore,
  Landaverde, Li, Rose, Burgay, Miller, Doucette-Stamm, Lockard, Elmore,
  Schroeder, Zaia, Kolaczyk, Waters, and Brown]{Hamer2021}
D.H. Hamer, L.F. White, H.E. Jenkins, C.J Gill, H.E. Landsberg, C.~Klapperich,
  K.~Bulekova, J.~Platt, L.~Decarie, W.~Gilmore, M.~Pilkington, T.L. MacDowell,
  M.A. Faria, D.~Densmore, L.~Landaverde, W.~Li, T.~Rose, S.P. Burgay,
  C.~Miller, L.~Doucette-Stamm, K.~Lockard, K.~Elmore, T.~Schroeder, A.M. Zaia,
  E.D. Kolaczyk, G.~Waters, and R.A. Brown.
\newblock {{A}ssessment of a {C}{O}{V}{I}{D}-19 {C}ontrol {P}lan on an {U}rban
  {U}niversity {C}ampus {D}uring a {S}econd {W}ave of the {P}andemic}.
\newblock \emph{JAMA Netw Open}, 4\penalty0 (6):\penalty0 e2116425, 06 2021.

\bibitem[Hill et~al.(2021)Hill, Atkins, Keeling, Tildesley, and
  Dyson]{Hill2021}
E.M. Hill, B.D. Atkins, M.J. Keeling, M.J. Tildesley, and L.~Dyson.
\newblock {{M}odelling {S}{A}{R}{S}-{C}o{V}-2 transmission in a {U}{K}
  university setting}.
\newblock \emph{Epidemics}, 36:\penalty0 100476, Jun 2021.

\bibitem[Jonnerby et~al.(2020)Jonnerby, Lazos, Lock, Marmolejo-Cossío, Ramsey,
  Shukla, and Sridhar]{jonnerby2020}
J.~Jonnerby, P.~Lazos, E.~Lock, F.~Marmolejo-Cossío, C.B. Ramsey, N.~Shukla,
  and D.~Sridhar.
\newblock Maximising the benefits of an acutely limited number of covid-19
  tests, 2020.

\bibitem[Larremore et~al.(2021)Larremore, Wilder, Lester, Shehata, Burke, Hay,
  Tambe, Mina, and Parker]{larremore2021}
D.B. Larremore, B.~Wilder, E.~Lester, S.~Shehata, J.M. Burke, J.A. Hay,
  M.~Tambe, M.J. Mina, and R.~Parker.
\newblock Test sensitivity is secondary to frequency and turnaround time for
  covid-19 screening.
\newblock \emph{Science Advances}, 7\penalty0 (1):\penalty0 eabd5393, 2021.
\newblock \doi{10.1126/sciadv.abd5393}.
\newblock URL \url{https://www.science.org/doi/abs/10.1126/sciadv.abd5393}.

\bibitem[Lopman et~al.(2021)Lopman, Liu, Le~Guillou, Handel, Lash, Isakov, and
  Jenness]{Lopman2021}
B.~Lopman, C.Y. Liu, A.~Le~Guillou, A.~Handel, T.L. Lash, A.P. Isakov, and S.M.
  Jenness.
\newblock {{A} modeling study to inform screening and testing interventions for
  the control of {S}{A}{R}{S}-{C}o{V}-2 on university campuses}.
\newblock \emph{Sci Rep}, 11\penalty0 (1):\penalty0 5900, 03 2021.

\bibitem[Luedtke and {van der Laan}(2016)]{luedtke2016resource}
A.R. Luedtke and M.J. {van der Laan}.
\newblock {{O}ptimal {I}ndividualized {T}reatments in {R}esource-{L}imited
  {S}ettings}.
\newblock \emph{Int J Biostat}, 12\penalty0 (1):\penalty0 283--303, 05 2016.

\bibitem[Malenica et~al.(2021{\natexlab{a}})Malenica, Bibaut, and {van der
  Laan}]{malenica2021adaptive}
I.~Malenica, A.~Bibaut, and M.J. {van der Laan}.
\newblock Adaptive sequential design for a single time-series,
  2021{\natexlab{a}}.

\bibitem[Malenica et~al.(2021{\natexlab{b}})Malenica, Phillips, Pirracchio,
  Chambaz, Hubbard, and {van der Laan}]{malenica2021personalized}
I.~Malenica, R.V. Phillips, R.~Pirracchio, A.~Chambaz, A.~Hubbard, and M.J.
  {van der Laan}.
\newblock Personalized online machine learning, 2021{\natexlab{b}}.

\bibitem[Martin et~al.(2020)Martin, Schooley, and De~Gruttola]{Martin2020}
N.~Martin, R.T. Schooley, and V.~De~Gruttola.
\newblock {{M}odelling testing frequencies required for early detection of a
  {S}{A}{R}{S}-{C}o{V}-2 outbreak on a university campus}.
\newblock \emph{medRxiv}, Jun 2020.

\bibitem[Matheson et~al.(2021)Matheson, Warne, Weekes, and
  Maxwell]{matheson2021}
N.J. Matheson, B.~Warne, M.P. Weekes, and P.H. Maxwell.
\newblock {{M}ass testing of university students for covid-19}.
\newblock \emph{BMJ}, 375:\penalty0 n2388, 10 2021.

\bibitem[Muller and Muller(2021)]{Muller2021}
K.~Muller and P.~Muller.
\newblock {{M}athematical modelling of the spread of {C}{O}{V}{I}{D}-19 on a
  university campus}.
\newblock \emph{Infect Dis Model}, Aug 2021.

\bibitem[Ogburn et~al.(2022)Ogburn, Sofrygin, Díaz, and van~der
  Laan]{ogburn2022}
E.L. Ogburn, O.~Sofrygin, I.~Díaz, and M.J. van~der Laan.
\newblock Causal inference for social network data.
\newblock \emph{Journal of the American Statistical Association}, 0\penalty0
  (ja):\penalty0 1--46, 2022.
\newblock \doi{10.1080/01621459.2022.2131557}.
\newblock URL \url{https://doi.org/10.1080/01621459.2022.2131557}.

\bibitem[Organization(2020)]{who2020}
World~Health Organization.
\newblock Laboratory testing strategy recommendations for covid-19: interim
  guidance, 21 march 2020.
\newblock Technical documents, World Health Organization, 2020.

\bibitem[Paltiel et~al.(2020)Paltiel, Zheng, and Walensky]{Paltiel2020}
A.D. Paltiel, A.~Zheng, and R.P. Walensky.
\newblock {{C}{O}{V}{I}{D}-19 screening strategies that permit the safe
  re-opening of college campuses}.
\newblock \emph{medRxiv}, Jul 2020.

\bibitem[Pearl(2009)]{pearl2009}
J.~Pearl.
\newblock \emph{Causality: Models, Reasoning and Inference}.
\newblock Cambridge University Press, New York, NY, USA, 2nd edition, 2009.
\newblock ISBN 052189560X, 9780521895606.

\bibitem[Perkins et~al.(2020)Perkins, Cavany, Moore, Oidtman, Lerch, and
  Poterek]{Perkins2020}
T.A. Perkins, S.M. Cavany, S.M. Moore, R.J. Oidtman, A.~Lerch, and M.~Poterek.
\newblock Estimating unobserved sars-cov-2 infections in the united states.
\newblock \emph{Proceedings of the National Academy of Sciences}, 117\penalty0
  (36):\penalty0 22597--22602, 2020.
\newblock ISSN 0027-8424.
\newblock \doi{10.1073/pnas.2005476117}.
\newblock URL \url{https://www.pnas.org/content/117/36/22597}.

\bibitem[Poole et~al.(2021)Poole, Gronsbell, Winter, Nickels, Levy, Fu, Burq,
  Saeb, Edwards, Behr, Kumaresan, Macalalad, Shah, Prevost, Snoad, Brenner,
  Myers, Varghese, Califf, Washington, Lee, and Fromer]{Poole2021}
S.F. Poole, J.~Gronsbell, D.~Winter, S.~Nickels, R.~Levy, B.~Fu, M.~Burq,
  S.~Saeb, M.~D. Edwards, M.K. Behr, V.~Kumaresan, A.R. Macalalad, S.~Shah,
  M.~Prevost, N.~Snoad, M.P. Brenner, L.J. Myers, P.~Varghese, R.M. Califf,
  V.~Washington, V.S. Lee, and M.~Fromer.
\newblock {{A} holistic approach for suppression of {C}{O}{V}{I}{D}-19 spread
  in workplaces and universities}.
\newblock \emph{PLoS One}, 16\penalty0 (8):\penalty0 e0254798, 2021.

\bibitem[Rennert et~al.(2021)Rennert, McMahan, Kalbaugh, Yang, Lumsden, Dean,
  Pekarek, and Colenda]{Rennert2021}
L.~Rennert, C.~McMahan, C.A. Kalbaugh, Y.~Yang, B.~Lumsden, D.~Dean,
  L.~Pekarek, and C.C. Colenda.
\newblock {{S}urveillance-based informative testing for detection and
  containment of {S}{A}{R}{S}-{C}o{V}-2 outbreaks on a public university
  campus: an observational and modelling study}.
\newblock \emph{Lancet Child Adolesc Health}, 5\penalty0 (6):\penalty0
  428--436, 06 2021.

\bibitem[Schultes et~al.(2021)Schultes, Clarke, Paltiel, Cartter, Sosa, and
  Crawford]{schultes2021}
O.~Schultes, V.~Clarke, A.~D. Paltiel, M.~Cartter, L.~Sosa, and F.W. Crawford.
\newblock {{C}{O}{V}{I}{D}-19 {T}esting and {C}ase {R}ates and {S}ocial
  {C}ontact {A}mong {R}esidential {C}ollege {S}tudents in {C}onnecticut
  {D}uring the 2020-2021 {A}cademic {Y}ear}.
\newblock \emph{JAMA Netw Open}, 4\penalty0 (12):\penalty0 e2140602, 12 2021.

\bibitem[Stock(1989)]{stock1989nonparametric}
J.H Stock.
\newblock Nonparametric policy analysis.
\newblock \emph{Journal of the American Statistical Association}, 84\penalty0
  (406):\penalty0 567--575, 1989.

\bibitem[Tuells et~al.(2021)Tuells, Egoavil, Pena~Pardo, Montagud, Montagud,
  Caballero, Zapater, Puig-Barberá, and Hurtado-Sanchez]{Tuells2021}
J.~Tuells, C.M. Egoavil, M.A. Pena~Pardo, A.C. Montagud, E.~Montagud,
  P.~Caballero, P.~Zapater, J.~Puig-Barberá, and J.A. Hurtado-Sanchez.
\newblock {{S}eroprevalence {S}tudy and {C}ross-{S}ectional {S}urvey on
  {C}{O}{V}{I}{D}-19 for a {P}lan to {R}eopen the {U}niversity of {A}licante
  ({S}pain)}.
\newblock \emph{Int J Environ Res Public Health}, 18\penalty0 (4), 02 2021.

\bibitem[{van der Laan}(2012)]{laan2012networks}
M.J. {van der Laan}.
\newblock Causal inference for networks, 2012.
\newblock URL \url{https://biostats.bepress.com/ucbbiostat/paper300}.

\bibitem[{van der Laan} and Dudoit(2003)]{dudoit2003b}
M.J. {van der Laan} and S.~Dudoit.
\newblock Unified cross-validation methodology for selection among estimators
  and a general cross-validated adaptive epsilon-net estimator: Finite sample
  oracle inequalities and examples, 2003.
\newblock URL \url{https://biostats.bepress.com/ucbbiostat/paper130}.

\bibitem[{van der Laan} and Rose(2011)]{book2011}
M.J. {van der Laan} and S.~Rose.
\newblock \emph{Targeted Learning: Causal Inference for Observational and
  Experimental Data (Springer Series in Statistics)}.
\newblock Springer, 2011.

\bibitem[{van der Laan} and Rose(2018)]{book2018}
M.J. {van der Laan} and S.~Rose.
\newblock \emph{Targeted Learning in Data Science: Causal Inference for Complex
  Longitudinal Studies}.
\newblock Springer Science and Business Media, 2018.

\bibitem[{van der Laan} and Rubin(2006)]{Rubin_2006}
M.J. {van der Laan} and D.~Rubin.
\newblock Targeted maximum likelihood learning.
\newblock \emph{The International Journal of Biostatistics}, 2\penalty0 (1),
  2006.
\newblock \doi{doi:10.2202/1557-4679.1043}.
\newblock URL \url{https://doi.org/10.2202/1557-4679.1043}.

\bibitem[{van der Laan} et~al.(2006){van der Laan}, Dudoit, and {van der
  Vaart}]{laan2006oracle}
M.J. {van der Laan}, D.~Dudoit, and A.W. {van der Vaart}.
\newblock {The cross-validated adaptive epsilon-net estimator}.
\newblock \emph{Statistics \& Risk Modeling}, 24\penalty0 (3):\penalty0 1--23,
  December 2006.
\newblock URL \url{https://ideas.repec.org/a/bpj/strimo/v24y2006i3p23n4.html}.

\bibitem[{van der Laan} et~al.(2007){van der Laan}, Polley, and
  Hubbard]{sl2007}
M.J. {van der Laan}, E.C. Polley, and A.E. Hubbard.
\newblock Super learner.
\newblock Technical Report Working Paper 222., U.C. Berkeley Division of
  Biostatistics Working Paper Series, 07 2007.

\bibitem[{van der Laan} et~al.(2018){van der Laan}, A., and
  S.]{vanderLaan2018onlinets}
M.J. {van der Laan}, Chambaz A., and Lendle S.
\newblock \emph{Online Targeted Learning for Time Series}, pages 317--346.
\newblock Springer International Publishing, Cham, 2018.
\newblock ISBN 978-3-319-65304-4.

\bibitem[{van der Vaart} et~al.(2006){van der Vaart}, Dudoit, and {van der
  Laan}]{vaart2006}
A.W. {van der Vaart}, D.~Dudoit, and M.J. {van der Laan}.
\newblock {Oracle inequalities for multi-fold cross validation}.
\newblock \emph{Statistics \& Risk Modeling}, 24\penalty0 (3):\penalty0 1--21,
  December 2006.
\newblock URL \url{https://ideas.repec.org/a/bpj/strimo/v24y2006i3p21n3.html}.

\bibitem[van Handel(2011)]{handel2009}
R.~van Handel.
\newblock {On the minimal penalty for Markov order estimation}.
\newblock \emph{Probab. Theory Relat}, page 709–738, 2011.

\bibitem[Vander~Schaaf et~al.(2021)Vander~Schaaf, Fund, Munnich, Zastrow, Fund,
  Senti, Lynn, Kane, Love, Long, Troendle, and Sharda]{vanderscaaf2021}
N.A. Vander~Schaaf, A.J. Fund, B.V. Munnich, A.L. Zastrow, E.E. Fund, T.L.
  Senti, A.F. Lynn, J.J. Kane, J.L. Love, G.J. Long, N.J. Troendle, and D.R.
  Sharda.
\newblock {{R}outine, {C}ost-{E}ffective {S}{A}{R}{S}-{C}o{V}-2 {S}urveillance
  {T}esting {U}sing {P}ooled {S}aliva {L}imits {V}iral {S}pread on a
  {R}esidential {C}ollege {C}ampus}.
\newblock \emph{Microbiol Spectr}, 9\penalty0 (2):\penalty0 e0108921, 10 2021.

\bibitem[Walke et~al.(2020)Walke, Honein, and Redfield]{walke2020}
H.T. Walke, M.A. Honein, and R.R. Redfield.
\newblock {Preventing and Responding to COVID-19 on College Campuses}.
\newblock \emph{JAMA}, 324\penalty0 (17):\penalty0 1727--1728, 11 2020.
\newblock ISSN 0098-7484.
\newblock \doi{10.1001/jama.2020.20027}.
\newblock URL \url{https://doi.org/10.1001/jama.2020.20027}.

\bibitem[Weeden and Cornwell(2020)]{Weeden2020}
K.A. Weeden and B.~Cornwell.
\newblock The small-world network of college classes: Implications for epidemic
  spread on a university campus.
\newblock \emph{Sociological Science}, 7\penalty0 (9):\penalty0 222--241, 2020.
\newblock ISSN 2330-6696.
\newblock \doi{10.15195/v7.a9}.
\newblock URL \url{http://dx.doi.org/10.15195/v7.a9}.

\bibitem[Wu et~al.(2021)Wu, Liang, Chen, Wang, Fang, Shen, Yang, Wang, Chen,
  Chen, Wu, Liu, Yang, Li, Zhu, Zhou, Wang, Li, Lu, Liu, Li, Krawczyk, Lu,
  Yang, Deng, Dittmer, Trilling, and Zheng]{wu2021}
J.~Wu, B.~Liang, C.~Chen, H.~Wang, Y.~Fang, S.~Shen, X.~Yang, B.~Wang, L.~Chen,
  Q.~Chen, Y.~Wu, J.~Liu, X.~Yang, W.~Li, B.~Zhu, W.~Zhou, H.~Wang, S.~Li,
  S.~Lu, D.~Liu, H.~Li, A.~Krawczyk, M.~Lu, D.~Yang, F.~Deng, U.~Dittmer,
  M.~Trilling, and X.~Zheng.
\newblock {{S}{A}{R}{S}-{C}o{V}-2 infection induces sustained humoral immune
  responses in convalescent patients following symptomatic {C}{O}{V}{I}{D}-19}.
\newblock \emph{Nat Commun}, 12\penalty0 (1):\penalty0 1813, 03 2021.

\end{thebibliography}
 
\end{document}